%% file: ctlb.tex
\documentclass[11pt]{article}

\usepackage{sn-preamble} 
\usepackage{tikz}
\usepackage{verbatim}
\usepackage{cancel}
\usepackage{titling}
\usepackage{thm-restate}
\usepackage{microtype}

\newcommand{\rnote}[1]{\footnote{{\bf \color{red}Rocco}: {#1}}}

\newcommand{\blue}[1]{{{\color{blue}#1}}}
\newcommand{\red}[1]{{\color{red} {#1}}}

\newcommand{\ignore}[1]{}

\newcommand{\Ball}{\mathrm{Ball}}
\newcommand{\Thinshell}{{\mathrm{ThinShell}}}

\newcommand{\Leb}{\mathrm{Leb}}
\newcommand{\Naz}{{\mathrm{Naz}}}

\renewcommand{\S}{\mathbb{S}}

\newcommand{\Dno}{{\cal D}_{\mathrm{no}}}
\newcommand{\Dyes}{{\cal D}_{\mathrm{yes}}}
\newcommand{\Shell}{{\mathrm{Shell}}}

\newcommand{\Avg}{\mathrm{Avg}}
\newcommand{\dtv}{d_{\mathrm{TV}}}
\newcommand{\duo}{d_{\mathrm{UO}}}
\def\ff{{\alpha}}
\newcommand{\clip}{\mathrm{clip}}
\newcommand{\Ryes}{\bR_{\mathrm{yes}}}
\newcommand{\Rno}{\bR_{\mathrm{no}}}

\newcommand{\curb}{\mathrm{Curb}}
\newcommand{\road}{\mathrm{Middle}}

\newcommand{\fno}{f_{\mathrm{no}}}
\newcommand{\boldfyes}{\boldf_{\mathrm{yes}}}
\newcommand{\boldfno}{\boldf_{\mathrm{no}}}
\newcommand{\bfyes}{\boldf_{\mathrm{yes}}}
\newcommand{\bfno}{\boldf_{\mathrm{no}}}
\newcommand{\Pconv}{\calP_{\mathrm{conv}}}
\newcommand{\Top}{{\mathrm{Left}}}
\newcommand{\Bottom}{{\mathrm{Right}}}
\newcommand{\Middle}{{\mathrm{Middle}}}
\newcommand{\Bad}{\mathsf{Bad}}
\newcommand{\EVol}[1]{\Ex\sbra{\Vol\pbra{#1}}}
\newcommand{\slab}{\mathrm{hfsp}}
\newcommand{\ring}{\slab}

\title{Lower Bounds for Convexity Testing}

\author{
Xi Chen\thanks{Columbia University. Email: \url{xc2198@columbia.edu}.} \and 
Anindya De\thanks{University of Pennsylvania. Email: \url{anindyad@seas.upenn.edu}.} \and 
Shivam Nadimpalli\thanks{MIT. Email: \url{shivamn@mit.edu}.} \and 
Rocco A. Servedio\thanks{Columbia University. Email: \url{rocco@cs.columbia.edu}.} \and 
Erik Waingarten\thanks{University of Pennsylvania. Email: \url{ewaingar@seas.upenn.edu}.}
}

\date{\today}

\begin{document}

\pagenumbering{gobble}

\maketitle

%
%
%
%
%

\begin{abstract}

We consider the problem of testing whether an unknown and arbitrary set $S \subseteq \R^n$ (given as a black-box membership oracle) is convex, versus $\varepsilon$-far from every convex set, under the standard Gaussian distribution.  The current state-of-the-art testing algorithms for this problem make $2^{\tilde{O}(\sqrt{n})\cdot \mathrm{poly}(1/\varepsilon)}$ non-adaptive queries, both for the standard testing problem and for tolerant testing.

We give the first lower bounds for convexity testing in the black-box query model:
\begin{itemize}
	\item We show that any one-sided tester (which may be adaptive) must use at least $n^{\Omega(1)}$ queries in order to test to some constant accuracy $\varepsilon>0$.
	\item We show that any non-adaptive \emph{tolerant} tester (which may make two-sided errors) must use at least $2^{\Omega(n^{1/4})}$ queries to distinguish sets that are $\varepsilon_1$-close to convex versus $\varepsilon_2$-far from convex, for some absolute constants $0<\varepsilon_1<\varepsilon_2$. 
\end{itemize}
Finally, we also show that for any constant $c>0$, any non-adaptive tester (which may make two-sided errors) must use at least $n^{1/4 - c}$ queries in order to test to some constant accuracy $\varepsilon>0$.
\end{abstract}

\newpage

\setcounter{page}{1}

\pagenumbering{arabic}

\input{sections/intro} 
\input{sections/prelims}

\input{sections/nazarov}

\input{sections/one-sided-adaptive-2}
\input{sections/tolerant}

\input{sections/two-sided-nonadaptive}

\section*{Acknowledgements}

X.C.~is supported by NSF grants IIS-1838154, CCF-2106429, and CCF-2107187. A.D.~is supported by NSF grants CCF-1910534 and CCF0-2045128. S.N.~is supported by NSF grants CCF-2106429, CCF-2211238, CCF-1763970, and CCF-2107187. R.A.S.~is supported by NSF grants CCF-2106429 and CCF-2211238. E.W.~is supported by~NSF grant CCF-2337993.

This work was partially completed while a subset of the authors was visiting the Simons Institute for the Theory of Computing. 

\bibliography{allrefs}
\bibliographystyle{alphaurl}

\appendix

\end{document}

%% file: sections/intro.tex

\newpage

\section{Introduction}
\label{sec:intro}

High-dimensional convex geometry 
is a rich topic at the intersection of geometry, probability, and analysis (see \cite{ball1997elementary,GruberWills93,LeonardLewis15,Tropp18,TkoczNotes,HugWeil20}, among many other works, for general overviews).  Apart from its intrinsic interest, a strong motivation for the study of high-dimensional convex sets from the perspective of theoretical computer science is that convexity often translates into a form of mathematical niceness which facilitates efficient computation, as witnessed by the plethora of positive results in algorithms and optimization for convex functions and convex sets. In this work, the object of study is the convex set:
\begin{flushleft}\begin{quote}
A set $K \subset \R^n$ is convex if and only if for every two points $x, y \in \R^n$, any point $z$ on the segment between $x$ and $y$ lies in $K$ whenever $x$ and $y$ lie in $K$. 
\end{quote}\end{flushleft}
The above gives a ``local'' characterization of convex sets, where ``local'' refers to the fact that (aside from quantifying over \emph{all} co-linear points $x,z,y$,) an algorithm may make three membership queries to check the condition --- in particular, non-convexity can be verified with three queries. Can one relax the ``for all'' quantification to give 
a local condition which characterizes \emph{approximately} convex sets?
Is there an algorithm which, by making very few queries, can determine whether or not a set is (almost) convex? 


A natural vantage point for this broad question is that of \emph{property testing}~\cite{BLR93, RS96}, which provides an algorithmic framework for studying the above questions. 
In our setting, we consider property testing of convex sets with respect to the \emph{standard Gaussian distribution}, arguably the most natural distribution over $\R^n$. 
Indeed, various learning, property testing, and other algorithmic problems in the Gaussian setting have been intensively studied in theoretical computer science~\cite{KOS:08,Vempalafocs10,MORS:10,Kane:Gotsman11,Kane12,Kane14,KK14,KNOW14,Kane15,CFSS17,CDS19,DMN19,OSTK21,DMN21,HSSV22,DNS23-convex}. 
Furthermore, while a large body of mathematical work (e.g. \cite{Borell:75,Ball:93,Latala1999,latala2003some,Nazarov:03,Borell2003,cordero2004b,Latala2005,Borell2008,roy14}) investigates the geometry of high-dimensional convex sets against the Gaussian distribution, convexity over Gaussian space arises naturally within theoretical computer science in the context of algorithmic discrepancy theory~\cite{gluskin1989extremal,bansal2010constructive,lovett2015constructive,rothvoss2017constructive,levy2017deterministic,eldan-chvatal-corr,reis2023vector} and lattice problems~\cite{reis2023subspace,rothvoss2023lattices,regev2024reverse}.

We consider the following algorithmic task:
A (randomized) testing algorithm has black-box query access to an unknown and arbitrary function $f \colon \R^n \to \zo$ (the indicator function of a subset of $\R^n$), and its goal is to make as few membership queries on $f$ as possible while 
deciding whether $f$ is convex or $\eps$-far from convex (meaning $f$ and any indicator of a convex set $g:\R^n\to\zo$ disagree on $\bx \sim N(0, I_n)$ with probability at least $\eps$). Thus, a testing algorithm gives an efficiently-checkable (randomized) condition which all convex sets satisfy, and furthermore, any set which satisfies this condition is ``almost'' convex (with respect to the standard Gaussian distribution). For example, the definition of a convex set naturally leads to the following property testing question, whose positive resolution would directly give a ``constant-query'' testing algorithm (i.e.~an algorithm whose query complexity depends only on $\eps$ and not on the ambient dimension $n$):
\begin{flushleft}\begin{quote}
Does there exist a probability distribution over co-linear points $\bx, \bz, \by$ in $\R^n$ such that the condition
$\Pr[\bz \in K \hspace{0.06cm} | \hspace{0.06cm} \bx,\by \in K] \geq 1-\delta(\eps)$  implies that the set $K$ must be $\eps$-close to convex with respect to the standard Gaussian?\footnote{Such a distribution would immediately yield a \emph{proximity-oblivious} testing algorithm~\cite{GoldreichRon11}, one of the strongest forms of property testing. Prior to this work, the existence of a proximity-oblivious tester for convexity was entirely possible.}
\end{quote}\end{flushleft}
In this work, we show the first non-trivial lower bounds for testing convexity under the standard Gaussian distribution. Our lower bounds not only give a negative resolution to the above question, they imply that, in a variety of property testing models (non-adaptive, adaptive, one-sided, two-sided, and tolerant), a dependence on the ambient dimension $n$ is always necessary. Prior to this work, an $O(1/\eps)$-query test was entirely possible for all of those models.\footnote{An $\Omega(1/\eps)$-query lower bound is easily seen to hold for essentially every non-trivial property, since this many queries are required to distinguish between the empty set (which is convex) and a random set of volume $2\eps$ (which is far from convex and far from having most properties of interest).}

As further discussed in \Cref{sec:related-work}, a number of prior works have studied convexity testing in a range of different settings, yet large gaps remain in our understanding. Most closely related are the works of~\cite{KOS:08}, who study learning convex sets over $N(0, I_n)$, and~\cite{CFSS17}, who study testing convexity over $N(0, I_n)$ when restricted to sample-based testers (i.e.~the algorithm~can~only query a given number of random points independently drawn from $N(0, I_n)$). On the upper bound side, the best algorithm for convexity testing \cite{CFSS17} is based on~\cite{KOS:08} and queries $\smash{2^{\tilde{O}(\sqrt{n})/\eps^2}}$ randomly sampled points from $N(0, I_n)$. Hence, this ``sample-based'' tester gives a non-adaptive property testing algorithm.\footnote{Recall that a \emph{non-adaptive} testing algorithm is one in which the choice of its $i$-th query point does not depend on the responses received to queries $1,\dots,i-1$.} 
Turning to lower bounds,~\cite{CFSS17} showed that, when restricted to sample-based testers, (i) algorithms which incur one-sided error must make $2^{\Omega(n)}$ queries,\footnote{Recall that a \emph{one-sided} tester for a class of functions is one which must accept (with probability 1) any function $f$ that belongs to the class. This is in contrast to making \emph{two-sided} error, where an algorithm may reject a function in the class with small probability.} and (ii) algorithms which incur two-sided error must make $2^{\Omega(\sqrt{n})}$ queries. Importantly, lower bounds on sample-based testers do not imply any lower bounds for algorithms which are allowed to make unrestricted queries. There are many prominent property testing problems (e.g., linearity and monotonicity) where the complexity of sample-based testing is significantly higher than the complexity in the (standard) query-based model.\footnote{For example, linearity testing over $\zo^n$ admits $O(1/\eps)$-query algorithms~\cite{BLR93}, but requires $\Omega(n)$ queries for sample-based testers~\cite{GoldreichRon16}. Monotonicity testing over $\zo^n$ admits $\poly(n)$-query algorithms~\cite{GGLRS,CS13a,CST14,KMS18}, but requires $\Omega(2^{n/2})$ for sample-based testers~\cite{GGLRS}.} 

\subsection{Our Results and Discussion}

This work gives the first non-trivial lower bounds for query-based convexity testing.  We prove three different lower bounds for three variants of the property testing model, which we now describe. As mentioned, the best known algorithm for convexity testing is the non-adaptive algorithm of \cite{CFSS17}, which makes $2^{\tilde{O}(\sqrt{n})/\eps^2}$ non-adaptive queries (and makes two-sided error).  

Our first~result gives a polynomial lower bound for one-sided adaptive testers:

\begin{restatable}[One-sided adaptive lower bound]{theorem}{theomaintheoremtwo}
\label{thm:one-sided}
For some absolute constant $\eps>0$, any one-sided (potentially adaptive) $\eps$-tester for convexity over $N(0,I_n)$ must use $n^{\Omega(1)}$ queries.
\end{restatable}

We also consider a challenging and well-studied extension of the standard testing model which is known as \emph{tolerant testing}~\cite{PRR06}.  
Recall that an $(\eps_1,\eps_2)$-tolerant tester for a class of functions is a testing algorithm which must accept with high probability if the input is $\eps_1$-close to some function in the class and reject with high probability if the input is $\eps_2$-far from every function in the class; thus the standard property testing model corresponds to $(0,\eps)$-tolerant testing.

The sample-based algorithm for convexity testing that is given in \cite{CFSS17} is based on agnostic learning results from \cite{KOS:08}. It follows easily from the analysis in \cite{CFSS17} and results of \cite{KOS:08} that for any $0 \leq \eps_1 < \eps_2$ with $\eps_2-\eps_1=\eps$, the \cite{CFSS17} approach gives a $2^{\tilde{O}(\sqrt{n})/\eps^4}$-query sample-based algorithm for $(\eps_1,\eps_2)$-tolerant testing of convexity.  As our final result, we give a mildly exponential lower bound on the query complexity of two-sided non-adaptive tolerant convexity testing:

\begin{restatable}[Two-sided non-adaptive tolerant testing lower bound]{theorem}{tolerantthm} \label{thm:tolerant}
There exist absolute constants $0 < \eps_1 < \eps_2 < 0.5$ such that any non-adaptive  $(\eps_1,\eps_2)$-tolerant tester for convexity over $N(0,I_n)$ (which may make two-sided errors) must use at least $2^{\Omega(n^{1/4})}$ queries.
\end{restatable}

Returning to the standard testing model, our final result gives a polynomial lower bound for two-sided non-adaptive testers:

\begin{restatable}[Two-sided non-adaptive lower bound]{theorem}{maintheoremone}
\label{thm:two-sided}
For any 
constant $c>0$, there is a constant $\eps=\eps_c > 0$ such that any non-adaptive $\eps$-tester for convexity over $N(0,I_n)$ (which may make two-sided errors) must use at least $n^{1/4 - c}$ queries.
\end{restatable}

Since $q$-query non-adaptive lower bounds imply $(\log q)$-query adaptive lower bounds,  Theorem~\ref{thm:two-sided} implies an $\Omega(\log n)$ two-sided adaptive convexity testing lower bound. (This is in contrast to the $n^{\Omega(1)}$-query lower bound against one-sided adaptive testers given by~Theorem~\ref{thm:one-sided}.)

\subsection{Techniques}

Our lower bounds rely on a wide range of techniques and constructions, and draw inspiration from prior work on \emph{monotonicity testing} of Boolean functions $f\isazofunc$~\cite{CST14,BB16,CWX17stoc,PallavoorRW22,chen2024mildly}.\footnote{Recall that a Boolean function $f\isazofunc$ is \emph{monotone} if whenever $x, y\in\zo^n$ satisfy $x_i \leq y_i$ for $i \in [n]$, we have $f(x) \leq f(y)$.} Indeed, a conceptual contribution of our work is to highlight a (perhaps unexpected) connection between ideas in monotonicity testing and convexity testing.  
Our work thus adds to and strengthens a recently emerging analogy between monotone Boolean functions and high-dimensional convex sets~\cite{DNS21itcs,DNS22,DNS23-polytope}.
Establishing this connection requires a number of technical and conceptual innovations for each of our main results; we highlight some of the key ideas below.

\subsubsection{The Nazarov Body}
A central role in our lower bounds in Theorem~\ref{thm:one-sided} and Theorem~\ref{thm:tolerant} is played by the so-called ``Nazarov body"~\cite{Nazarov:03,KOS:08,CFSS17}. This is a randomized construction of a convex set $\bB$ which is a slight variation of a construction originally due to Nazarov~\cite{Nazarov:03}, which is essentially as follows:  we choose $N \approx 2^{\sqrt{n}}$ halfspaces $\bH_1, \ldots, \bH_N$ in the space $\mathbb{R}^n$, where each halfspace $\bH_i$ is a {\em random halfspace} at a distance roughly $n^{1/4}$ from the origin.  In more detail, each halfspace is $\bH_i (x):= \Indicator\cbra{\bg^i \cdot x \geq r}$ where $r \approx n^{3/4}$ and $\bg^i$ is drawn from $N(0,I_n)$.   The convex set $\bB$ is obtained by taking the intersection of all $N$ halfspaces with $\Ball(\sqrt{n})$, the origin-centered ball of radius $\sqrt{n}.$\footnote{We remark that the original construction of \cite{Nazarov:03} differs from our construction in a number of ways: the distribution over random halfspaces is slightly different, and the body is not intersected with $\Ball(\sqrt{n}).$  For technical reasons, our specific construction facilitates our lower bound arguments.}
The exact parameters $r$ and $N$ are set carefully so that with high probability the ``Gaussian volume'' of $\bB$, i.e.~$\Pr_{\bg \sim N(0,I_n)} [\bg \in \bB],$ is a constant bounded away from $0$ and $1$. 

Note that for the Nazarov body $\bB$ and any point $x \in \Ball(\sqrt{n}) \setminus \bB$, there is a non-empty subset $J_x \subseteq [N]$ such that $j \in J_x$ iff $\bH_j(x)=0$, i.e., the point $x$ {\em violates} the halfspace $\bH_j$ for all $j \in J_x$. Now, define a point $x \in \Ball(\sqrt{n}) \setminus \bB$ to lie in the set $\bU_i$ if the set $J_x =\{i\}$, so $x \in \Ball(\sqrt{n})$ lies in $\bU_i$ if $\bH_i$ is the unique halfspace violated by $x$. 
The set $\bU := \cup_{i \in [N]} \bU_i$ is thus the set of ``uniquely violated points'' in $\Ball(\sqrt{n}).$
A crucial feature of the Nazarov construction is that the Gaussian volume of the set of points which are uniquely violated, i.e., Gaussian volume of the set $\bU$, is ``large'' compared to the 
Gaussian volume of the set $\Ball(\sqrt{n}) \setminus \bB$
(see \Cref{lemma:flaps-vs-dog-ears} for the precise statement).  


The original construction of Nazarov may be viewed as a Gaussian-space analogue of \emph{Talagrand's random CNF formula} \cite{Talagrand:96} (see \cite{DNS23-polytope} for a discussion of this connection).  Talagrand's random CNF has been very useful in lower bounds for monotonicity testing over the Boolean hypercube, as demonstrated by \cite{BB16,CWX17stoc,chen2024mildly}. We use our modified Nazarov body to obtain new lower bounds for convexity testing, as described below.

\subsubsection{One-Sided Adaptive Lower Bound} Recall that a \emph{one-sided} tester always outputs ``accept'' on convex sets and outputs ``reject'' on far-from-convex sets with probability at least $2/3$ --- this requirement implies that the tester rejects only if a certificate of non-convexity is found (i.e.~a set of queries $x_1,\dots, x_t$ which lie in the body, and a query $y$ in the convex hull of $x_1,\dots, x_t$ which is not in the body). In order to argue a $q$-query lower bound, it suffices to (1) design a distribution $\Dno$ over sets which are far-from-convex with high probability, and (2) argue that no $q$-query deterministic algorithm can find a certificate of non-convexity.

The key will be to ``hide'' the non-convexity within the uniquely violated sets of the Nazarov body. Consider working in $\R^{2n}$ and first randomly draw an $n$-dimensional ``control subspace'' $\bC$ and the orthogonal $n$-dimensional ``action subspace'' $\bA$; we embed the $n$-dimensional Nazarov body in the control subspace $\bC$. A point $x \in \R^{2n}$ lies in our (non-convex) body iff:
\begin{flushleft}\begin{itemize}
\item It has Euclidean norm at most $\sqrt{2n}$, and in addition, $x_{\bC}$ (the projection onto the control subspace) has norm at most $\sqrt{n}$; and
\item The point $x_{\bC}$ lies within an $n$-dimensional Nazarov body that we randomly sample within the control subspace $\bC$; \emph{or}, for every $j \in [N]$ where $H_j(x_{\bC}) = 0$, the projection $x_{\bA}$ on the action subspace lies outside of a ``strip'' of width $1$ along a randomly sampled direction $\bv^j$ in the action subspace. (See~\Cref{subsec:adapt-dno}).
\end{itemize}\end{flushleft}
Consider a line through a point $x \in \R^{2n}$ in direction $\bv^j$, for $j \in [N]$ such that $x_{\bC}$ lies in the uniquely violated region $\bU_j$ and $x_{\bA}$ lies inside the strip along $\bv^j$ (and therefore \emph{outside} our body). Then, as the line proceeds out from $x$ in directions $\bv^j$ and $-\bv^j$, it remains in the uniquely violated region $\bU_j$ (since $\bv_j$ is orthogonal to $\bC$) but exits the strip, thereby entering the body and exhibiting non-convexity. By design, the uniquely violated regions and the strips are large enough to constitute a constant fraction of the space, giving the desired distance to convexity (\Cref{lem:distance}). Intuitively, detecting non-convexity is hard because the algorithm does not know $\bC$, the halfspace $\bH_j(\cdot)$, or the direction $\bv^j$. In fact, we show that an algorithm which makes few queries cannot find, with probability at least $2/3$, two points $x, z$ outside the same halfspace $\bH_j(\cdot)$ such that $x$ lies inside and $z$ outside the strip in direction $\bv^j$.

Roughly speaking, the proof proceeds as follows. First, we show that, except with $o(1)$ probability, any two queries $x, z$ which are far (at distance at least $1000\sqrt{q} n^{1/4}$) cannot lie outside the same halfspace $\bH_j(\cdot)$ while having projections onto $\bC$ with norm at most $\sqrt{n}$ (\Cref{lem:ev1}), and moreover it is extremely unlikely for a query to be falsified by more than $q$ halfspaces (this follows from a calculation in~\Cref{lemma:small-vol-high-degree}). The argument is geometric in nature and is given in~\Cref{sec:proofev1}, and essentially argues that it is unlikely, since the algorithm does not know the control subspace $\bC$ or the vector defining the halfspace, that two far-away queries happen to uniquely falsify the same halfspace. 

On the other hand, consider the halfspaces which are falsified by some query (and notice there are most $q^2$ such halfspaces, since each query is falsified by at most $q$ halfspaces). Since all such queries are within distance $1000\sqrt{q}n^{1/4}$ of each other, the projection of any two such queries onto the direction defining the strip is a segment of length $O(\sqrt{q}/n^{1/4})$ with high probability, and the precise location of this segment is uniform(-like, from Gaussian anti-concentration) (see~\Cref{section:final}). Therefore, the probability of any particular segment of length $O(\sqrt{q} / n^{1/4})$ which goes from inside to outside the  strip of width $1$ is roughly $O(\sqrt{q} / n^{1/4})$. We take a union bound over the $q^2$ possible halfspaces, each containing at most $q$ queries which define segments which may ``cross'' the strip with probability $O(\sqrt{q} / n^{1/4})$, for a total probability of $O(q^{3.5} / n^{1/4})$. Since this must be at least $2/3$ for the algorithm to succeed, this gives the $n^{\Omega(1)}$ lower bound.

\medskip

\subsubsection{Two-Sided Non-Adaptive Tolerant Lower Bound}

Continuing the analogy with monotonicity testing lower bounds, the proof of Theorem~\ref{thm:tolerant} is inspired by recent lower bounds on tolerant monotonicity testing, namely \cite{PallavoorRW22} and the follow-up work of~\cite{chen2024mildly}. The basic idea of \cite{PallavoorRW22}
 is to construct a family of functions by randomly partitioning  the space of variables into {\em control variables} and {\em action variables}:  if the control variables are not balanced, i.e.~there are more 1s than 0s (or vice-versa), then the function is trivially set to $1$ (resp. to $0$) both for $\boldf \sim \Dyes$ and for $\boldf \sim \Dno$. If the control variables are balanced, 
 then, at a high level, \begin{enumerate}
 \item for $\boldf \sim \Dyes$ the function on the action variables is close to monotone; 
 \item for $\boldf \sim \Dno$ the function on the action variables is far from monotone. 
 \end{enumerate} 
 Roughly speaking, the analysis in \cite{PallavoorRW22} shows that unless the algorithm queries two points such that  both these points (a) have the same setting of the control variables, and (b) the control variables are balanced, the algorithm cannot distinguish between $\boldf \sim \Dyes$ and $\boldf \sim \Dno$. As the control and action variables are partitioned at random, it turns out that satisfying both (a) and (b) is not possible for a non-adaptive algorithm unless the algorithm makes $2^{\Omega(\sqrt{n})}$ many queries. In particular, \cite{PallavoorRW22} shows that distinguishing between functions which are $c_1/\sqrt{n}$-close to monotone versus $c_2/\sqrt{n}$-far from monotone (where $c_2>c_1>0$) cannot be done with  $2^{o(\sqrt{n})}$ queries. 
 
  The main modification in \cite{chen2024mildly} vis-a-vis \cite{PallavoorRW22} is the following: one can think of the balanced setting of the control variables in the construction described above as the ``minimal satisfying assignments" of the Majority function. In \cite{chen2024mildly}, the Majority function is replaced by  Talagrand's random monotone DNF~\cite{Talagrand:96}, a well-studied function in Boolean function analysis and related areas~\cite{mossel2003noise,o2007approximation}. 
  The specific properties of Talagrand's monotone DNF allows \cite{chen2024mildly} to obtain a $2^{n^{1/4}}$ query lower bound for non-adaptive testers where the functions in $\Dyes$ are $c_1$-close to monotone and functions in $\Dno$ are $c_2$-far from monotone, where $c_2>c_1$ are positive constants.
  
For Theorem~\ref{thm:tolerant}, the goal is to obtain lower bounds for tolerant convexity testing rather than monotonicity testing. Towards that goal, let us assume that the ambient space is $\mathbb{R}^{n+1}$. We choose a random $n$-dimensional subspace $\bC$ and think of it as the {\em control subspace}, and we view its one-dimensional orthogonal complement as the {\em action subspace} $\bA$ (analogous to the notion of control and action variables in \cite{PallavoorRW22, chen2024mildly}).\footnote{We remark that in the one-sided adaptive lower bound described above, it would not have been possible to use a one-dimensional action subspace because an adaptive algorithm would be able to detect that ``global structure,'' which is shared across all the $\bU_i$'s; this is why the dimension of the action subspace $\bA$ was $n$ in the earlier construction, and there was a different random ``action direction'' $\bv^j$ from $\bA$ for each $j \in [N]$ in the earlier construction.}  We embed the Nazarov body $\bB$ (described earlier) in the control subspace. We define the $\Dyes$ and $\Dno$ distributions in analogy with \cite{chen2024mildly}, roughly as follows: for $x \in \R^{n+1}$,
  \begin{enumerate}
  \item If the projection $x_C$ does not lie in the uniquely violated set of $\bB$, then $\boldf(x)$ is defined the same way for $\boldf \sim \Dyes$ and $\boldf \sim \Dno$;
  \item If the projection $x_C$ lies in the uniquely violated set, then $\boldf(x)$ is set differently for $\boldf \sim \Dyes$ and $\boldf \sim \Dno$ (depending on the projection $x_{\bA}$ to the action subspace). In particular, for $\boldf \sim \Dyes$, $\boldf$ is defined in such a way that $\boldf^{-1}(1)$ is close to a convex set, and for $\boldf \sim \Dno$, $\boldf$ is defined in such a way that $\boldf^{-1}(1)$ is far from every convex set. 
This crucially uses the fact that the Gaussian volume of $\bU$ is ``large'' compared to the 
Gaussian volume of the set $\Ball(\sqrt{n}) \setminus \bB$, as mentioned in our earlier discussion of the Nazarov body.
\end{enumerate} 
At a high level, the indistinguishability argument showing that $q = 2^{\Omega(n^{1/4})}$ non-adaptive queries are required to distinguish $\boldf \sim \Dyes$ from $\boldf \sim \Dno$ is a case analysis based on the distance between any given pair of query vectors $x$ and $y$ (see \Cref{lem:xyfar} and \Cref{lem:xynear}), combined with a union bound over all ${q \choose 2}$ possible pairs of query vectors.  Roughly speaking, if $\|x-y\|$ is small, then the way that $\boldf$ depends on the projection to the action subspace makes it very unlikely to reveal a difference between $\boldf \sim \Dyes$ and $\boldf \sim \Dno$. On the other hand, if $\|x-y\|$ is large, then it is very unlikely for $x$ and $y$ to lie in the same set $\bU_i$, which must be the case for the pair $x,y$ to reveal a difference between $\boldf \sim \Dyes$ and $\boldf \sim \Dno$.  
There are many technical issues and geometric arguments required to carry out this rough plan, but when all the dust settles the argument gives a $2^{\Omega(n^{1/4})}$ lower bound for tolerant convexity testing.   

   




\subsubsection{Two-Sided Non-Adaptive Bound} Our approach to prove Theorem~\ref{thm:two-sided} is inspired by the lower bounds of \cite{CDST15} on non-adaptive monotonicity testing.
As in most property testing lower bounds for non-adaptive algorithms, the high-level approach is to use Yao's principle; we follow \cite{CDST15} in that we use a suitable high-dimensional central limit theorem as the key technical ingredient for establishing indistinguishability between the yes- and no- distributions.  In \cite{CDST15} both the yes- and no- functions are linear threshold functions over $\bn$, but since any linear threshold function is trivially a convex set, the \cite{CDST15} construction cannot be directly used to prove a convexity testing lower bound.  Instead, in order to ensure that our no- functions are both indistinguishable from the yes- functions and are far from every convex set, we work with \emph{degree-2 polynomial threshold functions (PTFs)} over $\R^n$ rather than linear threshold functions over $\bn$.  At a high level, degree-2 PTFs of the form $\sum_i \lambda_i x_i^2$ where each $\lambda_i$ is positive (note that any such PTF is a convex set) play the ``yes-function'' role that monotone LTFs play in the \cite{CDST15} argument, and degree-2 PTFs of the form $\sum_i \lambda'_i x_i^2$ where a constant fraction of the $\lambda'_i$'s are negative play the ``no-function'' role that far-from-monotone LTFs play in the \cite{CDST15} argument.  We show that having a constant fraction of the $\lambda'_i$'s be negative is sufficient, in the context of our construction, to ensure that no-functions are far from convex, and we show that the multi-dimensional central limit theorem used in \cite{CDST15} can be adapted to our context to establish indistinguishability and thereby prove the desired lower bound.

\subsection{Related Work}
\label{sec:related-work}

%
%

A number of earlier papers have considered different aspects of convexity testing.
One strand of work deals with testing convexity of (real-valued) functions $f \colon [N] \to \R$, where convexity means the second derivative is positive.\footnote{These works study discrete domains, where a discrete derivative is used.} This study was initiated by Parnas et al.~\cite{PRR03}, 
and extended by Pallavor et al.~\cite{PRV18}, who gave an improved result parameterized by the image size of the function being tested; by Blais et al.~\cite{BRY14b}, who gave lower bounds on testing convexity of real-valued functions over the hypergrid $[N]^d$; and by Belovs et al.~\cite{BBB20}, who gave upper and lower bounds on the number of queries required to test convexity of real-valued functions over various discrete domains including the discrete line, the ``stripe'' $[3] \times [N]$, and the hypergrid $[N]^d.$ (See also the work of Berman et al.~\cite{BRY14}, who investigated a notion of ``$L_1$-testing'' real-valued functions over $[N]^d$ for convexity.)

A different body of work, which is closer to this paper, deals with testing convexity of \emph{high-dimensional sets} (equivalently, Boolean indicator functions).  The earliest work we are aware of along these lines is that of Rademacher and Vempala \cite{RademacherVempala05}.\footnote{The study of convexity testing in two dimensions was initiated in earlier work of Raskhodnikova \cite{Raskhodnikova03} for the domain $[N]^2$, and has since been extended to sample-based testing \cite{BMR16}, testing over the continuous domain $[0,1]^2$ \cite{BMR19}, and tolerant testing \cite{BMR22}; see also \cite{B-EF18}.} In their formulation, a body $K \subseteq \R^n$ is $\eps$-far from being convex if $\Leb(K \ \triangle \ C) \geq \eps \cdot \Leb(K)$ for every convex set $C$, where $\Leb(\cdot)$ denotes the Lebesgue volume (note that, in contrast, our model uses absolute volume under the Gaussian measure, rather than relative volume under the Lebesgue measure).  Moreover, \cite{RademacherVempala05} allow the testing algorithm access to a black-box membership oracle (as in our model) as well as a ``random sample'' oracle which can generates a uniform random point from $K$ (for testing with respect to relative measures, such an oracle is necessary). The main positive result of \cite{RademacherVempala05} is a $(cn/\eps)^{n}$ sample- and query- algorithm for testing convexity in their model.  \cite{RademacherVempala05} also give an exponential lower bound for a simple ``line segment tester,'' which checks whether a line segment connecting two (uniformly random) points from the body is contained within the body. This lower bound was strengthened and extended to an exponential lower bound for a ``convex hull tester'' in recent work of Blais and Bommireddi \cite{BB20}.  We note that the negative results of \cite{RademacherVempala05} and \cite{BB20}, while they deal with natural and interesting candidate testing algorithms, only rule out very specific kinds of testers and do not provide lower bounds against general testing algorithms in their framework.

The most closely related work for us is the study of sample-based testing algorithms for convexity under the $N(0,I_n)$ distribution \cite{CFSS17}. As was mentioned earlier, \cite{CFSS17} gave a $2^{\tilde{O}(\sqrt{n})/\eps^2}$-sample algorithm for convexity testing and showed that any sample-based tester must use $2^{\Omega(\sqrt{n})}$ samples; we remark that lower bounds for sample-based testers do not have any implications for query-based testing.\footnote{\cite{CFSS17} also gave a $2^{O(n \log(n/\eps))}$-sample one-sided algorithm, which was generalized to testing under arbitrary product distributions by \cite{HarmsYoshida22}.}  Finally, another closely related paper is the recent work of Blais et al.~\cite{BBH24} which gives nearly matching upper and lower bounds of $3^{\tilde{\Omega}(\sqrt{n})}$ queries for one-sided non-adaptive convexity testing over $\{-1,0,1\}^n$.  \cite{BBH24} cites the high-dimensional Gaussian testing problem as motivation for their study of the ternary cube, and asks ``Can queries improve upon the bounds of \cite{CFSS17,HarmsYoshida22} for testing convex sets with samples in $\R^n$ under the Gaussian distribution?'' (Question 1.15 of~\cite{BBH24}).  Our work makes progress on this question by establishing the first \emph{lower} bounds for query-based testing under the Gaussian distribution.









%% file: sections/prelims.tex

\section{Preliminaries}
\label{sec:prelims}

We use boldfaced letters such as $\bx, \bX$, etc. to denote random variables (which may be real- or vector-valued; the intended type will be clear from the context). 
We write $\bx \sim \calD$ to indicate that the random variable $\bx$ is distributed according to probability distribution $\calD$. We will frequently identify a set $K\sse\R^n$ with its $0/1$-valued indicator function, i.e., $K(x)=1$ if $x\in K$ and $K(x)=0$ otherwise. 
We write $\ln$ to denote natural logarithm and $\log$ to denote base-two logarithm.

\subsection{Geometry}
\label{subsec:prelims-spherical-cap-bounds}

We write $\S^{n-1}$ for the unit sphere in $\R^n$, i.e.  $\S^{n-1} = \{x\in\R^n : \|x\| = 1 \}$ where $\|x\|$ denotes the $\ell_2$-norm of $x$.  
We write $\Ball(r) \sse \R^n$ to denote the $\ell_2$-ball of radius $r$ in $\R^n$, i.e. 
\[
	\Ball(r) := \cbra{ x \in \R^n : \|x\| \leq r }.
\]
We will frequently write $\Ball := \Ball(\sqrt{n})$. 
We recall the following standard bound on the volume of spherical caps (see e.g. Lemma~2.2 of \cite{ball1997elementary}):

\begin{lemma} \label{lem:spherical-cap}
For $0 \leq \eps < 1$, we have
$\Prx\sbra{\bu_1 \geq \eps} \leq e^{-n\eps^2/2}$,
where $\bu \sim \mathbb{S}^{n-1}$, i.e. $\bu$ is a Haar random vector drawn uniformly from the unit sphere $\S^{n-1}$.
\end{lemma}

\subsection{Gaussian and Chi-Squared Random Variables}
\label{subsec:gaussian-tail-bounds}

For $\mu \in \R^n$ and $\Sigma \in \R^{n\times n}$, we write $N(\mu, \Sigma)$ to denote the $n$-dimensional Gaussian distribution centered at $\mu$ and with covariance matrix $\Sigma$. In particular, identifying $0\equiv 0^n$ and writing $I_n$ for the $n\times n$ identity matrix, we will denote the $n$-dimensional {standard} Gaussian distribution by $N(0, I_n)$. We write $\vol(K)$ to denote the Gaussian measure of a (Lebesgue measurable) set $K \subseteq \R^n$, i.e.  
\[\vol(K) := \Prx_{\bg \sim N(0,I_n)}[\bg \in K].\]  

We recall the following standard tail bound on Gaussian random variables:

\begin{proposition}[Theorem~1.2.6 of \cite{durrett_2019} or Equation~2.58 of \cite{TAILBOUND}] \label{prop:gaussian-tails}
	Let $\Phi: \R\to (0,1)$ denote the cumulative density function of the (univariate) standard Gaussian distribution, i.e. 
	\[
		\Phi(r) = \Prx_{\bg\sim N(0,1)}\sbra{\bg\leq r}.
	\] 
	 Then {for all $r>0$,} we have
\[
\varphi(r)
\left({\frac 1 r} - {\frac 1 {r^3}} \right) \leq 1-\Phi(r) \leq
\varphi(r)
\left({\frac 1 r} - {\frac 1 {r^3}} + {\frac 3 {r^5}}\right)
\]
where $\phi$ is the one-dimensional standard Gaussian density which is given by 
\[\phi(x) := \frac{1}{\sqrt{2\pi}}e^{-x^2/2}.\]
\end{proposition}

It is well known that if $\bg\sim N(0,I_n)$, then $\|\bg\|$ is distributed according to the chi distribution with $n$ degrees of freedom, i.e. $\|\bg\| \sim \chi(n)$. 
{It is well known (see e.g.~\cite{Wiki-chisquare}) that the mean of the $\chi^2(n)$ distribution is $n$, the median is $n(1 - \Theta(1/n))$, and for $n\geq 2$ the probability density function is everywhere at most 1. We note that an easy consequence of these facts is that the origin-centered ball $\Ball(\sqrt{n})$ of radius $\sqrt{n}$ in $\R^n$ has $\Vol(B(\sqrt{n})) = 1/2 + o(1).$} 

We will require the following tail bound on $\chi^2(n)$ random variables:

\begin{proposition}[Section 4.1 of \cite{laurent2000}] \label{prop:chi-squared-tail}
	Suppose $\by\sim\chi^2(n)$. Then for any $t > 0$, we have 
	\begin{align*}
		\Prx_{\by\sim\chi^2(n)}\sbra{\by \geq n + 2\sqrt{nt} + 2t} \leq e^{-t}
		\quad\text{\ and\ }\quad
		\Prx_{\by\sim\chi^2(n)}\sbra{\by \leq n - 2\sqrt{nt}} \leq e^{-t} .
	\end{align*}
\end{proposition}
	
\subsection{Property Testing and Tolerant Property Testing}
\label{subsec:testing-prelims}

Let $\Pconv := \Pconv(n)$ denote the class of convex subsets of $\R^n$, i.e. 
\[
	\Pconv = \big\{L \sse\R^n : L~\text{is convex}\big\}. 
\]
Given a set $K \sse \R^n$, we define its \emph{distance to convexity} as
\[
	\dist(K, \Pconv) := \inf_{L \in \Pconv} \Vol(K \, \triangle \, L)
\]
where $K\,\triangle\,L = (K\setminus L) \cup (L\setminus K)$ denotes the symmetric difference of $K$ and $L$. In particular, we will say that $K$ is $\epsilon$-close to (resp.~$\eps$-far from) a convex set if $\dist(K,\Pconv) \leq \eps$ (resp.~$\geq \eps$).  

\begin{definition}[Property testers and tolerant property testers]
	Let $\eps, \eps_1, \eps_2 \in [0,0.5]$ with $\eps_1 < \eps_2$. An algorithm $\calA$ is an \emph{$\eps$-tester} for convexity if, given black-box query access to an unknown set $K \sse \R^n$, it has the following performance guarantee:
	\begin{itemize}
		\item If $K$ is convex, then $\calA$ outputs ``accept'' with probability at least $2/3$;
		\item If $\dist(K, \Pconv) \geq \eps$, then $\calA$ outputs ``reject'' with probability at least $2/3$.
	\end{itemize}
	An algorithm $\calA$ is an \emph{$(\eps_1, \eps_2)$-tolerant tester} (or simply an \emph{$(\eps_1, \eps_2)$-tester}) for convexity if it has the following performance guarantee:
	\begin{itemize}
		\item If $\dist(K, \Pconv) \leq \eps_1$, then $\calA$ outputs ``accept'' with probability at least $2/3$;
		\item If $\dist(K, \Pconv) \geq \eps_2$, then $\calA$ outputs ``reject'' with probability at least $2/3$.
	\end{itemize}
	In particular, note that every $\eps$-tester is a $(0, \eps)$-tolerant tester. 
\end{definition}

Our query-complexity lower bounds for {non-adaptive} property testing algorithms are obtained via Yao's minimax principle~\cite{Yao:77}, which we recall below.
(We remind the reader that an algorithm for the problem of $(\eps_1,\eps_2)$-tolerant testing is correct on an input function $f$ provided that it outputs ``yes'' if $f$ is $\eps_1$-close to the property and outputs ``no'' if $f$ is $\eps_2$-far from the property; if the distance to the property is between $\eps_1$ and $\eps_2$ then the algorithm is correct regardless of what it outputs.)

\begin{theorem}[Yao's principle] \label{thm:yao-minimax}
	To prove an $\Omega(q)$-query lower bound on the worst-case query complexity of any non-adaptive {randomized} testing algorithm, it suffices to give a distribution $\calD$ on instances
	such that for any $q$-query non-adaptive \emph{deterministic} algorithm $\calA$, we have 
	\[\Prx_{\boldf\sim \calD}\big[\calA\text{ is correct on }\boldf\big]\leq c.\] 
	where $0\leq c<1$ is a universal constant. 
\end{theorem}

%% file: sections/nazarov.tex

\section{Nazarov's Body}
\label{sec:nazarov-prelims}

Our constructions in~\Cref{sec:one-sided-adaptive,sec:tolerant} will employ modifications of a probabilistic construction of a convex body due to Nazarov~\cite{Nazarov:03}. 
Nazarov's randomized construction yields a convex set with asymptotically maximal Gaussian surface area~\cite{Ball:93,Nazarov:03}, and modifications thereof have found applications in learning theory and polyhedral approximation~\cite{KOS:08,DNS23-polytope}.   

\begin{definition}[Nazarov's body]
\label{def:nazarov-body}
	For $r, N > 0$, we write $\Naz(r, N)$ to be the distribution over convex subsets of $\R^n$ where a draw $\bB\sim\Naz(r,N)$ is obtained as follows:
	\begin{enumerate}
		\item For $i\in[N]$, draw independent vectors $\bg^i \sim N(0, I_n)$ and let $\bH_i \sse\R^n$ denote the halfspace
		\begin{equation} \label{eq:boldH}
			\bH_i := \{x\in\R^n : x\cdot\bg^i \leq r\}.
		\end{equation}
		\item Output the convex set $\bB\sse\R^n$ where
		\[
			\bB := \Ball(\sqrt{n}) \cap \pbra{\bigcap_{i=1}^N \bH_i}.
		\]
	\end{enumerate}
\end{definition}

Note that for any fixed $x\in\R^n$, 
\begin{align}
	\Prx_{\bH_i}[x\in\bH_i] &= \Prx_{\bg^i\sim N(0,I_n)}\sbra{x\cdot\bg^i \leq r} \nonumber \\ 
	&= \Prx_{\bg^i\sim N(0,I_n)}\sbra{\sum_{j=1}^n x_j\bg^i_j \leq r} \nonumber \\
	&= \Prx_{\bg \sim N(0,1)}\sbra{\bg \leq \frac{r}{\|x\|}} \nonumber \\
	&= \Phi\pbra{\frac{r}{\|x\|}} \label{eq:halfspace-prob}
\end{align}
where $\Phi(\cdot)$ is the univariate Gaussian cumulative density function. 
Consequently, because of the independence of $\bg^i$, we have 
\begin{equation} \label{eq:nazarov-prob}
	\Prx_{\bB\sim\Naz(r,N)}\sbra{x\in\bB} = 
		\Indicator\cbra{\|x\| \leq \sqrt{n}}\cdot\Phi\pbra{\frac{r}{\|x\|}}^N.
\end{equation}
Note that $\bB$ can be also written as
\[
	\bB=\Ball(\sqrt{n}) \setminus \bigcup_{i\in [N]}\Big( \Ball(\sqrt{n})\setminus \bH_i\Big).
\]
For each $i \in [N]$,
we define $\bF_i$ (for ``flap'') to be points in $\Ball(\sqrt{n})$ which are falsified by   $\bH_i$, i.e.
\[
	\bF_i := \Ball(\sqrt{n})\setminus \bH_i.
\]
Given a non-empty $T \subseteq [N]$, we write $\bF_T := \bigcap_{i\in T} \bF_i$. We will be interested in points in $\Ball(\sqrt{n})$ that are falsified by a unique halfspace $\bH_i$ and denote the set of such points as $\bU_i$ (for ``unique''): 
\[
	\bU_i := \bF_i \setminus \bigcup_{j \neq i} \bF_j.
\]

\begin{figure}
	\centering
	\begin{tikzpicture}[x=0.75pt,y=0.75pt,yscale=-1,xscale=1] 

		\def\A{(218,229) .. controls (218,161.62) and (272.62,107) .. (340,107) .. controls (407.38,107) and (462,161.62) .. (462,229) .. controls (462,296.38) and (407.38,351) .. (340,351) .. controls (272.62,351) and (218,296.38) .. (218,229) -- cycle};
	
		\begin{scope}
			\clip \A;
			\fill[color=green, opacity=0.1] \A;
			\fill[color=blue!20] (351.5,56) to (523.5,302)|- cycle;
			\fill[color=white] (442.5,69) to (178,203)|- cycle;
			\fill[color=white] (462.5,91) to (206.5,168)|- cycle;
			\fill[color=white] (267.5,89) to (203.5,341)|- cycle;
			\fill[color=white] (480.5,379) to (187,293)|- cycle;
			\fill[color=white] (321,388) to (514.5,245)|- cycle;
			
			\clip (462.5,91) to (206.5,168)|- cycle;
			\fill[color=red!20] (442.5,69) -- (178,203)|- cycle;
			\fill[color=white] (267.5,89) to (203.5,341)|- cycle;
		\end{scope}

		\draw \A;
		\draw (442.5,69) -- (178,203); 
		\draw (267.5,89) -- (203.5,341); 
		\draw (480.5,379) -- (187,293); 
		\draw (514.5,245) -- (321,388); 
		\draw (462.5,91) -- (206.5,168); 
		\draw (523.5,302) -- (351.5,56);	 
		
		\node[fill, inner sep=1pt,circle] (ori) at (340,235) {};
		\node (oriname) at (340, 225) {\small ~$0^n$};
		\draw (ori) -- (234,290);
		\node (sqrtn) at (280,255) {\small $\sqrt{n}$};
		\draw (ori) -- (314,330);
		\node (rad) at (340,290) {\small $\frac{r}{\|\bg_i\|}$};		
		
		\node (Hi) at (495, 382) {\small $\bH_i$};
		\node (H1) at (535, 305) {\small $\bH_1$};
		\node (H2) at (475, 92) {\small $\bH_2$};
		\node (H3) at (456, 69) {\small $\bH_3$};
		\node (HN) at (310, 392) {\small $\bH_N$};
		
		\node (U1) at (455, 164) {\small $\bU_1$};
		\node (F) at (300, 128) {\small $\bF_{2,3}$};
	\end{tikzpicture}
	
	\
	
	\caption{A depiction of $\bB$ (in green) sampled from $\Naz(r,N)$.}	
\end{figure}
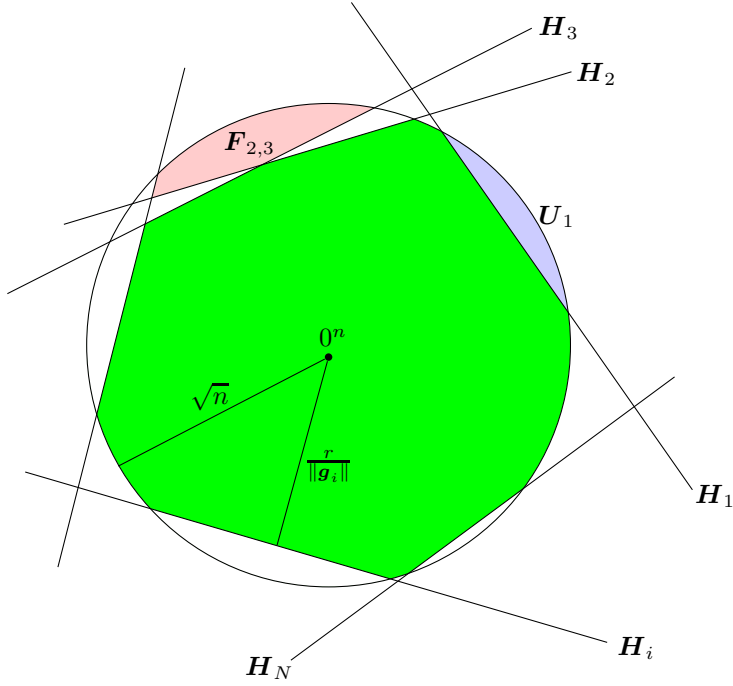

\subsection{Useful Estimates}
\label{subsec:nazarov-prelims}

Suppose $N$ satisfies $N=n^{\omega_n(1)}$; in~both \Cref{sec:one-sided-adaptive,sec:tolerant} we will take $N = 2^{\sqrt{n}}$. 
Let $c_1 > 0$ be a parameter; in~\Cref{sec:one-sided-adaptive}, we will set $c_1 = \ln 2 \pm O(1)/N$,  {and in \Cref{sec:tolerant} we will take $c_1$ to be a suitable small absolute constant.}

Throughout this section we will take $r$ to be the unique positive number such that 
\begin{equation} \label{eq:r-condition}
	\Phi\pbra{\frac{r}{\sqrt{n}}} = 1 - \frac{c_1}{N}.
\end{equation}
Gaussian tail bounds allow us to relate $r$ and $N$: 

\begin{lemma} \label{lemma:r-N-relationship}
	We have 
	\[
		r = \sqrt{2n(1-o(1)) \ln\pbra{\frac{N}{c_1}\sqrt{\frac{n}{2\pi}}}}.
	\]
\end{lemma}

\begin{proof}
	Note that because $N = \omega_n(1)$, it follows that $r = \omega(\sqrt{n})$; otherwise, note that $1 - \Phi\pbra{\frac{r}{\sqrt{n}}} = \Omega_n(1)$, contradicting~\Cref{eq:r-condition}. 
	Next, it follows from~\Cref{prop:gaussian-tails} and~\Cref{eq:r-condition} that 
	\begin{equation} \label{eq:gtb-app}
		\pbra{\frac{\sqrt{n}}{r} - \pbra{\frac{\sqrt{n}}{r}}^3}\cdot\phi\pbra{\frac{r}{\sqrt{n}}}
		\leq \frac{c_1}{N} \leq \frac{\sqrt{n}}{r}\cdot\phi\pbra{\frac{r}{\sqrt{n}}}	.
	\end{equation}
	The upper bound implies that 
	\[
		r\cdot\exp\pbra{\frac{r^2}{2n}} \leq \frac{N}{c_1}\sqrt{\frac{n}{2\pi}}, 
		\qquad\text{and so}\qquad 
		\ln r + \frac{r^2}{2n} \leq \ln\pbra{\frac{N}{c_1}\sqrt{\frac{n}{2\pi}}}.
	\]
	In particular, this implies that 
	\begin{equation} \label{eq:r-upper-bound}
		r \leq \sqrt{2n \ln\pbra{\frac{N}{c_1}\sqrt{\frac{n}{2\pi}}}}.
	\end{equation}
	Next, note that the lower bound from~\Cref{eq:gtb-app} implies that 
	\[
		\frac{N}{c_1}\sqrt{\frac{n}{2\pi}}\pbra{1-\frac{n}{r^2}} \leq r\exp\pbra{\frac{r^2}{2n}},
		\qquad\text{and so}\qquad
		\ln\pbra{\frac{(1-o(1))N}{c_1}\sqrt{\frac{n}{2\pi}}} \leq \ln r + \frac{r^2}{2n}.
	\]
	This in turn implies that
	\begin{equation} \label{eq:r-lower-bound}
		r \geq \sqrt{2n(1-o(1))\cdot \ln\pbra{\frac{N}{c_1}\sqrt{\frac{n}{2\pi}}}}
	\end{equation}
	The result follows from~\Cref{eq:r-upper-bound,eq:r-lower-bound}.
%
\end{proof}

We need the following lemma which will be useful in analyzing our construction in~\Cref{sec:one-sided-adaptive}.

\begin{lemma} \label{lemma:small-vol-high-degree}
	Let $x\in\R^n$ be a point with $\|x\|\leq\sqrt{n}$. Then
	\[\Prx_{\bB\sim\Naz(r,N)}\sbra{x\in \bigcup_{|T|\geq q} \bF_T} \leq \frac{c_1^q}{q!}.\]
	for all $q \in [N]$.
\end{lemma}

\begin{proof}
	Note that
	\begin{align*}
		\Prx_{\bB\sim\Naz(r,N)}\sbra{x \in \bigcup_{|T|\geq q} \bF_T}
		& \leq {N \choose q}\Prx_{\bB\sim\Naz(r,N)}\sbra{x\in \bF_1 \cap \ldots \cap \bF_{q}} \\
		& \leq \frac{1}{q! }\pbra{N\pbra{1-\Phi\pbra{\frac{r}{\|x\|}}}}^q \\
		& \leq \frac{1}{q!}\pbra{N\pbra{1-\Phi\pbra{\frac{r}{\sqrt{n}}}}}^q \\ 
		& = \frac{c_1^q}{q!}
	\end{align*}
	where the penultimate equality relies on~\Cref{eq:r-condition}.
\end{proof}

We will also require a lower bound on the expected volume of $\bigsqcup_{i=1}^N \bU_i$:

\begin{lemma} \label{lemma:unique-expected-volume-lb-adaptive}
	 For constant $0 < c_1 \leq 0.9$ and $N = 2^{\sqrt{n}}$, we have 
	\[
	\Ex_{\bB \sim \Naz(r,N)}\sbra{\vol\pbra{\bigsqcup_{i=1}^N \bU_i}}=\Omega(c_1).
%
	\]
\end{lemma}

\begin{proof}
	Fix any $x \in \R^n$ and any $i \in [N]$. Note that 
	\begin{equation} \label{eq:fix-x-prob-in-U_i}
		\Prx_{\bB\sim\Naz(r,N)}\sbra{x \in \bU_i} = \Indicator\cbra{\|x\|\leq\sqrt{n}}\cdot\pbra{1-\Phi\pbra{\frac{r}{\|\bx\|}}}\Phi\pbra{\frac{r}{\|\bx\|}}^{N-1}.
	\end{equation}
	It follows that 
	\begin{align}
		\EVol{\bU_i} 
		& = \Ex_{\bx\sim N(0,I_n)}\sbra{\Prx_{\bB\sim\Naz(r,N)}\sbra{\bx\in\bU_i}} \nonumber \\
		& = \Vol\pbra{\Ball(\sqrt{n})}\Ex_{\bx\sim N(0,I_n)}\sbra{\pbra{1-\Phi\pbra{\frac{r}{\|\bx\|}}}\Phi\pbra{\frac{r}{\|\bx\|}}^{N-1} ~\Bigg|~ \|\bx\|\leq\sqrt{n}} \nonumber \\
		& \geq {\frac 1 2} \Ex_{\bx\sim N(0,I_n)}\sbra{\pbra{1-\Phi\pbra{\frac{r}{\|\bx\|}}}\Phi\pbra{\frac{r}{\sqrt{n}}}^{N-1} ~\Bigg|~ \|\bx\|\leq\sqrt{n}} \nonumber \\
		& \geq {\frac 1 2} \pbra{1-\frac{c_1}{N}}^{N-1} \Ex_{\bx\sim N(0,I_n)}\sbra{1-\Phi\pbra{\frac{r}{\|\bx\|}} ~\Bigg|~\|\bx\|\leq\sqrt{n}} \nonumber \\
		& \geq {\frac 1 2} \pbra{1- c_1 + \frac{c_1}{N}} \Ex_{\bx\sim N(0,I_n)}\sbra{1-\Phi\pbra{\frac{r}{\|\bx\|}} ~\Bigg|~\|\bx\|\leq\sqrt{n}} \label{eq:simons-institute-for-the-theory-of-computing}
	\end{align}
	where the penultimate inequality follows from~\Cref{eq:r-condition} and the final inequality relies on the fact that $(1-y)^z \geq 1-yz$.

	Next, at the cost of $0.01$ probability mass (thanks to~\Cref{prop:chi-squared-tail}), we can assume that $\|\bx\| \in [\sqrt{n}-10, \sqrt{n}]$. It follows that 
	\[
		\Ex_{\bx\sim N(0,I_n)}\sbra{\pbra{1-\Phi\pbra{\frac{r}{\|\bx\|}}} \Bigg| \sqrt{n}-10\leq \|\bx\|\leq\sqrt{n}} \geq 0.99\cdot\pbra{1-\Phi\pbra{\frac{r}{\sqrt{n}-10}}}.
	\]
	Standard Gaussian tail bounds give that 
	\begin{align*}
		1 \geq \frac{1-\Phi\pbra{\frac{r}{\sqrt{n}-10}}}{1-\Phi\pbra{\frac{r}{\sqrt{n}}}} &\geq \frac{\pbra{\frac{\sqrt{n}-10}{r} - \frac{(\sqrt{n}-10)^3}{r^3}}\exp\pbra{\frac{-r^2}{2(\sqrt{n}-10)^2}}}{\frac{\sqrt{n}}{r}\exp\pbra{\frac{-r^2}{2n}}} \tag{\Cref{prop:gaussian-tails}}\\
		& \geq (1- o(1))\cdot{\exp\pbra{\frac{r^2(100-20\sqrt{n})}{2n(\sqrt{n}-10)^2}}} \\ 
		& = {\Theta(1)}, 
	\end{align*}
where the last line uses our bounds on $r$ from \Cref{lemma:r-N-relationship} and our bounds on $c_1$ from the statement of the current lemma.
%
	Putting everything together and recalling that $c_1 < 0.9$, we get that for $n$ large enough, 
	\begin{align}
		\EVol{\bU_i} &\geq \Omega\pbra{1 - \Phi\pbra{\frac{r}{\sqrt{n}}}} = \Omega\pbra{{\frac{c_1}{N}}}
	\end{align}
	thanks to~\Cref{eq:r-condition}. Consequently, we have 
	\begin{equation} \label{eq:ellicott-citty}
		\EVol{\bigsqcup_{i\in[N]} \bU_i } \geq \Omega(c_1),
	\end{equation}
	completing the proof. 
\end{proof}

Next we show that the volume of $\bigsqcup_{i} \bU_i$ is highly concentrated:

\begin{lemma}\label{lemma:concentration}
	Suppose $N = 2^{\sqrt{n}}$. With probability at least $1-o(1)$, we have
	$$\Vol\left(\bigsqcup_{i\in[N]} \bU_i \right)\ge 0.9\cdot \Ex_{\bB\sim\Naz(r,N)}\left[\Vol{\left(\bigsqcup_{i\in[N]} \bU_i\right) }\right].
	$$
\end{lemma}
\begin{proof}
Let $\Naz^*(r,N)$ be the same distribution as $\Naz(r,N)$ except that
  when drawing $\bB$, each $\bg^i$ is drawn from $N(0,I_n)$ conditioning 
  on $\|\bg^i\|=\sqrt{n}\pm 10n^{1/4}$ (instead of just drawing $\bg^i\sim N(0,I_n)$).
Recall $N=2^{\sqrt{n}}$. By \Cref{prop:chi-squared-tail}, the probability of 
  $\|\bg^i\|\notin [\sqrt{n}-10n^{1/4},\sqrt{n}+10n^{1/4}]$ for some $i\in [N]$ is at most $o(1)$.
As a result, we have 
$$
\Ex_{\bB\sim \Naz^*(r,N)}\left[\Vol{\left(\bigsqcup_{i\in[N]} \bU_i\right) }\right]
\ge \Ex_{\bB\sim \Naz (r,N)}\left[\Vol{\left(\bigsqcup_{i\in[N]} \bU_i\right) }\right]-o(1).
$$
Moreover, it suffices to show that when $\bB\sim \Naz^*(r,N)$, we have
\begin{align}\label{eq:hehe3}	
\Vol\left(\bigsqcup_{i\in[N]} \bU_i \right)\ge 0.99\cdot \Ex_{\bB\sim\Naz^*(r,N)}\left[\Vol{\left(\bigsqcup_{i\in[N]} \bU_i\right) }\right].
\end{align}
	with probability at least $1-o(1)$.
To this end, we recall McDiarmid's inequality:

\begin{theorem}[McDiarmid bound \cite{McDiarmid:89}]
Let $\bX_1,
\dots, \bX_S$ be independent random variables taking values in a set
$\Omega$. Let $G \colon \Omega^S \to \R$ be such that for all $i \in
[S]$ we have
\[
\big|G(x_1, \dots, x_S) - G(x_1, \dots, x_{i-1}, x_i', x_{i+1}, \dots,
x_S)\big| \leq c_i
\]
for all $x_1, \dots, x_S$ and $x_i'$ in $\Omega$.  Let $\mu =
\E[G(\bX_1,\dots, \bX_S)]$.  Then for all $\tau> 0$, we have
\[
\Pr\left[G(\bX_1, \dots, \bX_S) < \mu - \tau\right] <
\exp\left(-\frac{\tau^2}{\sum_{i\in [S]} c_i^2}\right).
\]
\end{theorem}

We will take $S=N$, $\bX_i$ to be the halfspaces $\bH_i$ and $G(\cdot)$ to be the volume 
  of $\sqcup_{i\in [N]}\bU_i$, as we draw $\bB\sim \Naz^*(r,N)$. 
Given the way $\bg^i$ is drawn in $\bB\sim \Naz^*(r,N)$, the
  volume of each $\bH_i$ is always at least (using $r\ge \sqrt{2n^{3/2}(1-o(1)}$ by \Cref{lemma:r-N-relationship})
$$
\Phi\left(\frac{r}{\sqrt{n}+10n^{1/4}}\right)\ge 1-e^{-(1-o(1))\sqrt{n}}, 
$$  
from which we have $c_i\le e^{-(1-o(1))\sqrt{n}}$.
As a consequence, 
\[
	\sum_{i\in [N]} c_i^2 \leq N\cdot  e^{-(1-o(1))\sqrt{n}}=e^{-\Omega(\sqrt{n})}.
\]
It follows from McDiarmid that \Cref{eq:hehe3} holds with probability at least $1-o(1)$.
\end{proof}

%
%

Finally, the following lemma will allow us to obtain bounds on the distance to convexity of the ``yes''- and ``no''-distributions in~\Cref{sec:tolerant}:

\begin{lemma} \label{lemma:flaps-vs-dog-ears}
	For $r$ satisfying~\Cref{eq:r-condition}, we have 
	\[
		\Ex_{\bB\sim\Naz(r,N)}\sbra{\Vol\left(\bigsqcup_{i\in[N]}\bU_i\right)} \geq 
		\pbra{\frac{2}{c_1} - {2}} \Ex_{\bB\sim\Naz(r,N)}\sbra{\Vol\pbra{\bigcup_{|T|\geq 2}\bF_T}}.
	\]
\end{lemma}

\begin{proof}
	Fix $x\in\R^n$ and $i\in[N]$. Recall \Cref{eq:fix-x-prob-in-U_i}.
On the other hand, we have 
	\begin{align}
		\Prx_{\bB\sim\Naz(r,N)}\sbra{x\in \bigcup_{|T|\geq 2}\bF_T} 
		&\leq {N \choose 2} \Prx_{\bB\sim\Naz(r,N)}\sbra{x\in\bF_1\cap\bF_2} \nonumber \\ 
		&= {N \choose 2} \Prx_{\bB\sim\Naz(r,N)}\sbra{x\in \Ball(\sqrt{n}) \cap (\R^n \setminus \bH_1) \cap (\R^n\setminus\bH_2)} \nonumber \\ 
		&\leq \frac{N^2}{2}\cdot\Indicator\cbra{\|x\|\leq\sqrt{n}}\cdot\pbra{1-\Phi\pbra{\frac{r}{\|x\|}}}^2 \label{eq:dog-ear-ub-fixed-x}
	\end{align}
	where we once again used~\Cref{eq:halfspace-prob}. It follows from~\Cref{eq:fix-x-prob-in-U_i,eq:dog-ear-ub-fixed-x} that for $x\in\Ball(\sqrt{n})$ (i.e $\|x\|\leq\sqrt{n}$), we have 
	\begin{align}
		\frac{N\cdot\Pr\sbra{x\in \bU_i}}{\Pr[x\in \bigcup_{|T|\geq 2} \bF_T]} &\geq \pbra{\frac{2}{N}}\Phi\pbra{\frac{r}{\|x\|}}^{N-1}\pbra{1-\Phi\pbra{\frac{r}{\|x\|}}}^{-1} \nonumber \\ 
		&\geq \pbra{\frac{2}{N}}\Phi\pbra{\frac{r}{\sqrt{n}}}^{N-1}\pbra{1-\Phi\pbra{\frac{r}{\sqrt{n}}}}^{-1}  \label{eq:ratio-calc-phi-increasing} \\
		&= \pbra{\frac{2}{c_1}}\pbra{1-\frac{c_1}{N}}^{N-1} \label{eq:altoids} \\
		&\geq \pbra{\frac{2}{c_1}}\pbra{1 - c_1 + \frac{c_1}{N}} \label{eq:icebreakers} \\
		&>\frac{2}{c_1} - 2\nonumber
	\end{align}
	where~\Cref{eq:ratio-calc-phi-increasing} relies on the fact that $\|x\|\leq\sqrt{n}$ and $\Phi(\cdot)$ being increasing, \Cref{eq:altoids} relies on our definition of $r$ from~\Cref{eq:r-condition}, and~\Cref{eq:icebreakers} relies on Bernoulli's inequality: $(1-y)^z \geq 1- yz$. 
	(Note that for $x$ with $\|x\| > \sqrt{n}$, we have $\Pr[x\in\bU_i] = \Pr[x \in \bigcup_{|T|\geq 2} \bF_T] = 0$.)
	
	To conclude, we have 
	\begin{align}
		N\cdot\Ex_{\bB\sim\Naz(r,N)}\sbra{\Vol(\bU_i)} 
		&= \Ex_{\bB\sim\Naz(r,N)}\sbra{N\cdot\Prx_{\bx\sim N(0,I_n)}\sbra{\bx\in\bU_i}} \nonumber \\
		&= \Ex_{\bx\sim N(0,I_n)}\sbra{N\cdot\Prx_{\bB\sim\Naz(r,N)}\sbra{\bx\in\bU_i}} \nonumber \\
		&\geq \pbra{\frac{2}{c_1} - 2}\Ex_{\bx\sim N(0,I_n)}\sbra{\Prx_{\bB\sim\Naz(r,N)}\sbra{\bx\in \bigcup_{|T|\geq 2}\bF_T}} \label{eq:fix-x-flap-dog-ratio-application} \\
		&= \pbra{\frac{2}{c_1} - 2} \Ex_{\bB\sim\Naz(r,N)}\sbra{\Prx_{\bx\sim N(0,I_n)}\sbra{\bx\in \bigcup_{|T|\geq 2}\bF_T}} \nonumber \\
		&= \pbra{\frac{2}{c_1} - 2} \Ex_{\bB\sim\Naz(r,N)}\sbra{\Vol\pbra{\bigcup_{|T|\geq 2}\bF_T}} \nonumber
	\end{align}
	where~\Cref{eq:fix-x-flap-dog-ratio-application} follows from the earlier calculation, completing the proof. 
\end{proof}

%% file: sections/one-sided-adaptive-2.tex

\section{One-Sided Adaptive Lower Bound}
\label{sec:one-sided-adaptive}

For this section, it will be most convenient for us to work over $\R^{2n}$. Let us restate Theorem~\ref{thm:one-sided} in this setting:

\begin{theorem} [One-sided adaptive lower bound, restated] \label{thm:one-sided2}
For some absolute constant $\eps>0$, any one-sided $\eps$-tester for convexity  over $N(0,I_{2n})$ (which may be adaptive) must use $n^{\Omega(1)}$ queries.
\end{theorem}

At a high level, the proof of Theorem~\ref{thm:one-sided} works by (1) first defining a distribution $\Dno$ of ``no-functions'' (Boolean-valued functions over $\R^{2n}$, or equivalently, subsets of $\R^{2n}$), and showing that an $\Omega(1)$ fraction of draws from $\Dno$ yield sets which are $\Omega(1)$-far from convex; and (2) then arguing that for a suitable absolute constant $c>0$, any $n^c$-query algorithm (even an adaptive one) has only an $o(1)$ probability of querying a set of points whose labels are inconsistent with every convex set in $\R^{2n}$.
In the next subsection we describe the distribution $\Dno$.

\subsection{The distribution $\Dno$ of far-from-convex sets}

\subsubsection{Setup}
\label{sec:adaptive-setup}

We will see that every function $f$ in the support of $\Dno$ outputs $0$
on every input point $x\in \mathbb{R}^{2n}$ with $\|x\|>\sqrt{2n}$.
To describe how $f$ behaves within the $\sqrt{2n}$-ball, denoted by
\[
	\Ball(\sqrt{2n}):= \big\{x\in \mathbb{R}^{2n}: \|x\|\le \sqrt{2n}\big\},
\] 
we require some more setup.
\medskip

\noindent {\bf The ``control subspace,''  the ``action subspace,'' and the Nazarov body.}
Let $\bC$ be a Haar random $n$-dimensional subspace of $\R^{2n}$; we call $\bC$ the \emph{control subspace}.  
Let $\bA$ be the orthogonal complement of $\bC$ (which is also an $n$-dimensional subspace); we call $\bA$ the ``action subspace.''
Given a vector $x \in \R^n$, we write $x_{\bC}$ to denote the orthogonal projection of $x$ onto $\bC$ and we write $x_{\bA}$ to denote the orthogonal projection of $x$ onto $\bA$, so every vector satisfies $x = x_{\bA} + x_{\bC}$.

Fix $N:= 2^{\sqrt{n}}$ (we assume without loss of generality that $n$ is a perfect square, so $N$ is an integer).
For this choice of $N$, let $\bB \sim \Naz(r,N,\bC)$ where $\Naz(r,N,\bC)$ is as defined in \Cref{def:nazarov-body} but with the $n$-dimensional control subspace $\bC$ playing the role of $\R^n$.  (We emphasize that $\bB \sim \Naz(r,N,\bC)$ is a subset of $\R^{2n}$ which is an ``$n$-subspace junta,'' meaning that for any $x \in \R^{2n}$, membership of $x$ in $\bB$ depends only on $x_{\bC}.$)
We take $r$ to be the unique positive number such that 
	\[
		\Phi\pbra{\frac{r}{\sqrt{n}}}^N = \frac{1}{2}.
	\]
	In other words, we choose $r$ to be the unique value such that any point $x$ with $\|x_{\bC}\|=\sqrt{n}$ has probability $1/2$ of being in $\bB \sim \Naz(r,N,\bC)$ (cf.~\Cref{eq:nazarov-prob}). Note that  
	\[
		\Phi\pbra{\frac{r}{\sqrt{n}}} = \pbra{\frac{1}{2}}^{\frac{1}{N}} = 1 - \frac{c_1}{N}
			\quad \text{for a value~}c_1 ={\ln 2} \pm {\frac {O(1)}{N}}
	\]
	by the Taylor expansion of $e^{-\ln(2)/N}$ and setting of $r$ (Lemma~\ref{lemma:r-N-relationship}).

%



\medskip

\noindent {\bf The ``action directions.''} 
For each $i \in [N]$, draw a random vector $\bv^i$ from the standard Normal distribution $N(0,I_n)$ over the $n$-dimensional action subspace $\bA$ (independent of everything else).  We say that $\bv^i$ is the \emph{action direction} for the $i$-th flap $\bF_i$ of the Nazarov body $\bB$.  We note that for every pair $i,j \in [N]$, the vector $\bg^i$ defining the $i$-th halfspace $\bH^i$ of the Nazarov body is orthogonal to the vector $\bv^j$ (because $\bg^i \in \bC$ and $\bv^j \in \bA$).


\subsubsection{The distribution $\Dno$}\label{subsec:adapt-dno}

For a fixed setting of the control subspace $C$ and the (orthogonal) action subspace $A$, of $\vec{H} := (H_1,\dots,H_N)$ (which also specifies $B$ and $F_i$'s) and of $\vec{v}:=(v^1,\dots,v^N)$, we define the function $f_{C,A,\vec{H},\vec{v}}: \R^{2n} \to \{0,1\}$ as follows:
 

\[
f_{C,A,\vec{H},\vec{v}}(x) = \begin{cases}
0 & x \notin \Ball(\sqrt{2n}) \text{~or~} \|x_C\| > \sqrt{n};\\
1 & x \in \Ball(\sqrt{2n}) \text{~and~} x_C \in B;\\
\bigwedge_{j \in T}\Indicator\sbra{\langle v^j,x\rangle \notin [-{\frac {\sqrt{n}}{2}},{\frac {\sqrt{n}}{2}}]} &  x \in \Ball(\sqrt{2n}) \text{~and~}x_C \in F_T
\text{~for some~}\emptyset \neq T \subseteq [N].
\end{cases}
\]
A random function $\boldf \sim \Dno$ is drawn as follows:  first we draw a Haar random $n$-dimensional subspace $\bC$; then $\bA$ is taken to be the $n$-dimensional (Haar random) orthogonal complement of $\bC$; then we draw $\bB \sim \Naz(r,N,\bC)$ (which gives a draw of $\vec{\bH}$ as in \Cref{eq:boldH}); then we draw  $\vec{\bv}=(\bv^1,\dots,\bv^N)$ from $\bA$ as described above; then we set the function
  $\boldf$ to be $\smash{f_{\bC,\bA,\vec{\bH},\vec{\bv}}}$.


\subsection{Sets in $\Dno$ are far from convex}

We need a constant fraction of the no-functions to be constant-far from convex. This is given by the following lemma:

\begin{lemma}  \label{lem:distance}
With probability $\Omega(1)$ over a draw of $\boldf \sim \Dno$, we have that $\Vol(\boldf  \triangle g)=\Omega(1)$ for every $g: \R^{2n} \to \zo$ that is the indicator function of a convex set in $\R^{2n}$.
\end{lemma}

We require a few definitions.
Define $\Thinshell := \{x \in \R^{2n}: \sqrt{2n}-2 \leq \|x\| \leq \sqrt{2n}-1\}.$  
Given an outcome of $\boldf \sim \Dno$ (which determines the $\bg^i$'s, $\bv^i$'s, $\bF^i$'s and $\bU_i$'s), for $i \in [N]$ define $\bU := \sqcup_{i \in [N]} \bU_i$.
Define $p := \Ex_{\boldf \sim \Dno}[\Vol[\bU \cap \Thinshell]].$

\begin{lemma} \label{lem:p-enough} 
$p = \Omega(1) \implies$  \Cref{lem:distance}.
\end{lemma}

\begin{proof}
If $p=\Omega(1)$ then $\Pr_{\boldf}[\Vol[\bU \cap \Thinshell]=\Omega(1)] = \Omega(1)$.
We view the draw of $\boldf$ as taking place in two stages: in the first one $\bC$, $\bA$, and $\vec{\bg} = (\bg^1,\dots,\bg^N)$ are drawn, and in the second one $\vec{\bv} = (\bv^1,\dots,\bv^N)$ is drawn. Observe that the set $\bU$ depends only on the first stage.
Say that any outcome of the first stage for which $\Vol[\bU \cap \Thinshell]=\Omega(1)$ holds is a \emph{good} outcome of the first stage, so an $\Omega(1)$ fraction of outcomes of the first stage are good.  

Fix any good outcome  $C,A,\vec{g}$ of the first stage (note that this fixes $U_1,\dots,U_N$ and hence $U$), and consider a draw of $\bx \sim N(0,I_{2n})$. 
We have the following claim:
\begin{claim} \label{claim:a}
For a suitable absolute constant $a>0$, we have
$\Prx_{\bx \sim N(0,I_{2n})}[\bx \in U \cap \Thinshell$ and $\|\bx_C\| \in [\sqrt{n}-a,\sqrt{n}]] = \Omega(1).$
\end{claim}
\begin{proof}
Since we have fixed a good outcome $C,A,\vec{g}$ of the first stage, we have that $\Pr_{\bx \sim N(0,I_{2n})}[\bx \in U \cap \Thinshell] \geq c$ for some absolute constant $c>0.$  
Moreover, every outcome of $\bx \in U \cap \Thinshell$ has $\|\bx_C\| \leq \sqrt{n}$, since $U$ is a subset of $B$.
So to prove the claim we need only show that $\Pr_{\bx \sim N(0,I_{2n})}[\|\bx_C\| < \sqrt{n}-a] \leq c/2.$

We first observe that by standard bounds on the chi-square distribution (\Cref{prop:chi-squared-tail}), we have that
$\Pr_{\bx \sim N(0,I_{2n})}[\|\bx\| \notin [\sqrt{2n} - a', \sqrt{2n} + a']] \leq c/4$ for a suitable constant $a'$.  So fix any length $\ell \in [\sqrt{2n} - a', \sqrt{2n} + a'].$
Fix any vector $z \in \R^{2n}$ with $\|z\|=\ell$; by the rotational symmetry of the $N(0,I_{2n})$ distribution and the rotational symmetry of drawing a Haar random $n$-dimensional subspace $\bC$ of $\R^{2n}$, the distribution of $\|\bx_{C}\|$ conditioned on $\|\bx\|=\ell$ is the same as the distribution of $\|z_{\bC}\|$ where $\bC$ is a Haar random $n$-dimensional subspace $\bC$ of $\R^{2n}$.  A routine application of the Johnson-Lindenstrauss theorem (see e.g. Theorem~5.3.1 of \cite{vershynin2018high}) gives us that $\Pr_{\bC}[\|z_{\bC}\| <\sqrt{n}-a] \leq c/4$, for a suitable choice of the constant $a$.
So $\Pr_{\bx \sim N(0,I_{2n})}[\|\bx_C\| < \sqrt{n}-a] \leq c/2$ as required, and the claim is proved.
\end{proof}

Now, given an $x$ that lies in $U \cap \Thinshell$ and has $\|x_C\| \in [\sqrt{n}-a,\sqrt{n}],$ consider an outcome of the second stage, i.e.~the draw of $\vec{\bv}$; note that this draw completes the draw of $\boldf \sim \Dno$.  Define the vectors
\[
x^+ := x + {\frac {\bv^i}{\|\bv^i\|}}, \quad \quad x^- := x - {\frac {\bv^i}{\|\bv^i\|}}.
\]
Let us say that an outcome of $\vec{\bv}$ for which $\boldf(x)=0$, $\boldf(x^+)=1$, $\boldf(x^-)=1$ is a \emph{fine outcome of $\vec{\bv}$ for $x$}.  
We will use the following claim:

\begin{claim} \label{claim:fine}
For any fixed  $x$ that lies in $U \cap \Thinshell$ and has $\|x_C\| \in [\sqrt{n}-a,\sqrt{n}],$ we have $\Pr_{\vec{\bv}}[\vec{\bv}$ is fine for $x]=\Omega(1)$.
\end{claim}
\begin{proof}
Since $x \in U_i \cap \Thinshell$ for some $i$, it must be the case that also $x^+,x^- \in U_i$ (because every possible outcome of $\bv^i$ is orthogonal to every possible outcome of $g^j$ for every $j \in [N]$). So $\vec{\bv}$ is fine if and only if 
\[
\abra{\bv^i,x^-} < -{\frac {\sqrt{n}} 2} \leq
\abra{\bv^i,x} \leq {\frac {\sqrt{n}} 2} <
\abra{\bv^i,x^+}, \quad \text{or equivalently},
\]
\begin{equation} \label{eq:zz}
\abra{\bv^i,x} - \|\bv^i\| < -{\frac {\sqrt{n}} 2} \leq
\abra{\bv^i,x} \leq {\frac {\sqrt{n}} 2} <
\abra{\bv^i,x} + \|\bv^i\|.
\end{equation}
Since $x \in \Thinshell$ we have $\sqrt{2n}-2 \leq \|x\| \leq \sqrt{2n}-1$, i.e.~
\[
2n - 4 \sqrt{2n} + 4 \leq \|x\|^2 = \|x_C\|^2 + \|x_A\|^2 \leq 2n - 2\sqrt{2n} + 1,
\]
 and since $\|x_C\| \in [\sqrt{n}-a,\sqrt{n}]$ we have that $n - 2 a \sqrt{n} + a^2 \leq \|x_C\|^2 \leq n$. So
 \begin{equation} \label{eq:xAbound}
 n - 4 \sqrt{2n} + 4 \leq \|x_A\|^2 \leq n - 2 \sqrt{2n} + 2a\sqrt{n} + 1 - a^2.
 \end{equation}
 Now since $\bv^i$ is drawn from a standard $N(0,I_n)$ distribution over the subspace $A$, a routine calculation using 
 (i) \Cref{eq:xAbound}; 
 (ii) the fact that $\|\bv^i - \abra{\bv^i,x} {\frac x{\|x\|}}\|^2$ and $\abra{\bv^i,x}$ are independent and are distributed as a draw from the $\chi^2(n-1)$ distribution and a draw from $N(0,\|x_A\|^2)$ respectively; and
(iii) the fact that a draw from the $\chi^2(n-1)$ distribution takes value $n(1 \pm o(1))$ except with vanishingly small probability,
gives that \Cref{eq:zz} holds with $\Omega(1)$ probability. 
\end{proof}

As an immediate consequence of \Cref{claim:fine}, we get that an $\Omega(1)$ fraction of outcomes of $\vec{v}$ are such that
\begin{equation} \label{eq:dandy}
\Prx_{\bx \sim N(0,I_{2n})}\sbra{
v \text{~is fine for~}\bx \ | \ \bx \in U \cap \Thinshell \ \& \ \|\bx_C\| \in \sbra{\sqrt{n}-a,\sqrt{n}}}
= \Omega(1).
\end{equation}

Fix any outcome $\vec{v}$ of $\vec{\bv}$ for which \Cref{eq:dandy} holds.
To conclude the proof of \Cref{lem:distance}, we observe that since $x \in U_i$ implies that $x^+,x^-$ are also in $U_i$, it follows that any $z \in \R^n$ can participate in at most three triples of the form $(x,x^-,x^+)$, so the maximum possible degree of overlap across all of the triples is at most a factor of three.
Moreover, for any $x \in \Thinshell$, it holds that $\sqrt{2n} - 3 \leq \|x\|-1\leq \|x^+\|,\|x^-\| \leq \|x\|+1 \leq \sqrt{2n}$, and hence the pdf of the $\chi^2(2n)$ distribution is within a constant factor on each of the three inputs $\|x\|,\|x^+\|$ and $\|x^-\|$ (so the $N(0,I_{2n})$ Gaussian's pdf is within a constant factor on each of the three inputs $x,x^+,x^-$).
Combining this with \Cref{claim:a},
we get that for an $\Omega(1)$ fraction of outcomes of $\boldf \sim \Dno$, the value of $\boldf$ needs to be altered on at least an $\Omega(1)$ fraction of inputs drawn from $N(0,I_n)$ in order to ``repair'' all of the violating triples $(x,x^+,x^-)$ for which $x \in U \cap \Thinshell$ and $\|x_C\| \in [\sqrt{n}-a,\sqrt{n}].$  This gives \Cref{lem:p-enough}.
 \end{proof}
 
\begin{proof}[\Cref{lem:distance}]
To prove \Cref{lem:distance} it remains only to show that $p=\Omega(1),$ i.e.~to show that
\begin{equation} \label{eq:fx}
\Prx_{\boldf \sim \Dno,\bx \sim N(0,I_{2n})}[x \in (\bU \cap \Thinshell)]=\Omega(1).
\end{equation}
We first observe that we have $\Pr_{\bx \sim N(0,I_{2n})}[\bx \in \Thinshell] = \Omega(1)$. Fix any outcome $x \in \Thinshell$. Consider a draw of the Haar random $n$-dimensional subspace $\bC$ of $\R^{2n}$ which is part of the draw of $\boldf \sim \Dno$. Similar to the proof of~\Cref{claim:a}, using $\|x\| \in [\sqrt{2n}-2,\sqrt{2n}-1]$ the Johnson-Lindenstrauss theorem gives that $\Pr_{\bC}[\|x_{\bC}\| \in [\sqrt{n}-1,\sqrt{n}]] = \Omega(1)$.  Finally, fix any outcome $C$ of $\bC$ such that $\|x_C\| \in [\sqrt{n}-1,\sqrt{n}]$, and consider the ``completion'' of the draw of $\boldf \sim \Dno$ (i.e.~the draw of $\bB \sim \Naz(r,N,C)$ which induces an outcome of $\bU$). We have
\begin{equation}\label{eq:make-me-constant}
\Prx_{\boldf}[x \in \bU] = N \cdot \pbra{1-\Phi\pbra{\frac{r}{\|x_C\|}}}\Phi\pbra{\frac{r}{\|x_C\|}}^{N-1},
\end{equation}
so to complete the proof of  \Cref{lem:distance} it suffices to show that $\eqref{eq:make-me-constant} = \Omega(1)$. We have 
\[
\Phi\pbra{\frac{r}{\|x_C\|}}^{N-1}
\geq
\Phi\pbra{\frac{r}{\sqrt{n}}}^{N-1}
=
\pbra{1 - \frac{c_1}{N}}^{N-1}=\Omega(1),
\]
where the first equality is \Cref{eq:r-condition} and the second is because $c_1 = \Theta(1)$.
Similar to the proof of \Cref{lemma:unique-expected-volume-lb-adaptive}, we have
	\begin{align*}
		\frac{1-\Phi\pbra{\frac{r}{\|x_C\|}}}{1-\Phi\pbra{\frac{r}{\sqrt{n}}}}
		 \geq \frac{1-\Phi\pbra{\frac{r}{\sqrt{n}-1}}}{1-\Phi\pbra{\frac{r}{\sqrt{n}}}} &\geq \frac{\pbra{\frac{\sqrt{n}-1}{r} - \frac{(\sqrt{n}-1)^3}{r^3}}\exp\pbra{\frac{-r^2}{2(\sqrt{n}-1)^2}}}{\frac{\sqrt{n}}{r}\exp\pbra{\frac{-r^2}{2n}}} \tag{\Cref{prop:gaussian-tails}} \\
		& \geq (1- o(1))\exp\pbra{\frac{r^2(1-2\sqrt{n})}{2n(\sqrt{n}-1)^2}} \\ 
		& = \Theta(1), \quad \text{using \Cref{lemma:r-N-relationship}}.
	\end{align*}
So
\[
N \cdot \pbra{1-\Phi\pbra{\frac{r}{\|x_C\|}}}
\geq N \cdot \Theta(1) \cdot \pbra{1-\Phi\pbra{\frac{r}{\|x_C\|}}}
= N \cdot \Theta(1) \cdot \frac{c_1}{N} = \Omega(1),
\]
where the first equality is by \Cref{eq:r-condition}. This concludes the proof of \Cref{lem:distance}. 

\end{proof} 

\subsection{Proof of Theorem~\ref{thm:one-sided}}

\newcommand{\Alg}{\mathrm{Alg}}
\newcommand{\ol}[1]{\overline{#1}}
\newcommand{\lsim}{\lesssim}
\newcommand{\conv}{\mathrm{conv}}


\begin{definition}[One-sided Adaptive Algorithms as Binary Trees]
Fix $n, q \in \N$. A $q$-query one-sided deterministic algorithm, $\Alg$, for testing convexity in $\R^{2n}$ is specified by a rooted binary tree of depth $q$ where each node contains the following information:
\begin{flushleft}\begin{itemize}
\item Each node $v$ which is a not a leaf contains a query vector $x_v \in \R^{2n}$, as well as two out-going edges, one labeled 0 and one labeled 1, to nodes which we label $v(0)$ and $v(1)$, respectively.
\item Each leaf node $v$ contains an output $o_v$ which is set to \emph{``accept''} or \emph{``reject.''} 
Let $Q_1$ (or $Q_0$) denote the set of points queried along the path that are labelled $1$
  (or $0$, respectively).
Then $o_v$ is set to be ``reject'' if and only if $Q_0\cap \conv(Q_1)\ne \emptyset$.
\end{itemize}\end{flushleft}
By adding nodes which repeat the queries, we may assume, without loss of generality, that the depth of every leaf of the tree is exactly $q$. 

\end{definition}
A $q$-query deterministic algorithm $\Alg$ executes on a function $f \colon \R^{2n} \to \{0,1\}$ by taking the natural root-to-leaf path given by following the function values which the oracle returns at the queries within each of the nodes. In particular, we will make repeated use of the following definitions which capture the execution of the algorithm $\Alg$ on a function $f$:
\begin{flushleft}\begin{itemize}
\item The node $v^0$ is the root of the tree, which is the starting point of the root-to-leaf path. Then, the nodes $v^1,\dots, v^q$ indicate the root-to-leaf path generated by executing the algorithm on the function $f$. In particular, at time step $t \in \{ 0, \dots, q-1\}$, we have $v^{t+1} = v^{t}(f(x_{v^{t}}))$
\item The set $Q^0$ is defined to be $\emptyset$, and for $t \in \{0, \dots, q-1\}$ the set $Q^{t+1}$ is defined to be $Q^{t} \cup \{ x_{v^{t}} \} \subset \R^n$.
Thus $Q^{t+1}$ is the set of vectors that are queried at time steps prior to $t+1$.
\end{itemize}\end{flushleft}
Once the algorithm reaches the leaf node $v^{q}$, the algorithm outputs $o_{v^q}$, and we will refer to $\Alg(f)$ as the output (``accept'' or ``reject'') produced by the algorithm. It is trivial to see that since any $q$-query deterministic algorithm corresponds to a tree of depth $q$, the total number of query vectors $x_v \in \R^{2n}$ across all nodes of the tree is at most $2^q$. Our goal is to show that, if $\Alg$ is a $q$-query deterministic algorithm which makes \emph{one-sided error}, then
\begin{align}
\Prx_{\boldf \sim \Dno}\left[ \Alg(\boldf) = \text{``reject''}\right] = o(1).  \label{eq:adapt-final-exp}
\end{align}
Recall that implicit in a fixed function $f$ in the support of $\Dno$ are the control and action subspaces $C, A \subset \R^{2n}$, as well as the vectors $g^1,\dots, g^{N} \in C$ and $v^1,\dots,v^{N} \in A$, and that $g^1,\dots,g^N$ define $B$, $H_i$ and $F_i$ regions. In order to simplify our notation, we will often refer to a subset of the queries $\tilde{Q}^k$ for any $k \leq q$ whose norm on the control subspace is bounded,
\[ \tilde{Q}^k = \left\{ x \in Q^k : \|x_C\| \leq \sqrt{n}\right\}. \]
Toward showing the above upper bound, we define two important events (which will depend on the draw $\boldf \sim \Dno$).



\begin{definition}
Given $\Alg$ and a function $f$ from $\Dno$, we consider the following three events:
\begin{flushleft}\begin{itemize}
\item $\calE_1(f)$: This event occurs 
if at the end of the execution of
  $\Alg$ on $f$, every point $x\in \tilde{Q}^q$ lies in at most $q$ flaps, and
   for every flap $F_i$ with $\tilde{Q}^q\cap F_i\ne\emptyset$,
  \begin{equation} \label{eq:iwantit}
  \|x - y\| \le 1000 \sqrt{q} n^{1/4}\quad\text{for all $x,y\in \tilde{Q}^q\cap F_i$.}
  \end{equation}
\item $\calE_2(f)$: This event occurs if at the end of
  the execution of $\Alg$ on $f$, for every flap $F_i$ with
  $\tilde{Q}^q\cap F_i\ne \emptyset$ and every $x,y\in \tilde{Q}^q\cap F_i$, we have 
$$
  \Indicator\sbra{\langle v^i , x\rangle \notin [-\sqrt{n}/2,\sqrt{n}/2]} = \Indicator\sbra{\langle v^i,y
  \rangle \notin [-\sqrt{n}/2,\sqrt{n}/2]}.
  $$
\end{itemize}\end{flushleft}
\end{definition}

Theorem \ref{thm:one-sided} follows immediately from the following three lemmas:

\begin{lemma}\label{lem:reject-and-3}
Let $\Alg$ be a one-sided, deterministic, $q$-query algorithm for testing convexity. Then, if $\Alg(f)$ outputs ``reject,'' the event $\overline{\calE_2(f)}$ occurred.
\end{lemma}

\begin{lemma}\label{lem:ev1}
Let $\Alg$ be a one-sided,  deterministic, $q$-query algorithm. Then,
\begin{align*}
\Prx_{\boldf \sim \Dno}\left[ \calE_1(\boldf)\right] \ge 1- o(1). 
\end{align*}
\end{lemma}

\begin{lemma}\label{lem:ev3}
Let $\Alg$ be a one-sided, deterministic, $q$-query algorithm, where $q \leq n^{0.05}$. Then, 
\begin{align*}
\Prx_{\boldf \sim \Dno}\left[ \overline{\calE_2(\boldf)}\cap \calE_1(\boldf)\right] \le o(1).
\end{align*}
\end{lemma}


\begin{proof}[Proof of~Theorem~\ref{thm:one-sided} Assuming~\Cref{lem:reject-and-3,lem:ev1,lem:ev3}]
We upper bound the expression
\begin{align*}
\Prx_{\boldf \sim \Dno}\left[ \Alg(\boldf) = \text{``accept''}\right] &\mathop{\ge}^{(\ref{lem:reject-and-3})} \Prx_{\boldf \sim \Dno}\left[ \calE_2(\boldf)\right]
\geq \Prx_{\boldf \sim \Dno}\left[ \calE_1(\boldf)\right]-
\Prx_{\boldf \sim \Dno}\left[ \overline{\calE_2(\boldf)}\cap \calE_1(\boldf)\right]\ge 1-o(1)
\end{align*}
using ~\Cref{lem:ev1,lem:ev3}.
\end{proof}

\subsection{Proof of \Cref{lem:reject-and-3}}


Since $\Alg$ is a $q$-query deterministic algorithm which has one-sided error, in order for the algorithm to output ``reject,'' the set $Q^q$ queried by the root-to-leaf path obtained by executing $\Alg$ on $f$ must contain $x_1,\dots, x_{\ell}, y \in Q^q$ satisfying 
$$y \in \conv(x_1,\dots, x_{\ell}),\quad f(y) = 0,\quad \text{and}\quad f(x_1) = \dots = f(x_{\ell}) = 1.$$ In particular, from $y\in \conv(x_1,\dots,x_\ell)$, we must have that, for any vector $u \in \R^{2n}$, there exists a $j \in [\ell]$ such that $\langle x_j, u \rangle \geq \langle y, u \rangle$. This implies that:
\begin{flushleft}\begin{itemize}
\item We must have that all $x_1,\ldots,x_\ell$ satisfy $\|(x_i)_C\|_2 \leq \sqrt{n}$, and $\|y_C\|_2 \leq \sqrt{n}$, and this means these vectors lie in $\tilde{Q}^q$. The part of $\| (x_i)_C\|_2 \leq \sqrt{n}$ follows trivially from $f(x_1)=\cdots =f(x_\ell)=1$.
On the other hand, if $\|y_C\|_2>\sqrt{n}$, letting $u \in C$ be the unit vector $u=y_C/\|y_C\|_2$, there exists an $x_j$ 
  with $$\|(x_j)_C\|_2 \geq \langle x_j, u \rangle \geq \langle y, u \rangle =\|y_C\|_2> \sqrt{n},$$ and hence $f(x_j) = 0$, which would be a contradiction with $f(x_j)=1$.

\item We must have $y\notin B$ since $f(y)=0$.
As a result, there is a nonempty $T$ such that $y\in F_T$.
   In addition, $f(y) = 0$ implies that 
   there exists an $i\in T$ such that $y\in F_i$ but 
   $$
   \Indicator\sbra{\langle v^i,y\rangle \notin [-\sqrt{n}/2,\sqrt{n}/2]} = 0.
   $$
Given that $y\in F_i$, setting $u=g^i$, there exists an $x_j$ such that 
  $\langle x_j,g^i\rangle >\langle y,g^i\rangle\ge r$ and thus, $x_j\in F_i$.
It follows from $f(x_j)=1$ and the construction that 
  $$
   \Indicator\sbra{\langle v^i ,x_j\rangle \notin [-\sqrt{n}/2,\sqrt{n}/2]} = 1.
   $$  
\end{itemize}\end{flushleft}
This concludes the proof using $i,y$ and $x_j$.

\subsection{Proof of \Cref{lem:ev1}}
\label{sec:proofev1}



To prove \Cref{lem:ev1}, we introduce five new, easy-to-analyze events $\calE_{1,1},\calE_{1,2},
  \calE_{1,3}$, $\calE_{1, 4}$ and $\calE_{1,5}$, show that each happens with probability at least $1-o(1)$,
  and that $\calE_{1,1}\cap \calE_{1,2}\cap \calE_{1,3} \cap   \calE_{1,4} \cap \calE_{1, 5}$ implies $\calE_1$. For the $n$-dimensional subspace $C \subset \R^{2n}$ (in particular, the control subspace for $f$), we denote $\Shell(C) := \{ x \in \R^{2n} : \sqrt{n} - 100q \leq \| x_{C} \|_2 \leq \sqrt{n} \}$, where $x_C$ denotes the orthogonal projection of $x$ onto the subspace $C$. 
\begin{flushleft}\begin{itemize}
\item $\calE_{1,1}(f)$:
This event occurs if no query $x$ in $\Alg$ with $\|x_C\|_2 \leq \sqrt{n}$ lies in $\bigcup_{|T| \geq q} F_T$ defined by $f$;
\item $\calE_{1,2}(f)$: This event occurs if 
  no query $x$ in $\Alg$ satisfies 
  $x\notin \Shell(C)$  and $x\notin B$ (or equivalently, $\|x_C\|_2<\sqrt{n}-100q$
  and $x\in F_i$ for some $i\in [N]$);
\item $\calE_{1,3}(f)$: This event occurs if 
  no query $x$ in $\Alg$ with $\|x_C \|_2 \leq \sqrt{n}$ has $$\langle x,g^i\rangle\ge r+100qn^{1/4},\quad\text{for some $i\in [N]$};$$
\item $\calE_{1,4}(f)$: This event \emph{does not} occur if 
  there exist $i\in [N]$ and two queries $x,z$ in $\Alg$ where (i) $x_C$ and $z_C$ are not scalar multiples of each other, $z_C=(1+a)x_C+by$ denotes the unique decomposition with $x_C \perp y$, $\|y\|_2=1$ and $b>0$,
  such that $x\in F_i$ and $|\langle y,g^i\rangle |\ge 100\sqrt{q}$.
  \item $\calE_{1,5}(f)$: The event occurs whenever every pair $x,y \in \Alg$ satisfy $\| x - y \|_2 \leq 2 \| (x - y)_C \|_2$.
\end{itemize}\end{flushleft}

\ignore{This lemma follows from a simple union bound. Note there are at most $2^q$ queries $x_v$ with $v \in \Alg$. Thus, we have:
\begin{align*}
\Prx_{\boldf \sim \Dno}\left[ \calE_1(\boldf)\right] \leq \sum_{v \in \Alg} \Prx_{\boldf \sim \Dno}\left[ x_v \in \bD^q  \right] \mathop{\lsim}^{(\ref{lem:exp-decrease})} 2^q \cdot 1/q! = o(1). 
\end{align*}
{\color{red}I hope the second part of $\calE_1$\rnote{I somewhat lost the thread - which statement is meant by ``the second part of $\calE_1$''?}  follows from the discussion in
  the email thread on March 6.} \blue{Yup, I think this should be fine - it'll follow from the kind of calculations that will be in the proof of \Cref{lem:rad-by-rad}. Shivam and Rocco will fill this in.}

I think each of the four events follow from a relatively simple union bound. 
For example, for event $\calE_{1,4}$, with a fixed triple of $i,x,z$,
  the probability of the bad event happens over $\bg^i$ is roughly
$$
\frac{1}{N} \cdot 2^{-\text{some big constant}\cdot q}.
$$
This is because we can consider the draw of $\bg^i$ (the relevant part)
  as drawing two standard Gaussians along $x$ and along $y$.}

\noindent We first prove that $\calE_{1}(f)$ is implied by the five events together. Then, we show that each of the events holds individually with probability $1-o(1)$. By a union bound over the five events, this gives~\Cref{lem:ev1}.

\begin{lemma}
$\calE_{1,1}(f)\cap \calE_{1,2}(f)\cap \calE_{1,3}(f)\cap \calE_{1,4}(f) \cap \calE_{1,5}(f)$ implies
  $\calE_{1}(f)$.
\end{lemma}
\begin{proof}
Recall that $\tilde{Q}^q$ denotes the set of (at most $q$) queries made by $\Alg$ when running on $f$ whose orthogonal projections onto $C$ each have norm at most $\sqrt{n}$.
First $\calE_{1,1}(f)$ implies that the number of ``nonempty'' flaps $i\in [N]$, i.e.~flaps $F_i$ that have $\tilde{Q}^q\cap F_i\ne \emptyset$, is at most $q^2$.
Fix any nonempty flap $F_i$ and any two points $x,z\in \tilde{Q}^q\cap F_i$.
First consider the case that $x_C$ and $z_C$ are scalar multiples of each other. Note that we have $x,z\in \Shell(C)$ by $\calE_{1,2}(f)$ and thus, $\|(x-z)_C\|_2\le 100q$ (since they are scalar multiples of each other). By $\calE_{1,5}(f)$, $\|x - z\|_2 \leq 200q$, which is consistent with the requirement of $\calE_{1}(f)$ since $q = o(\sqrt{q}n^{1/4}).$
  
So consider the case when $x_C,z_C$ are not scalar multiples of each other, and let $z_C=(1+a)x_C+by$ be
  the unique decomposition with $x_C\perp y$ and $y \in C$ with $\|y\|_2=1$ and $b>0$.
Let $\alpha := \|(x-z)_C\|_2^2 = a^2\|x_C\|_2^2+b^2$.    Our goal is to establish that
\begin{equation} \label{eq:goal-norm-bound}
\alpha \leq  250000 q \sqrt{n},
\end{equation}
so that we may use $\calE_{1,5}(f)$ to deduce that \eqref{eq:iwantit} holds for $x$ and $z$.

We have $\|z_C\|_2^2 = (1+a)^2 \|x_C\|_2^2 +  b^2$.
Given that $\|z_C\|_2\le \sqrt{n}$,
\[ 
(1+2a)\|x_C\|_2^2 + \alpha = \|z_C\|_2^2 \leq n. 
\]
By $\calE_{1,2}(f)$, we have $x\in \Shell(C)$ and thus, $\|x_C\|_2\ge \sqrt{n}-100q$.
Plugging this in, we have
\[
(1+2a)(n-200q\sqrt{n}+10000q^2) + \alpha \leq n,
\]
or equivalently,
\begin{equation} \label{eq:aa}
\alpha \leq -2an +(1+2a)(200q\sqrt{n}-10000q^2)
\le 200q\sqrt{n}+a(-2n+400q\sqrt{n}-20000q^2). 
\end{equation}

Let's consider two cases:

\medskip
\noindent {\bf Case 1:} $a \geq -200q/\sqrt{n}.$  We have from  \Cref{eq:aa} (note that
  the coefficient of $a$ is negative and is a value larger than $-2n$) 
%
$$\|(x-z)_C\|_2^2 = \alpha \leq 200q\sqrt{n}+2n\cdot \frac{200q}{\sqrt{n}}= 600q\sqrt{n},$$ and we get \Cref{eq:goal-norm-bound}.

\medskip
\noindent {\bf Case 2:} $a < -200q/\sqrt{n}.$  In this case, using $r\le \langle z,g^i\rangle  ,\langle x,g^i\rangle\le r+100q n^{1/4}$ (where the first inequality is because $x,z \in F_i$ and the second is from $\calE_{1,3}(f)$) gives 
\begin{align*}
r &\leq \overbrace{(1+a)\cdot \langle x, g^i\rangle + b \cdot \langle y,g^i\rangle}^{=\langle z, g^i\rangle} \leq a\cdot \langle x,g^i\rangle+  \overbrace{{r+100qn^{1/4}}}^{\text{b/c~}\langle x ,g^i\rangle \le r+100q n^{1/4}}  + b\cdot  \overbrace{(100 \sqrt{q})}^{\text{b/c~}|\langle y,g^i\rangle | \leq 100\sqrt{q}}
\end{align*}
so (recall that $a<-200q/\sqrt{n}$ is negative and $-a$ is positive)
\[
b\ge \frac{-a\cdot \langle x,g^i\rangle-100qn^{1/4}}{100\sqrt{q}}\ge  \frac{-a r}{200\sqrt{q}}.
\]
Recalling that
$\|z_C\|_2^2=(1+a)^2 \|x_C\|_2^2 + b^2\le n$ and that $\|x_C\|_2^2 \geq n-200q\sqrt{n}+10000q^2$, we get
$$
n\ge (1+2a+a^2)(n-200q\sqrt{n})+ {\frac {a^2r^2}{40000 q}} 
\ge (1+2a)(n-200q\sqrt{n})+{\frac {a^2r^2}{40000 q }} 
$$
and hence,
\[
a^2 \cdot \frac {r^2}{40000 q}
\leq
200q\sqrt{n} -2a (n-200q\sqrt{n} ).
\]
Recalling that $a<0$, dividing through by $-a$ we get
\[
(-a) \cdot \frac{r^2}{40000q}
\le {\frac {200q\sqrt{n}}{-a}} + 2 (n - 200q\sqrt{n} )
\overbrace{\le }^{\text{~(using $-a\ge 200q/\sqrt{n}$)}}  3n.
\]
So we have
\[
0 <
-a \leq \frac{120000qn}{r^2} = \frac{60000q}{\sqrt{n}} \cdot  (1+o(1)),
\]
by the setting of $r$ in~\Cref{lemma:r-N-relationship}. Recalling \Cref{eq:a}, we get 
\[
\alpha \leq 200q\sqrt{n} - 2an \leq 200q\sqrt{n} + 240000q\sqrt{n},
\]
as was to be shown.
\end{proof}

\paragraph{Event $\calE_{1, 1}(f)$.} We now show that with probability at least $1 - o(1)$ over the draw of $\boldf \sim \Dno$, all $2^q$ queries specified by $\Alg$ avoid the region which is the intersection of at least $q$ flaps. Consider any fixed query $x$ and fix any setting of the control subspace $C \subset \R^{2n}$ with $\|x_C\|_2 \leq \sqrt{n}$. Using~\Cref{lemma:small-vol-high-degree} (and the fact $C$ is isomorphic to $\R^n$),
\begin{align*}
\Prx_{\bB \sim \Naz(r, N,C)}\left[x \in \bigcup_{|T| \geq q} \bF_T \right] \leq \frac{c_1^q}{q!},
\end{align*} 
so that a union bound over $2^q$ queries gives $(2c_1)^q / q! = o(1)$ for large $q$.  

\paragraph{Event $\calE_{1,2}(f)$.} Similarly to above, we proceed by a union bound over all $2^q$ queries. We consider a fixed control subspace $C$ and we let $x$ be a query with $\|x_C\|_2 < \sqrt{n} - 100q$, so
\begin{align*}
\Prx_{\bB \sim \Naz(r, N,C)}\left[ \exists i \in [N] : x \in \bF_i \right]  &\leq 
N \cdot \left( 1 - \Phi\left(\frac{r}{\sqrt{n}-100q} \right)\right)
\\ &\le N \cdot \left( 1 - \Phi\left(\frac{r}{\sqrt{n}} \left(1 + \frac{100q}{\sqrt{n}} \right)\right)\right) \\
			&\leq N \cdot \frac{\sqrt{n}}{r} \cdot \exp\left(-\frac{r^2}{2n} \left(1 + \frac{200 q}{\sqrt{n}}\right) \right) \\
			&= N \cdot \frac{\sqrt{n}}{r} \cdot \exp\left(-\frac{r^2}{2n} \right) \cdot \exp\left( -\frac{100r^2 q}{n^{3/2}}\right) \leq 2c_1 \cdot \exp\left(-10q\right),
\end{align*}
by the setting of $r$ from~\Cref{lemma:r-N-relationship}. The desired claim then follows from a union bound over all $2^q$ queries.

\paragraph{Event $\calE_{1,3}(f)$.} Consider any query $x$, and consider a fixed setting of the control subspace $C$ with $\|x_C \|_2 \leq \sqrt{n}$. Then, 
\begin{align*}
\Prx_{\bB \sim \Naz(r, N)}\left[\exists i \in [N] : \langle x, \bg^i \rangle \geq r + 100 q n^{1/4} \right]   
&\le N \left( 1 - \Phi\left( \frac{r+100qn^{1/4}}{\sqrt{n}} \right) \right)
\\
&\le N \left( 1 - \Phi\left( \frac{r}{\sqrt{n}} \left( 1 + \frac{100qn^{1/4}}{r}\right)\right) \right) \leq o(2^{-q}), 
\end{align*}
where the computation proceeds similarly to $\calE_{1,2}(f)$.

\paragraph{Event $\calE_{1,4}(f)$.} For a fixed control subspace $C$, we may consider two arbitrary queries $x, z$ among the set of all $2^q$ queries with $\|x_C\|_2, \|z_C\|_2 \leq \sqrt{n}$. This gives $2^{2q}$ possible settings of the unit vector $y$ which is orthogonal to $x_C$. In order for the event to fail, there must exists some $i \in [N]$ where $\langle \bg_i, x_C \rangle \geq r$ and $|\langle \bg_i, y \rangle| \geq 100\sqrt{q}$. Furthermore, since $x_C$ and $y$ are orthogonal, these two events are independent:
\begin{align*}
\Prx_{\bg^i}\left[ \langle \bg_i, x_C\rangle \geq r \wedge |\langle \bg_i, y\rangle | \geq 100 \sqrt{q}\right] \leq \frac{c_1}{N} \cdot e^{-100^2 q / 2}, 
\end{align*}
hence, we may take a union bound over all $i \in [N]$ and all $2^{2q}$ pairs of vectors $x$ and $z$.

\paragraph{Event $\calE_{1,5}(f)$.} Finally, consider any two vectors $x$ and $y$ which are queries among the $2^q$ possible queries in $\Alg$. The Johnson-Lindenstrauss lemma (see Theorem~5.3.1 in~\cite{vershynin2018high}) says that a random $n$-dimensional subspace $\bC$ of $\R^{2n}$ will satisfy $\| (x-y)_{\bC} \|_2 \geq (1/\sqrt{2} - \eps) \|x-y\|_2$ except with probability $\exp\left(-\Omega(\eps^2 n)\right)$. Thus, for large enough $n$, $\|x - y\|_2 \leq 2 \|(x-y)_{\bC}\|_2$ except with probability $\exp\left(-\Omega(n)\right)$, and since $q \ll n$, we may union bound over all $2^{2q}$ pairs of queries in $\Alg$.

\subsubsection{Proof of~\Cref{lem:ev3}}\label{section:final}


For $\overline{\calE_2(f)}\cap \calE_1(f)$ to happen, there must exist a level
  $k\in [q]$ such that
\begin{flushleft}\begin{itemize}
\item After the the first $k-1$ queries $\tilde{Q}=\tilde{Q}^{k-1}$, $\calE_1(f)$ holds, i.e., the number of
  flaps $F_i$ with $\tilde{Q} \cap F_i\ne\emptyset$ is at most $q^2$. In every such $F_i$, every two points in $\tilde{Q} \cap F_i$ have distance at most
  $1000 \sqrt{q}n^{1/4}$ and share the same value of 
  $$
  \Indicator\sbra{\langle v^i ,x \rangle \notin [-\sqrt{n}/2,\sqrt{n}/2]},
  $$
  which we denote by $b_i\in \{0,1\}$.
\item Let $y$ be the $k$-th query.
There exists an $i$ such that $\tilde{Q}\cap F_i\ne \emptyset$ and $y\in F_i$ such that  
\begin{equation}\label{hehe1}
\|y-x\|_2\le 1000 \sqrt{q}n^{1/4}
\end{equation}
for all $x\in \tilde{Q}\cap F_i$ (the number of such $i$ is at most $q$) but
$$
\Indicator\sbra{\langle v^i ,x \rangle \notin [-\sqrt{n}/2,\sqrt{n}/2]}\ne b_i.
$$
\end{itemize}\end{flushleft}

We prove below that when $q\leq n^{0.05}$, the probability of the event above 
  for a fixed $k$ is $o(1/q)$ and thus,~\Cref{lem:ev3} follows by applying a union bound on $k$.
This follows from a union bound on all $i\in [N]$ such that 
  $\tilde{Q}\cap F_i\ne \emptyset$ and every $x\in \tilde{Q}\cap F_i$ satisfies (\Cref{hehe1}), taking $X$ (or $z$) below as $\tilde{Q}\cap F_i$ (or $y$, respectively)
  projected on the space orthogonal to $g^i$
  and $b$ as $b_i$.

\begin{lemma}\label{lem:final-lemma-sim}
Let $b\in \{0,1\}$, and let $X$ be a set of at most $q$ points in $\Ball(\sqrt{2n})$ and $y\in \Ball(\sqrt{2n})$. Suppose that every pair $x \in X$ and $y$ satisfy
\[ \| (x-y)_A \|_2 \leq 1000 \sqrt{q} n^{1/4}. \]
Then over the draw~of $\bv \sim N(0, I_n)$ in $A$, the probability of
  $\Indicator[\langle  \bv,y\rangle \notin[-\sqrt{n}/2,\sqrt{n}/2]]\ne b$ conditioning on the event that
  $\Indicator[\langle \bv,x\rangle \notin[-\sqrt{n}/2,\sqrt{n}/2]]=b$ for all $x\in X$ is at most $O(\sqrt{q} \log n/n^{1/4})$.
\end{lemma}

Before proving~\Cref{lem:final-lemma-sim}, we show how it implies~\Cref{lem:ev3} from a union bound. In particular, we have concluded that
\begin{align*}
\Prx_{\boldf \sim \Dno}\left[ \ol{\calE_2(\boldf)} \cap \calE_1(\boldf) \right] \leq \overbrace{q}^{\text{u.b over $k \in [q]$}} \times \overbrace{q^2}^{\text{u.b over $i$}} \times O\left(\sqrt{q} \log n / n^{1/4} \right) = O(q^{3.5} \log n / n^{1/4}).
\end{align*} 

\begin{proof}[Proof of~\Cref{lem:final-lemma-sim}]
Fix a point $x^*\in X$. We show that (1) the probability of $\Indicator[\langle \bv,x^*\rangle \notin[-\sqrt{n}/2,\sqrt{n}/2]]=b$
  for all $x$ is at least $\Omega(1)$; and (2) the probability of 
$$
\Indicator[\langle  \bv,y\rangle \notin[-\sqrt{n}/2,\sqrt{n}/2]]\ne \Indicator[\langle  \bv,x^*\rangle \notin[-\sqrt{n}/2,\sqrt{n}/2]]
$$
is at most $O(\sqrt{q} \log n / n^{1/4})$. The lemma then follows.

To analyze (2), we consider any $0 < \gamma \leq \sqrt{n}/4$ and any sign $\xi \in \{-1,1\}$. We have that a draw of a Gaussian $\bv$ lying in the action subspace $A$ satisfies 
\begin{align*}
\Prx_{\bv}\left[ \langle \bv, x^*\rangle \in [\xi\sqrt{n}/2 - \gamma, \xi\sqrt{n}/2 + \gamma] \right] &= \Prx_{\bg \sim N(0, 1)}\left[ \bg \in \left[\frac{\xi\sqrt{n}}{2\|x^*_A\|_2} - \frac{\gamma}{\|x^*_A\|_2}, \frac{\xi\sqrt{n}}{2\|x_A^*\|_2} + \frac{\gamma}{\|x^*_A\|_2} \right] \right] \\
&\leq \min_{\tau > 0} \left\{ \frac{2\gamma}{\tau}, \frac{4\tau}{n} \cdot e^{-n / 8\tau^2} \right\},
\end{align*}
where we have used Gaussian anti-concentration (to conclude it does not lie within an interval of width $2\gamma / \|x^*_A\|_2$), as well as Gaussian tail-bounds to say $\bg$ is is larger than $\sqrt{n} / (4 \|x_A^*\|_2)$. The minimum over $\tau > 0$ is meant to quantify over possible values of $\|x_A^*\|_2$. Letting $\gamma = 50000 \sqrt{q} n^{1/4} \log n$ allows us to conclude that the probability that $\langle \bv, x^*\rangle$ lies within distance $\gamma$ of $-\sqrt{n}/2$ or $\sqrt{n}/2$ is at most $O(\sqrt{q} \log n / n^{1/4})$. 

On the other hand, given that $\|(y-x^*)_A\|\le 1000\sqrt{q}n^{1/4}$, we also have that
\begin{align*}
\Prx_{\bv}\left[ \left| \langle \bv, y - x^*\rangle \right| \geq 1000 \sqrt{q} n^{1/4} \log n \right] \leq \Prx_{\bg \sim N(0, 1)}\left[ |\bg| \geq \log n\right], 
\end{align*}
which is smaller than any inverse polynomial in $n$. Therefore, we have that, except with probability at most $O(\sqrt{q} \log n / n^{1/4})$, for both $\xi \in \{-1,1\}$,
\begin{itemize}
\item $\left| \langle \bv, x^* \rangle - \xi \cdot \frac{\sqrt{n}}{2} \right| \leq 50000 \sqrt{q} n^{1/4} \log n$; and 
\item $|\langle \bv, x^* - y\rangle| \leq 1000 \sqrt{q} n^{1/4} \log n$.  
\end{itemize}
When these two events occur, event (2) cannot occur, which shows that (2) occurs with probability at most $O(\sqrt{q} \log n / n^{1/4})$. 

To conclude (1), we note that any $x^*$ with $\|x^*_A\|_2 \leq \sqrt{2n}$ satisfies 
\begin{align*}
\Prx_{\bv}\left[ \langle \bv, x^*\rangle \in [-\sqrt{n}/2, \sqrt{n}/2] \right] = \Prx_{\bg \sim N(0,1)}\sbra{ \bg \in \sbra{-\frac{1}{2\sqrt{2}}, \frac{1}{2\sqrt{2}}}} = \Omega(1),
\end{align*}
which shows that conditioning on event (1) does not significantly affect the probability of (2). 
\ignore{it follows from a union bound $x\in X$ that 
$$
\left|\langle\bv,x-x^*\rangle\right|> 100000\sqrt{q}n^{1/4}\log n
$$
for some $x\in X$
with probability at most $o(1)$.
Given that (1) happens when the two events above do not occur, we have that (1) holds with probability at least $\Omega(1)$.

The analysis of (2) is similar. 
(2) happens when either $\langle \bv,x^*\rangle$  lies the union of  () and ()
  or $\langle \bv, y-x^*\rangle$ satisfies the same inequality ().
So  (2) happens with probability at most $O(\sqrt{q}\log n/n^{1/4})$.}
\end{proof}

%% file: sections/tolerant.tex
\section{A Mildly-Exponential Lower Bound for Non-Adaptive Tolerant Testers}
\label{sec:tolerant}

We will prove the following:

\tolerantthm*

\subsection{The $\Dyes$ and $\Dno$ Distributions}
\label{subsec:tolerant-yes-no-distributions}

Before specifying the $\Dyes$ and $\Dno$ distributions, we first describe some necessary objects. 

\paragraph{The Control and Action Subspaces.}
Throughout, we will work over $\R^{n+1}$ for convenience. Let $\bA$ denote a random 1-dimensional subspace of $\R^{n+1}$, i.e. 
\[
	\bA=\{t\bv: t \in \R\}~\text{where}~\bv\sim\S^{n}~\text{is a Haar-random unit vector.}
\]
Let $\bC$ be the orthogonal complement of $\bA$; note that $\bC$ is a random $n$-dimensional subspace of $\R^{n+1}$.  We call $\bC$ the \emph{control subspace} and we call $\bA$ the \emph{action subspace}.  

Given a vector $x \in \R^n$, we write $x_{\bC}$ to denote the projection of $x$ onto $\bC$ and we write $x_{\bA}$ to denote the projection of $x$ onto $\bA$, so every vector satisfies $x = x_{\bA} + x_{\bC}$.
Recalling that $\bA$ is a 1-dimensional subspace, when there is no risk of confusion we write $x_{\bA}$ to denote the scalar value $t$ such that $x_{\bA}=t \bv$.

\paragraph{Constants.} 
We will use four positive absolute constants $c_0,c_1,c_2$ and $\tau$ in the construction.
Here $c_0$ is the constant hidden in the statement of \Cref{lemma:unique-expected-volume-lb-adaptive}.
We set $c_1,c_2$ and $\tau$ as follows:
\begin{align}\label{eq:constants}
c_1=\frac{1}{100} \quad \text{and}\quad c_2=\tau= \frac{c_0c_1}{100} 
\end{align}
so that \Cref{eq:make-me-constant} at the end of \Cref{subsec:tolerant-distance-to-convexity} is $\Omega(1)$.

\paragraph{Nazarov's Body on the Control Subspace.} 
Let $N := 2^{\sqrt{n}}$. 
We take $r$ to satisfy~\Cref{eq:r-condition} for the absolute constant $c_1$ given in \Cref{eq:constants}, and draw $\bB\sim\Naz(r,N,\bC)$ (cf.~\Cref{def:nazarov-body}) where $\Naz(r,N,\bC)$ is as defined in \Cref{def:nazarov-body} but with the $n$-dimensional control subspace $\bC$ playing the role of $\R^n$. (This notation is as in~\Cref{sec:adaptive-setup} where we write $\Naz(r,N,\bC)$ to mean the distribution $\Naz(r,N)$ over bodies in $\bC$.) Similar to~\Cref{sec:one-sided-adaptive},  $\bB \sse \R^{n+1}$ is a $\bC$-subspace junta. Note in particular that a draw of $\bB\sim\Naz(r,N,\bC)$ immediately specifies $\bg^i, \bH_i, \bF_i, \bU_i$ for $i\in[N]$ (and that the sets $\bH_i, \bF_i, \bU_i \sse \R^{n+1}$ are $\bC$-subspace juntas as well). 

\paragraph{Functions on the Action Subspace.} 

Let $c_2$ be the absolute constant given in \Cref{eq:constants}. Intuitively, $2c_2$ will be the Gaussian measure of two symmetric intervals in the action subspace $\bA$. More formally we define
\[
	\curb := \sbra{\Phi^{-1}\pbra{{\frac {1-2c_2} 3}},
		\Phi^{-1}\pbra{ {\frac {1+c_2} 3}}} \cup 
		\sbra{\Phi^{-1}\pbra{{\frac {2-c_2} 3}},
		\Phi^{-1}\pbra{{\frac {2+2c_2} 3}}
	}.
\]
Using the absolute constant $\tau$ given in \Cref{eq:constants},
we also define the interval 
\[
	I := \sbra{\sqrt{n+1} -  \sqrt{2\ln(2/\tau)},\sqrt{n+1}+\sqrt{2\ln(2/\tau)}},
\]
and we observe that a random draw of $\bx$ from $N(0,I_{n+1})$ has $\|\bx\| \in I$ with probability at least $1-\tau$.  
We write $\Shell_{n+1}$ to denote the corresponding spherical shell in $\R^{n+1}$, i.e.
\[
	\Shell_{n+1} = \big\{x \in \R^{n+1}: \|x\| \in I\big\}.
\]
Finally, let $\bP$ be a uniformly random subset of $[N]$. 

For a fixed setting of $C$ (which defines the complementary $A$), $B$ (which in turn defines $H_i$, $F_i$, the $F_T$'s, and $U_i$), and $P$, we define the function $g_{C,B,P}: \R^{n+1} \to \{0,1,0^\star,1^\star\}$ as follows:

\[
	g_{C,B,P}(x) = 
	\begin{cases}
		0 & \text{if~} \|x_C\| \geq \sqrt{n} \text{~or~}x \in \bigcup_{|T|\geq 2} F_T \text{~or~} x \notin \Shell_{n+1},\\
		1 & \text{if~} x \in \Shell_{n+1} \text{~and~} x \in B,\\
		0 & \text{if~} x \in \Shell_{n+1} \text{~and~} x \in U_i  \text{~ for some $i\in [N]$ and~} x_{A} \in \curb,\\ 
		0^\star & \text{if~} x \in \Shell_{n+1} \text{~and~} x \in U_i \text{~for some~}i\in P\text{~and~} x_{A} \notin \curb,\\
		1^\star & \text{if~} x \in \Shell_{n+1} \text{~and~} x \in U_i \text{~for some~}i\notin P\text{~and~} x_{A}\notin \curb.
	\end{cases}
\]

\begin{figure}[t]
	\centering
	\begin{tikzpicture}[x=0.75pt,y=0.75pt,yscale=-0.9,xscale=0.9]

		\def\A{(340, 235) circle (130)};
		\def\B{(340, 235) circle (110)};
	
		\begin{scope}
			\clip \A;
			\fill[color=green, opacity=0.1] \A;
			\fill[color=blue!20] (351.5,56) to (523.5,302)|- cycle;
			\fill[color=blue!20] (462.5,91) to (206.5,168)|- cycle;
			\fill[color=blue!20] (267.5,89) to (203.5,341)|- cycle;
			\fill[color=blue!20] (480.5,379) to (187,293)|- cycle;
			\fill[color=blue!20] (321,388) to (514.5,245)|- cycle;
		\end{scope}
		
		\begin{scope}
			\clip \A;
			\clip (462.5,91) to (206.5,168)|- cycle;
			\fill[color=white] (267.5,89) -- (203.5,341)|- cycle;
		\end{scope}
		
		\begin{scope}
			\clip \A;
			\clip (480.5,379) to (187,293)|- cycle;
			\fill[color=white] (321,388) to (514.5,245)|- cycle;
		\end{scope}
			
		\draw \A;
		\filldraw[fill=white, dashed] \B;
		\draw (267.5,89) -- (203.5,341); 
		\draw (480.5,379) -- (187,293); 
		\draw (514.5,245) -- (321,388); 
		\draw (462.5,91) -- (206.5,168); 
		\draw (523.5,302) -- (351.5,56);	 
		
		\node[fill, inner sep=1pt,circle] (ori) at (340,235) {};
		\node (oriname) at (340, 225) {\small ~$0^n$};
		\draw (ori) -- (222,290);
		\node (sqrtn) at (270,255) {\small $\sqrt{n}$};
		\draw (ori) -- (314,330);
		\node (rad) at (340,290) {\small $\frac{r}{\|\bg_i\|}$};		
		
		\node (Hi) at (495, 382) {\small $\bH_i$};
		\node (Ui) at (300, 348) {\small $\bU_i$};
		\node (Hj) at (475, 92) {\small $\bH_j$};
		\node (Uj) at (300, 122) {\small $\bU_j$};
		
		\node (Rn) at (600, 235) {$\bC \cong \R^n$};
		\node (phantomRn) at (600, 235) {\phantom{$\bC \cong \R^n$}};

		\tikzset{
		    overdraw/.style={preaction={draw,pink,line width=#1}},
		    overdraw/.default=5pt
		}

		\draw[latex-latex] (225, -10) -- (550, -10);
		\node (action-top) at (580, -10) {$\bA\cong\R$};
		\node (j-in-P) at (387.5, 20) {$j \in \bP$};
		\filldraw[overdraw] (300, -10) -- (290, -10);
		\filldraw[overdraw] (475, -10) -- (485, -10);	
		\draw (300, -5) -- (300, -15);
		\draw (290, -5) -- (290, -15);
		\draw (475, -5) -- (475, -15);
		\draw (485, -5) -- (485, -15);	
		\node (top-m) at (387.5, -20) {$1$};
		\node (top-r) at (512.5, -20) {$1$};
		\node (top-l) at (262.5, -20) {$1$};
		\draw[-latex] (Uj) [out=-40]to ([xshift=-20]j-in-P);
		
		
		\draw[latex-latex] (455, 480) -- (130, 480);
		\node (action-bot) at (100, 480) {$\bA\cong\R$};
		\node (i-in-P) at (277.5, 510) {$i \notin \bP$};
		\filldraw[overdraw] (190, 480) -- (180, 480);
		\filldraw[overdraw] (365, 480) -- (375, 480);	
		\draw (190, 475) -- (190, 485);
		\draw (180, 475) -- (180, 485);
		\draw (365, 475) -- (365, 485);
		\draw (375, 475) -- (375, 485);	
		\node (bot-m) at (277.5, 470) {$0$};
		\node (bot-r) at (402.5, 470) {$0$};
		\node (bot-l) at (152.5, 470) {$0$};
		\draw[-latex] (Ui) [out=180,in=270]to ([yshift=-20pt]bot-m);
		
		\filldraw[black, fill=green, fill opacity=0.1] (590, 350) rectangle (610, 370);
		\node (greeeeen) at (660, 360) {\small $\bfyes(x) = 1$};
		
	\end{tikzpicture}
	
	\
	
	\caption{A depiction of $\Dyes$. We identify the control subspace $\bC \cong \R^n$. The annulus defined by the boundary of $\Ball(\sqrt{n})$ and the dotted circle 
corresponds to points $x$ which satisfy $x \in \Shell_{n+1}$ and $\|x_{\bC}\| \leq \sqrt{n}$. 
	Finally, the red region in the action subspace $\bA \cong \R$ corresponds to $\curb$.}	
	\label{fig:tolerant-yes}
\end{figure}
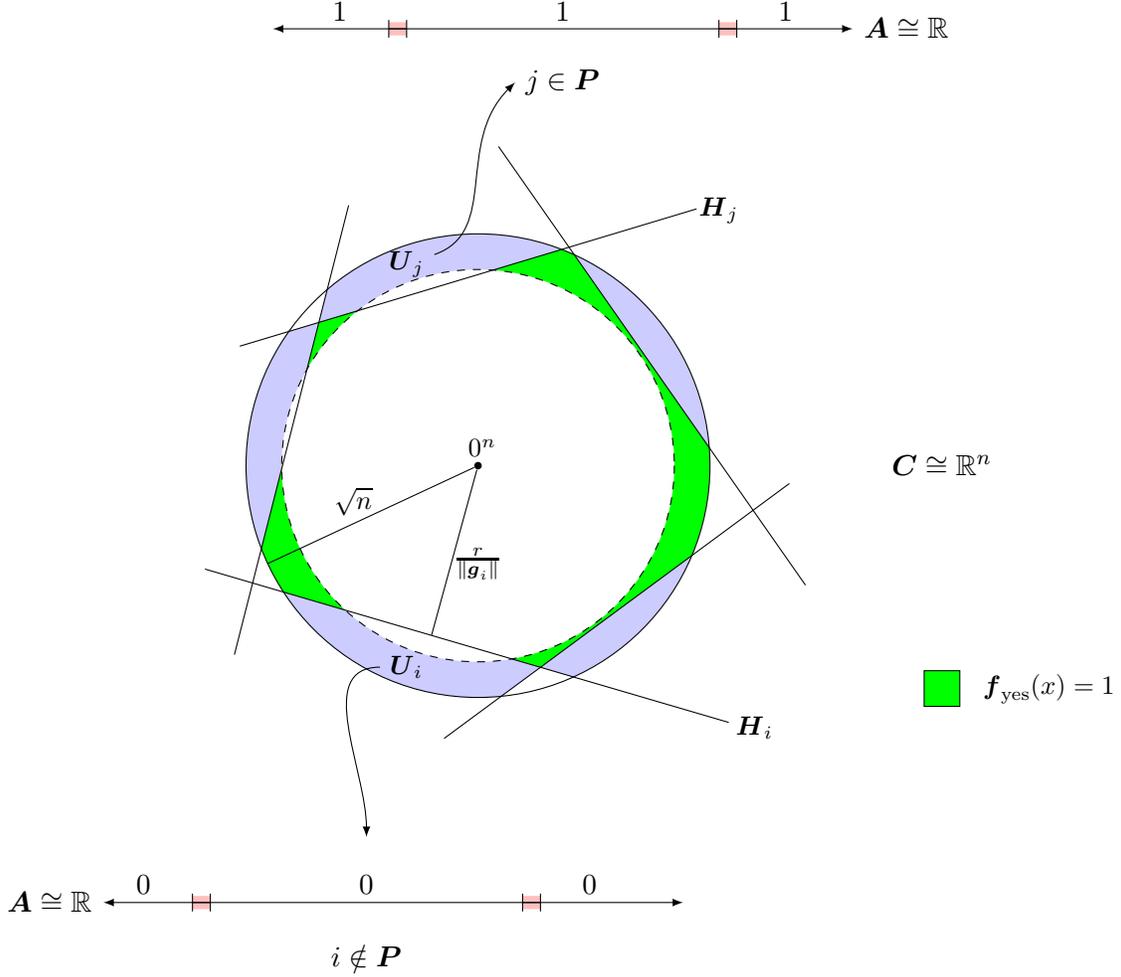

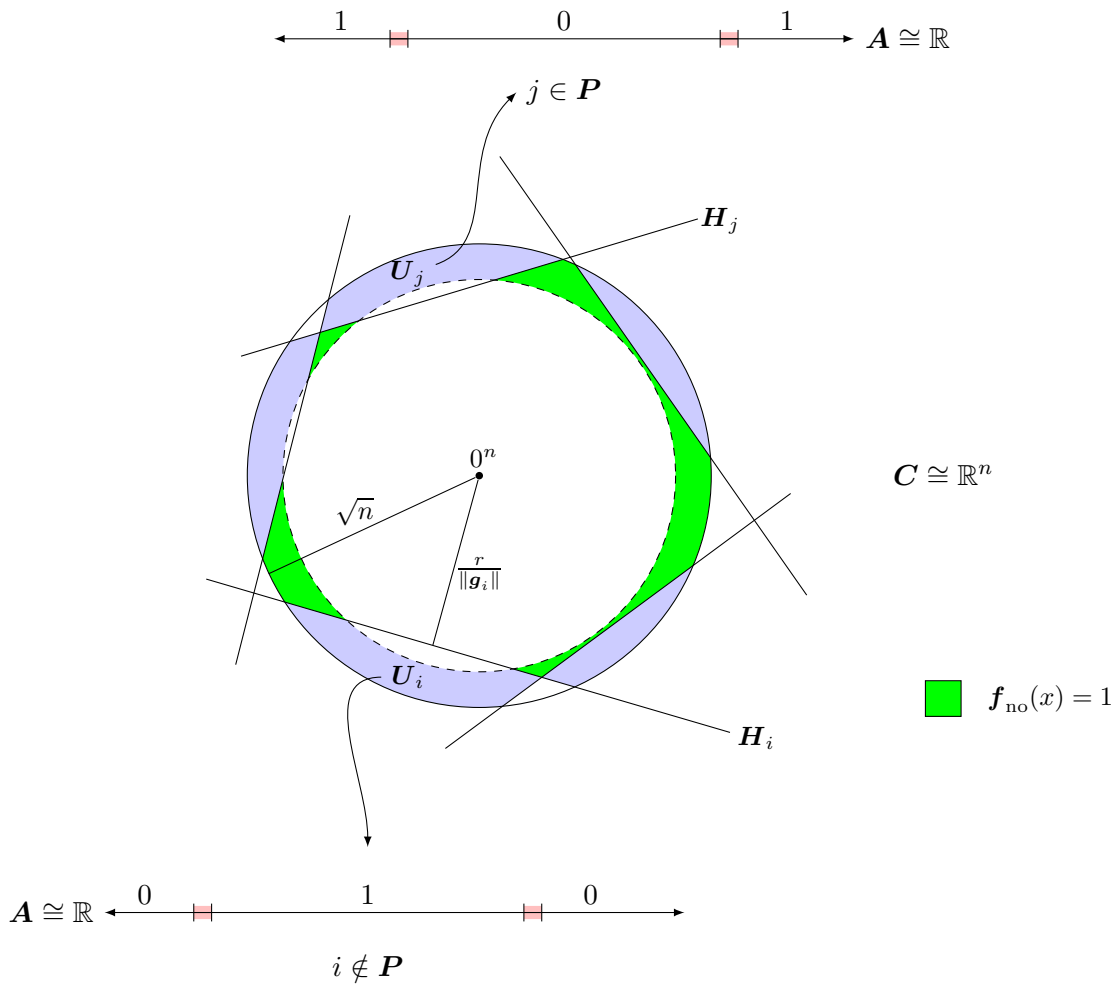
\begin{figure}[t]
	\centering
	\begin{tikzpicture}[x=0.75pt,y=0.75pt,yscale=-0.9,xscale=0.9]

		\def\A{(340, 235) circle (130)};
		\def\B{(340, 235) circle (110)};
	
		\begin{scope}
			\clip \A;
			\fill[color=green, opacity=0.1] \A;
			\fill[color=blue!20] (351.5,56) to (523.5,302)|- cycle;
			\fill[color=blue!20] (462.5,91) to (206.5,168)|- cycle;
			\fill[color=blue!20] (267.5,89) to (203.5,341)|- cycle;
			\fill[color=blue!20] (480.5,379) to (187,293)|- cycle;
			\fill[color=blue!20] (321,388) to (514.5,245)|- cycle;
		\end{scope}
		
		\begin{scope}
			\clip \A;
			\clip (462.5,91) to (206.5,168)|- cycle;
			\fill[color=white] (267.5,89) -- (203.5,341)|- cycle;
		\end{scope}
		
		\begin{scope}
			\clip \A;
			\clip (480.5,379) to (187,293)|- cycle;
			\fill[color=white] (321,388) to (514.5,245)|- cycle;
		\end{scope}
			
		\draw \A;
		\filldraw[fill=white, dashed] \B;
		\draw (267.5,89) -- (203.5,341); 
		\draw (480.5,379) -- (187,293); 
		\draw (514.5,245) -- (321,388); 
		\draw (462.5,91) -- (206.5,168); 
		\draw (523.5,302) -- (351.5,56);	 
		
		\node[fill, inner sep=1pt,circle] (ori) at (340,235) {};
		\node (oriname) at (340, 225) {\small ~$0^n$};
		\draw (ori) -- (222,290);
		\node (sqrtn) at (270,255) {\small $\sqrt{n}$};
		\draw (ori) -- (314,330);
		\node (rad) at (340,290) {\small $\frac{r}{\|\bg_i\|}$};		
		
		\node (Hi) at (495, 382) {\small $\bH_i$};
		\node (Ui) at (300, 348) {\small $\bU_i$};
		\node (Hj) at (475, 92) {\small $\bH_j$};
		\node (Uj) at (300, 122) {\small $\bU_j$};
		
		\node (Rn) at (600, 235) {$\bC \cong \R^n$};
		\node (phantomRn) at (600, 235) {\phantom{$\bC \cong \R^n$}};

		\tikzset{
		    overdraw/.style={preaction={draw,pink,line width=#1}},
		    overdraw/.default=5pt
		}

		\draw[latex-latex] (225, -10) -- (550, -10);
		\node (action-top) at (580, -10) {$\bA\cong\R$};
		\node (j-in-P) at (387.5, 20) {$j \in \bP$};
		\filldraw[overdraw] (300, -10) -- (290, -10);
		\filldraw[overdraw] (475, -10) -- (485, -10);	
		\draw (300, -5) -- (300, -15);
		\draw (290, -5) -- (290, -15);
		\draw (475, -5) -- (475, -15);
		\draw (485, -5) -- (485, -15);	
		\node (top-m) at (387.5, -20) {$0$};
		\node (top-r) at (512.5, -20) {$1$};
		\node (top-l) at (262.5, -20) {$1$};
		\draw[-latex] (Uj) [out=-40]to ([xshift=-20]j-in-P);
		
		\draw[latex-latex] (455, 480) -- (130, 480);
		\node (action-bot) at (100, 480) {$\bA\cong\R$};
		\node (i-in-P) at (277.5, 510) {$i \notin \bP$};
		\filldraw[overdraw] (190, 480) -- (180, 480);
		\filldraw[overdraw] (365, 480) -- (375, 480);	
		\draw (190, 475) -- (190, 485);
		\draw (180, 475) -- (180, 485);
		\draw (365, 475) -- (365, 485);
		\draw (375, 475) -- (375, 485);	
		\node (bot-m) at (277.5, 470) {$1$};
		\node (bot-r) at (402.5, 470) {$0$};
		\node (bot-l) at (152.5, 470) {$0$};
		\draw[-latex] (Ui) [out=180,in=270]to ([yshift=-20pt]bot-m);
		
		\filldraw[black, fill=green, fill opacity=0.1] (590, 350) rectangle (610, 370);
		\node (greeeeen) at (660, 360) {\small $\bfno(x) = 1$};
	\end{tikzpicture}
	
	\
	
	\caption{A depiction of $\Dno$. Our conventions are as in~\Cref{fig:tolerant-yes}.}	
	\label{fig:tolerant-no}
\end{figure}

\paragraph{The $\Dyes$ and $\Dno$ Distributions.}
To sample a set from either $\Dyes$ or $\Dno$, first draw $\bC, \bB, \bP$ as described above; note that this induces a draw of $\bA, \bH_i, \bF_i,$ the $\bF_T$'s, and $\bU_i$. Draws from $\Dyes$ and $\Dno$ are identical on points $x \in \R^{n+1}$ where $g_{\bC,\bB,\bP}(x) \in \{0,1\}$; on the other values, however, 
\begin{itemize}
	\item For functions in $\Dyes$, we set $0^\star \mapsto 0$ and $1^\star \mapsto 1$.
	\item For functions in $\Dno$, we set
	\[
		0^{\star} \mapsto \Indicator\cbra{x_A\notin\road} \qquad\text{and}\qquad 1^\star \mapsto \Indicator\cbra{x_A\in\road}
	\]
	where we define 
	\[
		\road := \pbra{\Phi^{-1}\pbra{{\frac {1+c_2} 3}}, \Phi^{-1}\pbra{{\frac {2-c_2} 3}}} \subset \R.
	\]
\end{itemize}
See~\Cref{fig:tolerant-yes,fig:tolerant-no} for illustrations of $\Dyes$ and $\Dno$.

\input{sections/anindya-distance-to-convexity}

\input{sections/anindya-indistinguishability}

%% file: sections/anindya-distance-to-convexity.tex

\subsection{Distance to Convexity}
\label{subsec:tolerant-distance-to-convexity}

Recall that a draw of a function from either $\Dyes$ or $\Dno$ induces a draw of $\bB, \bC,$ and $\bP$. 
First, we give an upper bound on the expected distance to convexity of a function drawn from $\Dyes$. 
Let $\eps_1$ be the constant given by
$$
\eps_1:=2c_2+\tau+2\cdot \Ex_{\bB\sim \Naz(r,N)}\sbra{
\Vol\left(\bigcup_{|T|\geq 2} \bF_T\right)}.
$$

\begin{proposition} \label{prop:tolerant-yes-distance}
	We have $\dist(\bfyes, \Pconv) \leq \eps_1$ with probability at least $0.5$ when 
	  $\bfyes\sim \Dyes$.
\end{proposition}


\begin{proof} 
Consider a fixed choice of $C,B$ and $P$. 
We then consider the convex set $G_{C, B, P}' \subset \R^{n+1}$ defined as the intersection of 
$H_i$ (for $i\in P$) and the set $\{x: \Vert x_C \Vert_2 \leq \sqrt{n}\}$. By construction, the set $G_{C, B, P}'$ is a convex set. Let $g_{C, B, P}'$ denote the corresponding indicator function. We now analyze the distance between the functions $g_{C, B, P}$ (where, as for functions in the support of $\Dyes$, we identify $0^\star$ with $0$ and $1^\star$ with $1$) and $g_{C, B, P}'$. 

First of all, by construction, if $g_{C, B, P}(x)=1$, then  $g_{C,B,P}'(x)$ is also $1$. So to bound the distance, we note that there are three possible ways in which $g_{C,B,P}(x)$ can be $0$ but $g_{C,B,P}'(x)$ can be $1$: 
\begin{enumerate}
\item $x \in \cup_{|T| \ge 2} F_T$;
\item $x \not \in \Shell_{n+1}$;
\item There is some $i \in [N]$ such that $x \in U_i$ and $x_A \in \curb$. 
\end{enumerate}
By definition, (i) the Gaussian volume of the first set is $\Vol( \cup_{\substack{|T|\geq 2}} F_T)$; (ii) the Gaussian volume of the second set is bounded by $\tau$; (iii) the Gaussian volume of the third set is bounded by $\Vol(\curb)$ which by definition is $2c_2$. Thus, for a specific instantiation of $C$, $B$ and $P$, 
\[
	\dist(g_{C,B,P}, \Pconv) \le 2c_2 + \tau + \Vol\pbra{\bigcup_{\substack{|T|\geq 2}} F_T}. 
\]
The claim follows from Markov's inequality, that the last term on the RHS above is at most twice the 
  expectation with probability at least $1/2$. 
\end{proof}

\noindent Next, we give a lower bound on the expected distance to convexity of a function drawn from $\Dno$.
Let $\eps_2$ be the constant given as 
$$
\eps_2:=\pbra{\frac{1-2c_2}{3}}\pbra{0.3\cdot \Ex_{\bB\sim \Naz(r,N)}\left[\Vol{\left(\bigsqcup_{i=1}^N \bU_i\right)}\right] - \frac{\tau}{2}\
		}.
$$

\begin{proposition} \label{prop:tolerant-no-distance}
	We have $\dist(\bfno, \Pconv) \geq\eps_2$ with probability at least $1-o(1)$
	when $\bfno\sim \Dno$.
\end{proposition}

\begin{proof}
	As before, consider a fixed choice of $C, B$ and $P$. 
Take any $x_C$ such that $x_C\in U_i$ for some $i\in P$ and $\|x_C\|\ge \sqrt{n+1}-\sqrt{2\ln(2/\tau)}$
  (which is the left end of $I$).
For any such $x_C$, the line along $x_A$ looks like the picture at the top of 	
  \Cref{fig:tolerant-no}, except that the {function} is set to $0$ when $\|x_C\|$ makes $x$ go outside of $\Shell_{n+1}$.
Given that $\|x_C\|\le \sqrt{n}$, this only happens when $|x_A|$ is at least $\Omega(n^{1/4})$ 
  given that $\tau\le 1/10000$ in \Cref{eq:constants}.

As a result, the distance to convexity along this line in the action space (with $x_C$ fixed)
  is at least $\Vol(\road)=(1-2c_2)/3$.	
This follows immediately from the fact that $\road$ is a symmetric interval about $0$ and the choice of parameters in the definition of $\road$. 

Given that the mass of $x$ with $\|x_C\|< \sqrt{n+1}-\sqrt{2\ln(2/\tau)}$ is at most $\tau/2$, 
it follows that 
	\begin{align*}
		\dist(\fno, \Pconv)
		 \geq \pbra{\pbra{\sum_{i \in P} \Vol\pbra{U_i}} - \frac{\tau}{2}}\cdot\pbra{\frac{1-2c_2}{3}}   
	\end{align*}
The result follows by a straight forward modification of \Cref{lemma:concentration} to show that with
  probability at least $1-o(1)$, we have 
  $\sum_{i\in P} \vol(U_i)$ is at least $0.3\cdot \Ex[\vol(\sqcup_{i\in [N]} \bU_i)]$
    when $\bB\sim \Naz(r,N)$.
\end{proof}

\paragraph{Setting Parameters.} 
We verify that $\eps_2-\eps_1=\Omega(1)$:
\begin{align}
	\eps_2 - \eps_1 
	&\geq \pbra{\frac{1-2c_2}{3}}\pbra{0.3\cdot  {\EVol{\bigsqcup_{i=1}^N \bU_i}}  - \frac{\tau}{2}} - {2c_2 - \tau - 2\cdot \EVol{\bigcup_{|T|\geq 2} \bF_T}} \nonumber \\
	&\geq \EVol{\bigsqcup_{i=1}^{N} \bU_i }\pbra{\frac{1-2c_2}{10} -   \frac{ c_1}{1-c_1} } - 2c_2 - \tau\pbra{\frac{7-2c_2}{6}} \tag{\Cref{lemma:flaps-vs-dog-ears}} \nonumber \\
	 &\geq c_0c_1\pbra{\frac{1-2c_2}{10} -  {{ } \frac{c_1}{1-c_1} }} - 2c_2 - \frac{7\tau}{6} ,\label{eq:make-me-constant}
\end{align}
(where the last line is by \Cref{lemma:unique-expected-volume-lb-adaptive})
which is $\Omega(1)$ given choices 
  of $c_0,c_1,c_2$ and $\tau$ made in \Cref{eq:constants}.

%% file: sections/anindya-indistinguishability.tex

\subsection{Proof of Theorem~\ref{thm:tolerant}}
\label{subsec:toleran-lb-proof}

We introduce some helpful notation and outline the high-level structure of the argument. 

\subsubsection{Setup and Outline of Argument}
\label{subsec:tolerant-indistinguishability-setup}

We introduce the following notation:

\begin{notation} \label{notation:flap-count}
	Given an outcome of the control subspace $C$ and of Nazarov's body $B=H_1 \cap \cdots \cap H_N \cap \Ball(\sqrt{n})\subset \R^{n+1}$ within $C$ as defined earlier, for $x \in \R^{n+1}$ we define the set $S_B(x)$ as 
	\[S_B(x) := \cbra{\ell\in [N] : x \in F_\ell}.\]
Note that if $x$ and $y$ have $x_C=y_C$, then $S_B(x)=S_B(y)$, i.e. only the $C$-part of $x$ affects $S_B$.
\end{notation}

We define the regions $\Top, \Middle, \Bottom \subset \R$ as follows:
\begin{align*}
\Top &:= \pbra{-\infty,\Phi^{-1}\pbra{{\frac {1-2c_2} 3}}},\\
\Middle &:= \pbra{\Phi^{-1}\pbra{{\frac {1 + c_2} 3}}, \Phi^{-1}\pbra{{\frac {2 - c_2} 3}}},\\
\Bottom &:= \pbra{\Phi^{-1}\pbra{{\frac {2+2c_2} 3}}, \infty}.
\end{align*}
Note that $\Top \sqcup \Middle \sqcup \Bottom \sqcup \curb = \R$ (where as before we identify $\R$ with an outcome of the one-dimensional action subspace $A$). 

To establish indistinguishability, we show that no non-adaptive 
  deterministic algorithm $\calA$ that makes
  $q=2^{c_3n^{1/4}}$ queries, for some sufficiently small constant $c_3$, 
  can distinguish $\Dyes$ from $\Dno$.
  Specifically, for any nonadaptive deterministic algorithm $\calA$ with query complexity $q $, we show that
\begin{equation} \label{eq:goal}
\Prx_{\bfyes\sim \Dyes}\big[\calA \text{ accepts }\bfyes\big]\leq \Prx_{\bfno\sim \Dno}\big[\calA  \text{ accepts }\bfno\big]+o(1). 
\end{equation} 
To this end, we define $\Bad$ to be the following event: 
\begin{flushleft}\begin{quote}
$\Bad$: There are $x,y \in \Shell_{n+1}$ queried
by $\calA$ that (i) satisfy $S_{\bB}(x) = S_{\bB}(y)=\{\ell\}$ for some $\ell \in [N]$ (or equivalently, $x,y\in U_\ell$ for some $\ell$),
and (ii) have $x_{\bA},y_{\bA}$ belonging to two \emph{distinct} sets among $\Top,\hspace{0.05cm}\Middle,\hspace{0.05cm}\Bottom.$
\end{quote}\end{flushleft}

We will first show in \Cref{lemma:bad-is-the-only-way-to-distinguish} that $\calA$ can distinguish $\Dyes$ from $\Dno$ only when $\Bad$ occurs.
On the other hand, in \Cref{lemma:bad-is-unlikely}, we show $\Bad$ occurs with probability $o(1)$ when the number of queries is $\smash{q=2^{c_3n^{1/4}}}$ and $c_3$ is sufficiently small. 
\Cref{lemma:bad-is-the-only-way-to-distinguish,lemma:bad-is-unlikely} together establishe~\Cref{eq:goal}; the proof of this is analogous to the proof of~Theorem~1 in Section 4.2 of~\cite{chen2024mildly} and we refer the reader to~\cite{chen2024mildly} for full details. 
Theorem~\ref{thm:tolerant} then follows from~\Cref{eq:goal} via Yao's minimax principle~(Theorem~\ref{thm:yao-minimax}).

\subsubsection{Indistinguishability of $\Dyes$ and $\Dno$}

We write $\calA(f)$ to denote the sequence of $q$ answers to the queries made by $\calA$ to $f$.  We write $\mathrm{view}_{\calA}(\Dyes)$ (respectively $\mathrm{view}_{\calA}(\Dno)$) to be the distribution of $\calA(\boldfyes)$ for $\boldfyes\sim\Dyes$ (respectively $\boldfno\sim\Dno$).
The following claim asserts that conditioned on $\Bad$ not happening, the distributions $\mathrm{view}_{\calA}(\Dyes|_{\overline{\Bad}})$ and $\mathrm{view}_{\calA}(\Dno|_{\overline{\Bad}})$ are identical.

\begin{lemma}\label{lemma:bad-is-the-only-way-to-distinguish}
	$\mathrm{view}_{\calA}(\Dyes|_{\overline{\Bad}})=\mathrm{view}_{\calA}(\Dno|_{\overline{\Bad}}).$
\end{lemma}
\begin{proof}
	Let $Q$ be the set of points queried by $\calA$. 
	Recall that the distributions of the subspaces $\bC$ and action variables $\bA$ are identical for $\Dyes$ and $\Dno$. So fix an arbitrary outcome of the $n$-dimensional subspace $C$ and the orthogonal one-dimensional subspace $A$. As the distribution of the Nazarov body $\bB\sim \mathrm{Naz}(r,N,C)$ is also identical for $\Dyes$ and $\Dno$, we fix an arbitrary outcome $B$ of $\bB$. Let $\boldf$ be a random function drawn from either $\Dyes$ or $\Dno$.

	Note that for any point $x\in\R^{n+1}$ such that $|S_B(x)|\neq 1$ or $x \notin \Shell_{n+1}$ or {$x_A  \in \curb$}, by construction we have that $\boldf(x)$ can be determined directly in the same way for both $\Dyes$ and $\Dno$ (no query is required). So it suffices for us to consider the points $x$ such that $|S_B(x)|=1$, $x \in \Shell_{n+1}$, and {$x_A\notin \curb$}. We call these points \emph{important} points.

    We divide these important points into disjoint groups according to $S_B(x)$. More precisely, for every $\ell\in [N]$, let $X_{\ell}=\{x \in \R^{n+1} \mid x \text{ is important},S_B(x)=\{\ell\}\}$. Let $\boldf_{\ell}(x)$ denote the function $\boldf(x)$ restricted to $X_{\ell}$ (where as stated above, $\boldf$ denotes either a function drawn from $\Dyes$ or from $\Dno$). 
The condition that $\Bad$ does not happen implies that \emph{either} $x_A \in \Top$ for all $x \in Q \cap X_{\ell}$, \emph{or}  $x_A \in \Middle$ for all $x \in Q \cap X_{\ell}$, \emph{or} $x_A \in \Bottom$ for all $x \in Q \cap X_{\ell}$.
In particular, this means $\boldf_{\ell}(x)=\boldf_{\ell}(y)$ for all $x,y\in Q \cap X_{\ell}$, and this holds for both $\Dyes$ and $\Dno$. 

	Since $\boldf_{\ell}(x)$ are the same for all $x\in Q\cap X_{\ell}$, the distribution of $\boldf_{\ell}$ is actually one random bit.
	Indeed, $\boldf_{\ell}(x)=0$ with probability $1/2$ and $\boldf_{\ell}(x)=1$ with probability $1/2$ (because each element $\ell \in [N]$ belongs to $\bP$ with probability $1/2$) independently, and this holds for both $\Dyes$ and $\Dno$. This completes the proof of the lemma. 
\end{proof}

Next, we show that  $\Bad$ happens with probability $o(1)$ 
(recall that $q=2^{c_3n^{1/4}}$). The proof of the following lemma follows the proof of an analogous lemma from~\cite{chen2024mildly}:

\begin{lemma}\label{lemma:bad-is-unlikely}
 For any fixed set of points $Q=\{x^1,\cdots,x^q\}\subset \R^{n+1}$, we have $\Pr[\Bad]=o(1)$.
\end{lemma}

\begin{proof}
	Fix a pair of query points $x,y \in \R^{n+1}$ that belong to $Q$.  By the definition of $\Bad$, we may assume without loss of generality that $x,y \in \Shell_{n+1}$. Let $\Bad_{x,y}$ be the event that 
	\begin{enumerate}
		\item[(a)] $x,y \in \bU_\ell$ for some $\ell \in [N]$ (equivalently, $S_{\bB}(x)=S_{\bB}(y)=\{\ell\}$), and
		\item[(b)] $x_{\bA},y_{\bA}$ belong to two \emph{distinct} sets among $\{\Top,\Middle,\Bottom\}$.
	\end{enumerate}
	Analogous to the argument in~\cite{chen2024mildly}, we will show that 
	\begin{equation} \label{eq:showme}
	\Prx_{\bB}[\Bad_{x,y}] \leq \min\{\Pr_{\bB}[(a)], \Pr_{\bB}[(b)]\} \text{~is very small}.
	\end{equation}
	Recall that each of the two intervals defining $\curb$ (cf.~\Cref{subsec:tolerant-indistinguishability-setup}) has the same width which we will denote $\rho(c_2)$ for succinctness, i.e. 
	\[
	\rho(c_2) := \Phi^{-1}\pbra{ {\frac {1+c_2} 3}} - \Phi^{-1}\pbra{{\frac {1-2c_2} 3}}.
\]

	On one hand, for (b) to happen on $x,y$, we must have
	\begin{equation} \label{eq:diamond}
		 \big|x_{\bA} - y_{\bA} \big| \geq \rho(c_2). \tag{$\diamond$}
	\end{equation}
	On the other hand, (a) means  
	\begin{equation} \label{eq:star}
		\text{There exists}~\ell\in[N]~\text{such that}~S_{\bB}(x)=S_{\bB}(y)=\{\ell\}. \tag{$\star$}
	\end{equation} 
	It follows that 
	$\Pr[\Bad_{xy}]\leq \min\{\Pr[\diamond],\Pr[\star]\}$. 
	We will show that $\min\{\Pr[\diamond],\Pr[\star]\}\leq 2^{-4c_3n^{1/4}}$ and will do so via the following lemmas, \Cref{lem:xyfar} and \Cref{lem:xynear} (below $c>0$ is a suitable positive absolute constant):

	\begin{lemma} \label{lem:xyfar}
	If $\|x-y\| \leq c n^{3/8}$, then $\Pr[\diamond] \leq 2^{-4c_3 n^{1/4}}.$
	\end{lemma}
	
	\begin{proof}
	Fix $x,y$ such that $\|x-y\| \leq cn^{3/8}$. For succinctness we write $z$ to denote $x-y$, so $z \in \R^{n+1}$ and $\|z\| \leq cn^{3/8}$; our goal is to show that $|z_{\bA}| \leq \rho(c_2)$ except with probability at most $2^{-4c_3 n^{1/4}}$.
	
	Since $\bA$ is a Haar-random direction in $\R^{n+1}$, the distribution of $z_{\bA}$ is the same as the distribution of $\|z\|\cdot\bv_1$ where $\bv \sim \S^{n-1}$. 
	Hence by standard bounds on spherical caps (\Cref{lem:spherical-cap}), 
	\[
		\Prx\sbra{|\bv_1| \geq \frac{t}{\sqrt{n}}} \leq e^{-t^2/2}.
	\]
	Taking $t= \sqrt{8c_3} \cdot n^{1/8}$, this probability is at most $e^{-4c_3 n^{1/4}} < 2^{-4c_3 n^{1/4}}$. So we set our threshold as 
	\[
		\|z\| \leq {\frac {\rho(c_2)}{2\sqrt{2c_3}}}\cdot n^{3/8},
	\]
	i.e.~we require that $c \leq {\frac {\rho(c_2)}{2\sqrt{2 c_3}}}$, and the lemma is established.
\end{proof}

	\begin{lemma} \label{lem:xynear}
	If $\|x-y\| > c n^{3/8}$, then $\Pr[\star] \leq 2^{-4c_3 n^{1/4}}.$
	\end{lemma}
	
	We defer the proof of~\Cref{lem:xynear} to~\Cref{subsec:tolerant-proof-xy-lemmas}. Thanks to \Cref{lem:xyfar,lem:xynear}, we get that
	\[
		\Pr[\Bad_{xy}] \leq  \min\cbra{\Pr[\diamond],\Pr[\star]} \leq 2^{-4c_3 n^{1/4}}.
	\]
	By a union bound over all (at most $2^{2c_3 n^{1/4}}$) pairs of points $x,y$ from $Q$, we get that
	\[
		\Pr[\Bad] \leq 2^{-4c_3 n^{1/4}} \cdot 2^{2c_3 n^{1/4}} = 2^{-2c_3 n^{1/4}} = o(1),
	\]
	which completes the proof.
\end{proof}


\subsubsection{Proof of~\Cref{lem:xynear}}
\label{subsec:tolerant-proof-xy-lemmas}

For the remainder of this section, we will always assume that $x,y\in\Shell_{n+1}$ satisfy $\|x-y\| > c n^{3/8}$. 
Note that we can view the construction $g_{\bC, \bB, \bP}$ as a two stage process:
\begin{itemize}
\item We first draw $\bC$, which is a Haar random $n$ dimensional subspace of $\mathbb{R}^{n+1}$. 
\item We then draw $\bB \sim \Naz(r,N,\bC)$, and we draw $\bP$ as a uniformly random subset of $[N]$. 
\end{itemize}
We require the following claim: 
\begin{claim}~\label{clm:goodxy}
	Suppose $x,y\in\Shell_{n+1}$ satisfy $\|x-y\| > c n^{3/8}$.  Then with probability at least $1 - 2^{-\Omega(n^{1/4})}$ over the outcomes of $\bC$, we have 	
	\begin{equation} \label{eq:way-station}
		\|x_{\bC}\| \geq \|x\| - 1,
		\quad 
		\|y_{\bC}\| \geq \|y\| - 1, 
		\quad\text{and}\quad 
		\|(x-y)_{\bC}\| \geq cn^{3/8} - 1.	
	\end{equation}
\end{claim}
\begin{proof}
	Fix $x,y \in \Shell_{n+1}$ such that $\|x-y\| > c n^{3/8}$. Because $\bA$ is drawn Haar-randomly (and since it defines $\bC$),  
	it follows from \Cref{lem:spherical-cap} that
	\[
		\Prx_{\bA}\left[\Vert x_{\bA} \Vert \ge t \cdot n^{-1/2} \cdot \Vert x \Vert\right] \le  e^{-t^2/2}. 
	\]
	Let $t = \beta n^{1/8}$ for a suitable constant $\beta>0$. The previous inequality gives
	\[
		\Prx_{\bA}\left[\Vert x_{\bA} \Vert \ge  \beta \cdot n^{-3/8} \cdot \Vert x \Vert\right] \le  e^{-\beta^2 n^{1/4}/2}.  
	\]
	Thus, with probability $1-e^{-\beta^2 n^{1/4}/2}$, we have 
	\begin{align*}
	\Vert x_{\bC} \Vert = \sqrt{\Vert x \Vert^2 - \Vert x_{\bA} \Vert^2} &= \Vert x \Vert \cdot \sqrt{ 1- \frac{\Vert x_{\bA} \Vert^2}{\Vert x \Vert^2}}\\
	&\ge \Vert x \Vert \cdot \bigg(1 -\frac{\Vert x_{\bA} \Vert^2}{2\Vert x \Vert^2} \bigg) \\
	&= \Vert x \Vert - \frac{\Vert x_{\bA} \Vert^2}{2 \Vert x\Vert}   \\
	&\ge \Vert x \Vert - \frac{\beta^2 n^{-3/4}  \Vert x \Vert }{2}. 
	\end{align*}
	As $x \in \Shell_{n+1}$ and thus $\Vert x \Vert \le 2 \sqrt{n}$, 
	it follows that the last expression is at least $\Vert x \Vert-1$. 
	Identical calculations yield the corresponding lower bounds on $\Vert y_{\bC} \Vert$ and $\Vert (x-y)_{\bC} \Vert$. 
\end{proof}

Fix an outcome $C$ of $\bC$ such that~\Cref{eq:way-station} holds. For convenience, we will write $x'$ for $x_{C}$, $y'$ for $y_{C}$, both of which lie in $\R^n$. For the rest of the argument, we will work over $C$, i.e.~we view sets such as $\bH_\ell,\bH'_\ell$  and $\bB$ as lying in $C$ (which we identify with $\R^n$) rather than in $\R^{n+1}$.

The following argument is analogous to (parts of) the proof of~Lemma~15 of~\cite{chen2024mildly}. Recall  
\[
	S_{\bB}(x') = \cbra{ \ell \in [N] : x' \in \bU_{\ell}}.
\]
By \Cref{clm:goodxy}, we have that
\begin{align*}
	\Prx\sbra{\star} 
	&\leq 2^{-\Omega(n^{1/4})} + \Prx_{\bB}\sbra{S_{\bB}(x')=S_{\bB}(y') = \{\ell\}\text{~for some~}\ell} \\  
	&= 2^{-\Omega(n^{1/4})} + \Prx_{\bB}\sbra{_{\bB}(x')=S_{\bB}(y')   ~\text{and}~ \exists\ell \text{~s.t.~}S_{\bB}(y')=\{\ell\} } \\
	&\leq 2^{-\Omega(n^{1/4})} + \Prx_{\bB}\sbra{S_{\bB}(x')=S_{\bB}(y')  ~\mid~\exists \ell \text{~s.t.~} S_{\bB}(y')=\{\ell\}},
\end{align*}
where $x',y'$ satisfy \Cref{eq:way-station}.
We will analyze the case when $\ell = 1$, which is without loss of generality since 
\[
	\Pr[\bX | \sqcup_\ell E_\ell] \leq \sup_\ell \Pr[\bX | E_\ell]
\]
for disjoint events $\cbra{E_\ell}$ and since the probabilities 
\[
	\Prx_{\bB}\left[S_{\bB}(x')=S_{\bB}(y')  \ \mid \ ~\exists\ell \text{~s.t.~} S_{\bB}(x')=\{\ell_1\}\right] = \Prx_{\bB}\left[S_{\bB}(x')=S_{\bB}(y')  \ \mid \ ~\exists\ell \text{~s.t.~} S_{\bB}(x')=\{\ell_2\}\right]
\]
for all $\ell_1, \ell_2 \in [N]$. 
So our goal is to upper bound 
\begin{equation} \label{eq:condit-goal}
\Pr_{\bB}\left[S_{\bB}(x')=S_{\bB}(y')  \ \mid \ S_{\bB}(y')=\{1\}\right].
\end{equation}
Observe that the event ``$S_{\bB}(y') = \{1\}$'' that we are conditioning on is an event over the random draw of $\bB$, i.e. over the draw of $\bg^1,\dots,\bg^N$. To analyze this event it is helpful to introduce the following notation: For $z' \in \R^n$, define 
\[
	\slab(z') := 
	\begin{cases}
		\cbra{g \in \R^n : g\cdot z' \geq r} & \|z'\| \leq \sqrt{n},\\ 
		\emptyset & \|z'\| > \sqrt{n}.
	\end{cases}
\]
Consequently, the event ``$S_{\bB}(y') = \{1\}$'' is the same as the event 
\[
	\cbra{\bg^1 \in \slab(y')} \wedge \cbra{\bg^i \notin \slab(y')~\text{for}~i\in\{2, \ldots, N\}}.
\]
We may fix any outcome $g^{2*}, \ldots, g^{N*}$ of $\bg^2,\dots,\bg^N$ all of which lie outside of $\slab(y')$, and we get (writing $\vec{\bg}$ to denote $(\bg^1,\dots,\bg^N)$) that
\begin{align}
	\text{\eqref{eq:condit-goal}} 
	& =\Prx_{\vec{\bg}}\sbra{S_{\bB}(x')=S_{\bB}(y') \ | \ (\bg^1 \in \slab(x')) \wedge \pbra{\bg^i \notin \slab(x') \text{~for~}i\in[2:N]}}\nonumber \\
&\leq\sup_{\substack{g^{i*} \notin \slab(y') \\ i \neq 1}}
\Prx_{\vec{\bg}}\sbra{S_{\bB}(x')=S_{\bB}(y') \ | \ (\bg^1 \in \slab(x')) \wedge \pbra{\bg^i = g^{i*}\text{~for~}i\in[2:N]}}\label{eq:decompose} \\
&\leq\sup_{\substack{g^{i*} \notin \slab(y') \\ i \neq 1}}
\Prx_{\vec{\bg}}\sbra{\bg^1 \in \slab(x') \cap \slab(y') \ | \ (\bg^1 \in \slab(x')) \wedge \pbra{\bg^i = g^{i*}\text{~for~}i\in[2:N]}} \label{eq:ubub} \\
&=
\Prx_{\bg \sim N(0,I_n)}\sbra{\bg \in \slab(x') \cap \slab(y') \ | \ \bg \in \slab(x')}, \label{eq:condit-goal2}
\end{align}
where \Cref{eq:decompose} uses $\Pr[\bX | \sqcup_\ell E_\ell] \leq \sup_\ell \Pr[\bX | E_\ell]$ as earlier; \Cref{eq:ubub} uses that if $\bg^1 \in \ring(y')$ then in order to have $S_{\bB}(x') = S_{\bB}(y')$ it must be the case that $\bg^1 \in \ring(x') \cap \ring(y')$; and \Cref{eq:condit-goal2} is because the event $\bg^1 \in \ring(x') \cap \ring(y')$ is independent of the outcome of $\bg^2,\dots,\bg^S$.
So in what follows our goal is to upper bound \eqref{eq:condit-goal2}.
In other words, recalling that we write $\Vol(K)$ to denote the Gaussian measure of the set $K$ (cf.~\Cref{subsec:gaussian-tail-bounds}), our goal is to obtain an upper bound on 
\begin{equation} \label{eq:ububub}
	\eqref{eq:condit-goal2} = {\frac {\Vol(\ring(x') \cap \ring(y'))}{\Vol(\ring(x'))}},	
\end{equation}
which is a two-dimensional problem because the only thing that matters about the outcome of $\bg \sim N(0,I_n)$ vis-a-vis \eqref{eq:ububub} is the projection of $\bg$ in the directions of $x'$ and $y'$.
Towards this goal, we recall the following tail bound for bivariate Gaussian random variables:

\begin{proposition}[Equation~(2.11) of~\cite{willink2005bounds}] \label{prop:bivariate-gaussian-tail}
	Suppose $(\bZ_1, \bZ_2) \sim N(0, \Sigma)$ where 
	\[
		\Sigma = 
		\begin{bmatrix}
			1 & \rho \\ \rho & 1
		\end{bmatrix} \qquad\text{for}~\rho > 0.
	\]
	Then for $h, k > 0$, we have 
	\[
		\Prx_{(\bZ_1, \bZ_2) \sim N(0, \Sigma)}\sbra{\bZ_1 > h, \bZ_2 > k} \leq \Phi(-h)\pbra{\Phi\pbra{\frac{\rho h - k}{\sqrt{1-\rho^2}}} + \rho e^{(h^2-k^2)/2}\Phi\pbra{\frac{\rho k - h}{\sqrt{1-\rho^2}}}}.
	\]
\end{proposition}

Let $\bg\sim N(0,I_n)$. Define the random variables
\[
	\bZ_1 := \frac{\bg\cdot x'}{\|x'\|}
	\qquad\text{and}\qquad
	\bZ_2 := \frac{\bg\cdot y'}{\|y'\|}
\]
and set $h := \frac{r}{\|x'\|}$, $k := \frac{r}{\|y'\|}$. 
It is immediate that 
\[\Vol(\slab(x')) = \Prx\sbra{\bZ_1 > h} \quad\text{and}\quad\Vol(\slab(x') \cap \slab(y')) = \Prx\sbra{\bZ_1 > h, \bZ_2 > k}.\]
Furthermore, note that $\Var[\bZ_1] = \Var[\bZ_2] = 1$. We also have 
$\rho := \Ex\left[\bZ_1\bZ_2\right] = \frac{x'\cdot y'}{\|x'\|\|y'\|} $. Thanks to~\Cref{clm:goodxy}, we have
\[
	(cn^{3/8} - 1)^2 \leq \|x' - y'\|^2 = \|x'\|^2 + \|y'\|^2 - 2x'\cdot y' \leq 2(n - x'\cdot y') 
\]
which in turn implies that 
\begin{equation} \label{eq:felton}
	\rho = \frac{x'\cdot y'}{\|x'\|\|y'\|} \leq \pbra{n - \frac{1}{2}(cn^{3/8} - 1)^2}\frac{1}{\|x'\|\|y'\|}.
\end{equation}
Using~\Cref{clm:goodxy} and the fact that $x, y \in \Shell_{n+1}$, we have that $\|x'\|, \|y'\| \geq \sqrt{n+1} -\sqrt{2\ln(2/\tau)} - 1$ and combining this with~\Cref{eq:felton} gives
\begin{equation} \label{eq:bigfoot}
	\rho \leq \frac{\pbra{n - \frac{1}{2}(cn^{3/8} - 1)^2}}{\pbra{\sqrt{n+1} -2\sqrt{2\ln(2/\tau)} - 1}^2} = 1 - \Omega(n^{-1/4}).
\end{equation}
Note that $\Vol(\slab(x')) = \Phi(-h)$. Consequently, using~\Cref{prop:bivariate-gaussian-tail} we get 
\[
	\eqref{eq:ububub} \leq \Phi\pbra{\frac{\rho h - k}{\sqrt{1-\rho^2}}} + \rho e^{(h^2-k^2)/2}\Phi\pbra{\frac{\rho k - h}{\sqrt{1-\rho^2}}},
\]
and we will obtain an upper bound on this in the remainder of this section. In particular, note that 
\begin{align*}
	\rho h - k &\leq \pbra{1 - \Omega(n^{-1/4})}\frac{r}{\|x'\|} - \frac{r}{\|y'\|} \\
	&= \frac{r}{\|x'\|}\pbra{1 - \Omega(n^{-1/4}) - \frac{\|x'\|}{\|y'\|}}
\end{align*}
Recall that $\|x'\|,\|y'\|\leq \sqrt{n}$ and that $\|x'\| \geq \sqrt{n+1} - \sqrt{2\ln(2/\tau)} - 1$. Hence, for $n$ large enough and an appropriate constant $\tau$, we have 
\begin{equation} \label{eq:LCD-soundsystem}
	\frac{\|x'\|}{\|y'\|} \geq 1 - {\Omega(n^{-1/2})}
\end{equation}
and consequently, we get that 
\begin{align*}
	\rho h - k \leq O\pbra{\frac{r}{\|x'\|\cdot n^{1/4}}\pbra{O\pbra{\frac{1}{n^{1/4}}} - \Theta(1)}} \leq O\pbra{\frac{-r}{\|x'\|\cdot n^{1/4}}} \leq -\Omega(1)
\end{align*}
for an appropriate choice of $\tau$. 
The final inequality relies on the above lower bound on $\|x'\|$ and~\Cref{lemma:r-N-relationship}. 
An identical calculation gives that $\rho k - h \leq -\Omega(1)$. 
It follows that 
\begin{align}
	\eqref{eq:ububub} &\leq 2\exp\pbra{\frac{r^2}{\|x'\|^2}\pbra{1 - \frac{\|x'\|^2}{\|y'\|^2}}}\Phi\pbra{\frac{-\Omega(1)}{\sqrt{1 - \rho^2}}} \nonumber \\
	&\leq 2\exp\pbra{o(1)}\Phi\pbra{\frac{-\Omega(1)}{\sqrt{1 - \rho^2}}} \label{eq:anthrax} \\
	&\leq 2\exp\pbra{o(1)}\Phi\pbra{-\Omega(n^{1/8})} \label{eq:daft-punk} \\
	&\leq 2^{-{\Theta}(n^{1/4})} \label{eq:FIN}
\end{align}
where the final inequality relies on a standard Gaussian tail bound (cf.~\Cref{prop:gaussian-tails}). To see~\Cref{eq:anthrax}, note that 
\begin{align*}
	\exp\pbra{\frac{r^2}{\|x'\|^2}\pbra{1 - \frac{\|x'\|^2}{\|y'\|^2}}} 
	& \leq \exp\pbra{\Theta\pbra{\frac{n^{3/2}\ln(1/c_1)}{\sqrt{n+1} - \sqrt{2\ln(2/\tau)} - 1}}\pbra{1 - \frac{\|x'\|^2}{\|y'\|^2}}} \\
	& \leq \exp\pbra{\Theta\pbra{\frac{n^{3/2}\ln(1/c_1)}{\sqrt{n+1} - \sqrt{2\ln(2/\tau)} - 1}}\cdot O(n^{-1})} \\ 
	& \leq \exp\pbra{O\pbra{\frac{1}{\sqrt{n}}}},
\end{align*}
where we used~\Cref{lemma:r-N-relationship} and~\Cref{eq:felton}. 
\Cref{eq:daft-punk} immediately follows from~\Cref{eq:bigfoot}. Finally, note that~\Cref{eq:FIN} completes the proof.

%% file: sections/two-sided-nonadaptive.tex
\section{Two-Sided Non-Adaptive Lower Bound}
\label{sec:two-sided-nonadaptive}

Our goal in this section is to prove Theorem~\ref{thm:two-sided} restated below:

\maintheoremone*

\subsection{Setup}

We recall some necessary tools and results from \cite{CDST15}.

\subsubsection{Distributions with Matching Moments}

The first results we need, stated below as \Cref{prop:matchmoments,prop:negsupport}, establish the existence of two finitely supported random variables that match the first $\ell$ moments of a univariate Gaussian, for any $\ell$. Crucially, one of the random variables is supported entirely on non-negative reals, while the other puts nonzero probability on negative values (so if $\ell$ is any fixed constant, it puts a constant amount of probability on negative values):

\begin{proposition}[\cite{CDST15} Proposition~3.1]
\label{prop:matchmoments}
Given an odd $\ell \in \N$, there exists a value $\mu=\mu(\ell)$ $> 0$
  and a real random variable $\bu$ such that\vspace{-0.04cm}
\begin{enumerate}
\item $\bu$ is supported on at most $\ell$ nonnegative real values; and\vspace{-0.12cm}
\item  $\E[\bu^k] = \E_{\bg \sim N(\mu,I_k)}[\bg^k]$ for all $k \in [\ell]$.\vspace{0.06cm}
\end{enumerate}
\end{proposition}

\begin{proposition}[\cite{CDST15} Proposition~3.2]
\label{prop:negsupport}
Given $\mu > 0$ and $\ell\in \N$, 
there exists a real random variable $\bv$ such that\vspace{-0.04cm}
\begin{enumerate}
\item $\bv$ is supported on at most $\ell+1$ real values, with $\Pr[\bv<0]>0$; and\vspace{-0.12cm}
\item $\E[\bv^k] = \E_{\bg \sim N(\mu,I_k)}[\bg^k]$ for all $k \in [\ell]$.\vspace{0.06cm}
\end{enumerate}
\end{proposition}

We will use $\bu$ (respectively $\bv$) to sample coefficients in our construction of the yes-distribution (respectively the no-distribution).

\subsubsection{Mollifiers, CLTs, Tail Bounds and Other Tools}

We recall the following basic proposition from \cite{CDST15} and its simple proof:

\begin{proposition} [\cite{CDST15} Proposition~4.1] \label{simplepro}
Let $\calA,\calA_{in}\sse \R^q$ where $\calA_{in} \sse \calA$. Let $\Psi_{in}: \R^q \to [0,1]$ be a function   satisfying $\Psi_{in}(X) = 1$ for all $X \in \calA_{in}$ and $\Psi_{in}(X) = 0$ for all $X \notin \calA$. Then for all random variables $\bS,\bT$:
\begin{equation*}
 \big|\hspace{-0.04cm}\Pr[\bS \in \calA]-\Pr[\bT\in \calA]
\hspace{0.01cm} \big| \le \big|\hspace{-0.04cm}\E[\Psi_{in}(\bS)]- \E[\Psi_{in}(\bT)]\hspace{0.01cm}\big|+
\max\hspace{-0.03cm}\big\{\hspace{-0.05cm}\Pr[\bS \in \calA\setminus \calA_{in}],\hspace{0.04cm}
\Pr[\bT\in \calA\setminus \calA_{in}]\big\}.\end{equation*}
\end{proposition}
\begin{proof}
Observe that $\Pr[\bS\in \calA] \ge \E[\Psi_{in}(\bS)]$  and $\Pr[\bS\in \calA] \le \E[\Psi_{in}(\bS)] + \Pr[\bS \in \calA\setminus \calA_{in}]$, and likewise for $\bT$. As a result, we have
\[\begin{aligned}
\Pr[\bS \in \calA]-\Pr[\bT\in \calA] &\le \E[\Psi_{in}(\bS)] + \Pr[\bS \in \calA\setminus \calA_{in}]-\E[\Psi_{in}(\bT)],\ \ \ \text{and}\\[0.3ex]
\Pr[\bS \in \calA]-\Pr[\bT\in \calA]  &\ge \E[\Psi_{in}(\bS)] -  \Pr[\bT \in \calA\setminus \calA_{in}]-\E[\Psi_{in}(\bT)].
\end{aligned}\]
Combining these, we have the proposition.
\end{proof}

We adopt the following notation:  for $J = (J_1,\ldots,J_q) \in \N^q$ a  $q$-dimensional multi-index, we let $|J|$ denote $J_1 + \cdots + J_q$ and let $J!$  denote $J_1!J_2! \cdots J_q!$.   We write $\#J$ to denote $|\{ i \in [q]\colon J_i \ne 0\}|$ (and we observe that $\#J \le |J|$).
Given $X \in \R^q$ we write $X^J$ to denote $\prod_{i=1}^q (X_i)^{J_i}$, and we write $X|_J \in \R^{\#J}$ to denote the projection of $X$ onto the coordinates for which $J_i \neq 0.$ For $f: \R^q \to \R$, we write $f^{(J)}$ to denote the $J$-th derivative,
i.e.
\[
f^{(J)} =
{\frac {\partial^{J_1 + \cdots + J_q} f}
{\partial x_1^{J_1} \cdots \partial x_q^{J_q}}}.
\]

We recall the standard multivariate Taylor expansion:

\begin{fact}[Multivariate Taylor expansion]\label{fact:taylor}
Given a smooth function $f:\R^q\to \R$ and $k\in \mathbb{N}$,
\[
f(X+\Delta)=\sum_{|J|\le k}\frac{f^{(J)}(X)}{J!}\cdot \Delta^J
+(k+1)\sum_{|J|=k+1}\left(\frac{\Delta^J}{J!}
\E\big[(1-\btau)^k f^{(J)}(X+\btau \Delta)\big]\right),
\]
for $X,\Delta\in \R^q$, where $\btau$ is uniform random over the interval $[0, 1]$.
\end{fact}

We recall the standard Berry--Esseen theorem for sums of independent real random variables (see for example, \cite{Fel68}), which is a quantitative form of the Central Limit Theorem:
\begin{theorem}[Berry--Esseen]
\label{thm:be}
Let $\bs = \bx_1 + \cdots + \bx_n$, where $\bx_1,\ldots,\bx_n$ are independent real-valued random variables with $\E[\bx_j] = \mu_j$ and $\Var[\bx_j] = \sigma_j^2$, and 
$\sum_{i=1}^n \E\sbra{|\bx_i|^3} \leq \kappa.$
Let 
$\bg$ denote a Gaussian random variable with mean $\sum_{j=1}^n \mu_j$ and variance $\sum_{j=1}^n\sigma_j^2$, matching those of $\bs$. Then for all $\theta\in \R$, we have
\[ \big| \Pr[\bs \le \theta] - \Pr[
\bg \le \theta] \big| \le \frac{O(\kappa)}{\sqrt{\sum_{j=1}^n \sigma_j^2}}.\]
\end{theorem}

For $\bg \sim N(0,I_n)$, the value $\sum_{i=1}^n \bg_i^2 $ is distributed according to a chi-squared distribution with $n$ degrees of freedom, denoted $\chi(n)^2$.  We recall the following tail bound:

\begin{lemma} [Tail bound for the chi-squared distribution, from \cite{Johnstone01}] \label{lem:johnstone}
Let $\bX \sim \chi(n)^2$. 
Then we have
\[\Pr\big[|\bX-n| \geq tn\big] \leq e^{-(3/16)nt^2},\quad\text{for all $t \in [0, 1/2)$.}\]
\end{lemma}

Following \cite{CDST15},
our proof will employ a carefully chosen ``mollifier,'' i.e. a particular smooth
function which approximates the indicator function of a set (the use of such mollifiers is standard in Lindeberg-type ``replacement method'' analyses).  We will use a specific mollifier, given in \cite{CDST15}, whose properties are tailored to our sets of interest (unions of orthants). The key properties of this mollifier are as follows:
\begin{proposition}[ \cite{CDST15} Proposition~4.3: ``product mollifier'']
\label{product-mollifier}
Let $\calO$ be a union of orthants in $\R^q$. 
For all $\eps > 0$, there exists a smooth function $\Psi_\calO:\R^q\to [0,1]$ with the following properties:
\begin{flushleft}\begin{enumerate}
\item $\Psi_\calO(X) = 0$ for all $X \notin \calO$.\vspace{-0.06cm}
\item $\Psi_\calO(X) = 1$ for all $X \in \calO$ with $\min_{i}\{ |X_i|\} \ge \eps$.\vspace{-0.08cm}
\item For any multi-index $J \in \mathbb{N}^q$ such that $|J|=k$, $\Vert
\Psi^{(J)}_{{\cal O}} \Vert_{\infty} \le \ff(k) \cdot (1/\eps)^k$, where $\ff(k)=k^{O(k)}$. \vspace{-0.12cm}
\item For any $J \in \mathbb{N}^q$,
$\Psi^{(J)}_{{\cal O}}(X) \not =0$ {only if} $X\in \calO$ and $|X_i| \le
\epsilon$ \vspace{0.02cm}for all $i$\vspace{-0.03cm} such that $J_i \ne 0$. Equivalently,
  $\Psi_{{\cal O}}^{(J)}(X) \ne 0$ {only if} $X\in \calO$ and $\| X|_J \|_\infty \le \eps$.
\end{enumerate}\end{flushleft}
\end{proposition}

\subsubsection{Clipping}

Given $C>0$, we define the ``clipping'' function $\clip_C: \R^n \to \zo$ which, on input a vector $x \in \R^n$, outputs 1 if and only if $\|x\| \leq \sqrt{n}+C$.

\subsection{The Yes- and No- Distributions} \label{sec:yesno}

Let $c>0$ (this is the $c$ of Theorem~\ref{thm:two-sided} ). Let $\bu$ and $\bv$ be the random variables given by \Cref{prop:matchmoments,prop:negsupport}, where we take $\ell$ to be the smallest odd integer that is at least $1/c$ and take $\mu=\mu(\ell)$.

A set $\bK$ drawn from our ``yes-distribution'' $\Dyes$ has indicator function defined as follows:
\begin{itemize}

\item
First, choose a Haar random orthonormal basis normalized so that each vector has Euclidean length $1/\sqrt{n}$, and denote those vectors $\ba^{(1)},\dots,\ba^{(n)}$.  (So $\ba^{(1)} \in \R^n$ is a Haar random unit vector in $\R^n$ scaled by $1/\sqrt{n}$; $\ba^{(2)}$ is Haar random over the radius-$(1/\sqrt{n})$ sphere in the $(n-1)$-dimensional subspace of $\R^n$ that is orthogonal to $\ba^{(1)}$; and so on.)

\item Then $n$ independent draws $\bu_1,\dots,\bu_n$ are made from the real random variable $\bu$ of \Cref{prop:matchmoments}.

\item The indicator function $\bK(x)$ is
\begin{equation} \label{eq:yes-distribution}
\bK(x) =
\Indicator\sbra{\bu_1(\ba^{(1)} \cdot x)^2 + \cdots + \bu_n(\ba^{(n)}\cdot x)^2 \leq  \mu \ \& \ \clip_C(x)=1}.
\end{equation}
(Here $C>0$ is a suitable constant, depending only on $c$ but not on $n$, that will be fixed later in our argument.)
\end{itemize}

A set $\bK$ drawn from our ``no-distribution'' $\Dno$ is defined very similarly, with the only difference being that $\bv$ takes the place of $\bu$:

\begin{itemize}

\item
The vectors $\ba^{(1)},\dots,\ba^{(n)}$ are chosen exactly as in the yes-case.

\item Then $n$ independent draws $\bv_1,\dots,\bv_n$ are made from the real random variable $\bv$ of \Cref{prop:negsupport}.

\item The indicator function $\bK(x)$ is
\begin{equation} \label{eq:no-distribution}
\bK(x) =
\Indicator\sbra{\bv_1(\ba^{(1)} \cdot x)^2 + \cdots + \bv_n(\ba^{(n)}\cdot x)^2 \leq  \mu \ \& \ \clip_C(x)=1}.
\end{equation}
(Here $C>0$ is the same constant as in the yes-case.)
\end{itemize}

We remark that our yes- and no- functions differ from the yes- and no- functions of \cite{CDST15} in a number of ways: Our functions involve a random orthonormal basis, they are degree-2 polynomial threshold functions rather than linear threshold functions, and they involve clipping. (In contrast the \cite{CDST15} functions do not involve choosing a random orthonormal basis, are LTFs, and do not incorporate any clipping.)

\subsubsection{Distance to Convexity}
We first consider yes-functions. Since $\bu$ is supported on non-negative real values and $\ba^{(1)},\dots,\ba^{(n)}$ are orthogonal vectors, every outcome of $\Indicator\sbra{\bu_1(\ba^{(1)} \cdot x)^2 + \cdots + \bu_n(\ba^{(n)}\cdot x)^2 \leq \mu}$ is an ellipsoid in $\R^n$. Since $\clip_C(x)$ is the indicator function of a ball in $\R^n$, and the intersection of a ball and an ellipsoid is a convex set, we immediately have the following:

\begin{corollary} \label{cor:yesconvex}
For every $C>0$, every $K \subset \R^n$ in the support of $\Dyes$ is convex.
\end{corollary}

The following lemma shows that a constant fraction of draws of $\bK \sim \Dno$ are constant-far from being convex (intuitively, this is because with extremely high probability a constant fraction of the coefficients $\bv_1,\dots,\bv_n$ are negative, which causes the degree-2 PTF to be far from an ellipsoid):

\begin{lemma} \label{lem:nofarfromconvex}
For a suitable choice of the constant $C>0$, with probability at least $1/2$ a random $\bK \sim \Dno$ is $\kappa$-far from convex (where $\kappa>0$ depends on $\mu$ and $\ell$ and hence only on $c$).
\end{lemma}

\begin{proof}
By the rotational symmetry of the $N(0,I_n)$ distribution, we may assume that the orthonormal basis $\ba^{(1)},\dots,\ba^{(n)}$ is the canonical basis $e_1,\dots,e_n$ scaled by $1/\sqrt{n}$.
Thus a draw of $\bK \sim \Dno$ (after a suitable rotation) is
\[
\bK(x) =
\Indicator\sbra{\bv_1x_1^2 + \cdots + \bv_nx_n^2 \leq n\mu \ \& \ \clip_C(x)=1}.
\]
Given this, it suffices to show that a random set 
\begin{equation} \label{eq:targetset}
\bK' := \Indicator\sbra{\bv_1x_1^2 + \cdots + \bv_nx_n^2 \leq n\mu}
\end{equation} is $2\kappa$-far from convex with probability at least $1/2$.  If we have this, then since $\bK$ has distance at most $\kappa$ from $\bK'$ (which holds for a suitable choice of the constant $C$, using \Cref{lem:johnstone}), the lemma follows.

To analyze \Cref{eq:targetset}, we begin by recalling that by \Cref{prop:negsupport}, the random variable $\bv$ has probability $p_i>0$ of taking value $d_i$ for $1 \leq i \leq \ell'$, where $\ell'$ is some value that is at most $\ell+1$, and we have $d_1<0$, $d_1 < d_2 < \cdots < d_{\ell'}$, and $p_1 + \cdots + p_{\ell'}=1$.  
Taking $k=1$ in item (2) of \Cref{prop:negsupport}, we have  
\begin{equation} \label{eq:mean-is-mu}p_1 d_1 + \cdots + p_{\ell'}d_{\ell'}=\mu. \end{equation}
For $i \in [\ell']$, let $\boldn_i$ denote the number of indices $j \in [n]$ such that $\bv_j = d_i$.  Since all of the values $p_1,\dots,p_{\ell'}$ are constants independent of the asymptotic parameter $n$, by a standard Chernoff bound and union bound, we have that for suitable constants $c_1,\dots,c_{\ell'}>0$ (which depend on the $p_i$'s),
\begin{equation} \label{eq:as-should-be}
\Prx_{\bv_1,\dots,\bv_n}\sbra{\boldn_i \in [p_i n - c_i \sqrt{n}, p_i n + c_i \sqrt{n}]\text{~for each $i \in [\ell']$}} \geq 1/2.
\end{equation}
Fix any outcome $(v_1,\dots,v_n)$ of $(\bv_1,\dots,\bv_n)$ such that the event on the LHS of \Cref{eq:as-should-be} is satisfied. 
In the rest of the proof we will argue that for such an outcome the set 
\begin{equation} \label{eq:newKprime}
K' = \Indicator\sbra{v_1x_1^2 + \cdots + v_nx_n^2 \leq n\mu}
\end{equation}
corresponding to \Cref{eq:targetset} is $\Omega(1)$-far from convex.

For each $i \in [\ell']$, let $S_i \subset [n]$ denote the set of indices $j \in [n]$ such that $v_j=d_i$.  Let $c'_i$ be such that $|S_i|=p_i n + c'_i\sqrt{n}$, and observe that $|c'_i| \leq c_i$. For ease of notation  we may suppose that $S_1$ consists of the first coordinates $\{1,\dots,p_1 n + c'_1 \sqrt{n}\}$ (this is without loss of generality by the rotational invariance of $N(0,I_n)$).

Fix any $i \in \{2,\dots,\ell'\}$ and consider the tuple of random Gaussian coordinates $(\bx_j)_{j \in S_i}$ for a draw of $\bx=(\bx_1,\dots,\bx_n) \sim N(0,I_n)$. We have
\[
\Ex_{\bx}\sbra{\sum_{j \in S_i} \bx_j^2}= p_i n + c'_i\sqrt{n},
\]
and by the Berry-Esseen theorem (Theorem~\ref{thm:be}), we get that
\begin{equation} \label{eq:band}
\Pr\sbra{\sum_{j \in S_i} \bx_j^2 \in \sbra{p_i n - A_i \sqrt{n},p_i n + A_i \sqrt{n}}} \geq 1 - {\frac 1 {10\ell'}}
\end{equation}
for suitable positive absolute constants $A_2,\dots,A_{\ell'}$ (depending on the $p_i$'s and the $c'_i$'s but not on $n$).

Let $A := \ell' \cdot \max\{|d_2|,\dots,|d_{\ell'}|\} \cdot \max\{A_2,\dots,A_{\ell'}\}.$. By a union bound applied to \Cref{eq:band} over all $i \in \{2,\dots,\ell'\}$, with probability at least $9/10$ we have that
\begin{equation} \label{eq:inrange}
\sum_{i=2}^{\ell'} \sum_{j \in S_i} d_i \bx_j^2 \in
\sbra{
\pbra{\sum_{i=2}^{\ell'} d_i p_i n} - A \sqrt{n},
\pbra{\sum_{i=2}^{\ell'} d_i p_i n} + A \sqrt{n}
};
\end{equation}
let us say that any such outcome of $(\bx_j)_{j \in S_2 \cup \cdots \cup S_{\ell'}}$ is \emph{good}.
Fix any good outcome $(x_j)_{j \in S_2 \cup \cdots \cup S_{\ell'}}$ of the last $n - (p_1 n + c'_1 \sqrt{n})$ coordinates of $\bx \sim N(0,I_n)$, and let $A' \in [-A,A]$ be the value such that the LHS of \Cref{eq:inrange} is equal to $\pbra{\sum_{i=2}^{\ell'} d_i p_i n} + A' \sqrt{n}.$
Recalling \Cref{eq:mean-is-mu}, for this good outcome of the last $n - (p_1 n + c'_1 \sqrt{n})$ coordinates, the set \eqref{eq:newKprime} (viewed as an indicator function of coordinates $1,\dots,p_1 n + c'_1 \sqrt{n}$) becomes
\begin{equation} \label{eq:newnewKprime}
\Indicator\sbra{\sum_{j=1}^{p_1 n + c'_1 \sqrt{n}} d_1 x_j^2 \leq p_1 d_1 n - A' \sqrt{n}}.
\end{equation}
Recalling that $d_1 < 0$, this is equivalent to
\begin{equation} \label{eq:almostthere}
\Indicator\sbra{\sum_{j=1}^{p_1 n + c'_1 \sqrt{n}} x_j^2 \geq p_1 n - (A'/d_1) \sqrt{n}}.
\end{equation}
Let $q$ denote the probability that \eqref{eq:almostthere} holds for independent standard Gaussians $\bx_1,\dots,\bx_{p_1 n + c'_1 \sqrt{n}}$; the Berry-Esseen theorem implies that $q$ is a constant in $(0,1)$ which is bounded away from both 0 and 1.
By the radial symmetry of the $N(0,1)^{p_1 n + c'_1 \sqrt{n}}$ distribution, it follows that the subset of $\R^{p_1 n + c'_1 \sqrt{n}}$ whose indicator function is given by \Cref{eq:almostthere} is $\Omega(1)$-far from convex, because it is $\Omega(1)$-far from convex on a ``line by line'' basis. In more detail, for each unit vector $v \in \R^{p_1 n + c'_1 \sqrt{n}}$, the function \eqref{eq:almostthere} labels points on the corresponding line $\{tv: t \in \R\}$ as follows:

\begin{itemize}

\item [(i)] if $|t| \geq \sqrt{p_1 n - (A'/d_1)\sqrt{n}}$ then \eqref{eq:almostthere} outputs 1 on $tv$;

\item [(ii)] if $|t| > \sqrt{p_1 n - (A'/d_1)\sqrt{n}}$ then \eqref{eq:almostthere} outputs 0 on $tv$.

\end{itemize}

Since this labeling corresponds to the complement of an interval, and since both (i) and (ii) have constant probability as explained above, the distance to convexity is $\Omega(1)$, and the proof of \Cref{lem:nofarfromconvex} is complete.
\end{proof}

\subsection{Proof of Theorem~\ref{thm:two-sided}}

As is usual for a non-adaptive lower bound, we use Yao's principle.
Let $\cal{X}$ be a $q \times n$ query matrix, so the $i$-th row ${\cal X}_{i\ast}=({\cal X}_{i1},\dots,{\cal X}_{in})$ is a vector in $\R^n$ corresponding to the $i$-th query made by some deterministic algorithm.
We will argue that the behavior of such a deterministic algorithm will be almost the same on a target function $\bK \sim \Dyes$ and on a target function $\bK \sim \Dno$.  

First, since our analysis will only consider target functions drawn from $\Dyes$ and $\Dno$, and any draw from either of these distributions always involves  clipping (the $\clip_C$ component of \Cref{eq:yes-distribution,eq:no-distribution}), we may suppose without loss of generality that each query vector ${\cal X}_{i\ast}$ has $\|{\cal X}_{i\ast}\| \leq \sqrt{n}+C$, i.e.~it satisfies $\clip_C({\cal X}_{i\ast})=1$.

Let $\Ryes$ be the $\zo^q$-valued random variable obtained by drawing $\bK \sim \Dyes$ (recall that this corresponds to drawing $\bu_1,\ba^{(1)},\dots,\bu_n,\ba^{(n)}$) and setting the $t$-th coordinate of $\Ryes$ to be
\[
\bK({\cal X}_{t*}) = \Indicator\sbra{\bu_1 (\ba^{(1)} \cdot {\cal X}_{t*})^2 + \cdots + \bu_n(\ba^{(n)} \cdot {\cal X}_{t*})^2 \leq \mu}.
\]
Similarly, let $\Rno$ be the $\zo^q$-valued random variable obtained by drawing $\bK \sim \Dno$ (recall that this corresponds to drawing $\bv_1,\ba^{(1)},\dots,\bv_n,\ba^{(n)}$) and setting the $t$-th coordinate of $\Rno$ to be
\[
\bK({\cal X}_{t*}) = \Indicator\sbra{\bv_1 (\ba^{(1)} \cdot {\cal X}_{t*})^2 + \cdots + \bv_n(\ba^{(n)} \cdot {\cal X}_{t*})^2 \leq \mu}.
\]
To prove a two-sided non-adaptive lower bound of $q$ queries, it suffices to show that for the $\Ryes,\Rno$ defined above, we have $\dtv(\Ryes,\Rno) = o(1).$

Let us write $\overline{a}$ to denote $\overline{a}=(a^{(1)},\dots,a^{(n)})$, and let us write $\Ryes^{\overline{a}}$ to denote the random variable $\Ryes$ conditioned on having the outcome of $\ba^{(1)},\dots,\ba^{(n)}$ come out equal to $\overline{a}$, and similarly for $\Rno^{\overline{a}}.$ 
Using the coupling interpretation of total variation distance and the natural coupling between $\Dyes$ and $\Dno$, we have that
\begin{equation} \label{eq:a}
\dtv(\Ryes,\Rno) \leq \Ex_{\overline{\ba} \sim \text{Haar}}\sbra{\dtv(\Ryes^{\overline{\ba}},\Rno^{\overline{\ba}})},
\end{equation} 
so it suffices to upper bound the RHS of \Cref{eq:a} by $o(1)$.

Let us say that an outcome $\overline{a}=(a^{(1)},\dots,a^{(n)}) \in (\R^n)^n$ is \emph{bad} if 
there is a pair $(t,j) \in [q] \times [n]$ such that $(a^{(j)} \cdot {\cal X}_{t*})^2 \geq {\frac {10 \ln n}{n}}$.
Recalling that each query vector ${\cal X}_{t*}$ has norm at most $\sqrt{n}+C$ and that each $\ba^{(j)}$ is a Haar random unit vector scaled by $1/\sqrt{n}$, it is easy to show that bad outcomes of $\overline{\ba}$ have very low probability:

\begin{lemma} \label{lem:badunlikely}
$\Pr[\overline{\ba}$ is bad$]=o(1)$.
\end{lemma}
\begin{proof}
Fix some pair $(t,j) \in [q] \times [n]$ and let $r \leq \sqrt{n}+C$ be the norm of the query vector ${\cal X}_{t \ast}$.
The distribution of $\ba^{(j)} \cdot {\cal X}_{t*}$ is precisely the distribution of the first coordinate of a Haar random point drawn from the $n$-dimensional sphere  of radius $r/\sqrt{n}$. Hence, writing $\bu \sim \mathbb{S}^{n-1}$ to denote a Haar random point from the $n$-dimensional unit sphere, we have
\begin{align*}
\Prx\sbra{(\ba^{(j)} \cdot {\cal X}_{t*})^2 \geq {\frac {10 \ln n}{n}}} &=
\Prx_{\bu \sim \mathbb{S}^{n-1}}\sbra{{\frac {|\bu_1|r}{\sqrt{n}}} \geq {\frac {\sqrt{10\ln n}}{\sqrt{n}}}}\\
&\leq
\Prx_{\bu \sim \mathbb{S}^{n-1}}\sbra{\bu_1 \geq \frac {\sqrt{10\ln n}}{r}}\\
&\leq \Prx_{\bu \sim \mathbb{S}^{n-1}}\sbra{\bu_1 \geq {\frac {3\sqrt{\ln n}}{\sqrt{n}}}} \tag{using $r \leq \sqrt{n}+C$}\\
&\leq e^{-(9/2)\ln n} = 1/n^{9/2},
\end{align*}
using a standard bound on spherical caps (see \Cref{lem:spherical-cap}).
Since there are only $qn<n^{5/4}$ many pairs $(t,j) \in [q] \times [n]$, a union bound concludes the proof.
\end{proof}
Fix $\overline{a}=(a^{(1)},\dots,a^{(n)})$ to be any non-bad outcome of $\overline{\ba}$.  
Recalling \Cref{eq:a}, by \Cref{lem:badunlikely} it suffices to show that $\dtv(\Ryes^{\overline{\ba}},\Rno^{\overline{\ba}})\leq o(1)$; this is our goal in the rest of the proof.

Let $\bS \in \R^q$ be the random column vector whose $t$-th entry is
\[
\bu_1 (a^{(1)} \cdot {\cal X}_{t*})^2 + \cdots + \bu_n(a^{(n)} \cdot {\cal X}_{t*})^2 - \mu,
\]
and let $\bT \in \R^q$ be the random column vector whose $t$-th entry is
\[
\bv_1 (a^{(1)} \cdot {\cal X}_{t*})^2 + \cdots + \bv_n(a^{(n)} \cdot {\cal X}_{t*})^2 - \mu.
\]
The response vector $\Ryes^{\overline{a}}$ is determined by the orthant of $\R^q$ in which $\bS$ lies and the response vector $\Rno^{\overline{a}}$ is determined by the orthant of $\R^q$ in which $\bT$ lies.
So to prove a $q$-query monotonicity testing lower bound for non-adaptive algorithms,
it suffices to upper bound
\begin{equation} \label{eq:duobound}
\duo(\bS,\bT) \leq o(1),
\end{equation}
where $\duo$ is the ``union-of-orthants'' distance:
\[ \duo(\bS,\bT)  := \max\Big\{ |\Pr[\bS \in \calO]-\Pr[\bT\in\calO]| \colon \text{$\calO$ is a union of orthants in $\R^q$\Big\}}. \]
In what follows, we will show that
$\duo(\bS,\bT)\le o(1)$ when $q=O(n^{1/4-c})$.
  To this end, let $\calO$ denote a union of orthants such that
\begin{equation}\label{eq:morning}
\duo(\bS,\bT)=\big|\hspace{-0.03cm}\Pr[\bS\in \calO]-\Pr[\bT\in \calO]\hspace{0.01cm}\big|.
\end{equation}
Following~\cite{Mos08,GOWZ10,CDST15}, we first use the Lindeberg replacement
  method to bound 
  \begin{equation} \label{eq:Ebd}
  \big|\hspace{-0.04cm}\E[\Psi_\calO(\bS)]-\E[\Psi_\calO(\bT)]\hspace{0.01cm}\big|,
  \end{equation}
  and then apply \Cref{simplepro} to bound (\ref{eq:morning}).

\medskip

For all $i\in \{0,1\ldots,n\}$ we introduce the $\R^q$-valued hybrid random variable $\bQ^{(i)}$ whose $t$-th coordinate is
\[
(\bQ^{(i)})_t = 
\sum_{j=1}^i \bv_j (a^{(j)} \cdot {\cal X}_{t*})^2 +
\sum_{j=i+1}^n \bu_j (a^{(j)} \cdot {\cal X}_{t*})^2.
\]
Observe that $\bQ^{(0)} = \bS$ and $\bQ^{(n)} = \bT$. Informally, we are considering a sequence of hybrid distributions between $\bS$ and $\bT$ obtained by swapping out each of the $\bu$-summands for a corresponding $\bv$-summand one by one. The main idea is to bound the difference in expectations
\begin{equation} \label{eq:boundme}
\big|\hspace{-0.03cm}\E [\Psi_{\calO} (\bQ^{(i-1)} )]- \E[
\Psi_{\calO} (\bQ^{(i)} )] \hspace{0.01cm}\big| \quad \text{for each $i$},
\end{equation}
since summing \eqref{eq:boundme} over all $i\in [n]$ gives an upper bound on
\[\big|\hspace{-0.03cm}\E [\Psi_{\calO} (\bS) ] - \E[\Psi_{\calO} (\bT)]\hspace{0.01cm}\big|
= \big|\hspace{-0.03cm}\E [\Psi_{\calO} (\bQ^{(0)} ) ] - \E[
\Psi_{\calO} (\bQ^{(n)} )] \hspace{0.01cm}\big|
\le \sum_{i=1}^{n} \big|\hspace{-0.03cm}\E [\Psi_{\calO} (\bQ^{(i-1)} )]- \E[
\Psi_{\calO} (\bQ^{(i)} )] \hspace{0.01cm}\big|
\] using the triangle inequality.

To bound \eqref{eq:boundme}, we define the $\R^q$-valued random variable $\mathbf{R}_{-i}$ whose $t$-th coordinate is
\begin{equation} (\mathbf{R}_{-i})_t = 
\sum_{j=1}^{i-1} \bv_j (a^{(j)} \cdot {\cal X}_{t*})^2 +
\sum_{j=i+1}^n \bu_j (a^{(j)} \cdot {\cal X}_{t*})^2 \label{eq:R-minus-i}.
\end{equation}
Writing $\bPhi(\bv_i,a^{(i)})$ to denote the random vector in $\R^q$
whose $t$-th coordinate is $\bv_i (a^{(i)} \cdot {\cal X}_{t*})^2$ and likewise for $\bPhi(\bu_i,a^{(i)})$, we have that
\[ \big|\hspace{-0.03cm}\E[\Psi_{\calO} (\bQ^{(i-1)} ) ] - \E[\Psi_{\calO} (\bQ^{(i)} )]
\hspace{0.01cm}\big|=\big|\hspace{-0.03cm} \E[\Psi_{\calO} (\mathbf{R}_{-i} + \bPhi(\bv_i,a^{(i)}))] - \E[\Psi_{\calO} (\mathbf{R}_{-i} + \bPhi(\bu_i,a^{(i)}))] \hspace{0.01cm}\big|.
\]
Truncating the Taylor expansion of $\Psi_\calO$ at the $\ell$-th term
(Fact~\ref{fact:taylor}), we get
\begin{equation}\begin{aligned}
\hspace{-0.4cm}\E\hspace{-0.05cm}\big[\Psi_{\calO} (\mathbf{R}_{-i} + \bPhi(\bv_i,a^{(i)}))\big]
 &= \sum_{|J| \le \ell} \frac{1}{J!} \cdot
  \E\left[\Psi_{\calO}^{(J)} (\mathbf{R}_{-i}) \cdot (\bPhi(\bv_i,a^{(i)}))^J
\right]
  \\
&\hspace{0.36cm}+ \sum_{|J|=\ell+1} \frac{\ell+1}{J!} \cdot
 \E\left[(1-\btau)^{\ell} \Psi_{\calO}^{(J)} (\mathbf{R}_{-i} + \btau  \bPhi(\bv_i,a^{(i)}))  (\bPhi(\bv_i,a^{(i)}))^{J}\right] \label{taylor-error}
\end{aligned}\end{equation}
where $\btau$ is a random variable uniformly distributed on the interval $[0,1]$ (so the very last expectation is with respect to $\btau$, $\bv_i$
and $\mathbf{R}_{-i}$).
Writing the analogous expression for $\E[\Psi_{\calO} (\bR_{-i} + \bPhi(\bv_i,a^{(i)}))]$,
we observe that by \Cref{prop:matchmoments,prop:negsupport} the first sums are equal term by term, i.e. we have
\[
\sum_{|J| \le \ell} \frac{1}{J!} \cdot
  \E\left[\Psi_{\calO}^{(J)} (\mathbf{R}_{-i}) \cdot  (\bPhi(\bv_i,a^{(i)}))^J
\right] = 
\sum_{|J| \le \ell} \frac{1}{J!} \cdot
  \E\left[\Psi_{\calO}^{(J)} (\mathbf{R}_{-i}) \cdot  (\bPhi(\bu_i,a^{(i)}))^J
\right]
\]
for each $|J| \leq h.$ Thus we may cancel
all but the last terms to obtain
\[
\big|\hspace{-0.03cm}\E[\Psi_{\calO} (\bQ^{(i-1)} )]
- \E[\Psi_{\calO} (\bQ^{(i)} )]\hspace{0.01cm} \big| \le \sum_{|J|=\ell+1} \frac{\ell+1}{J!} \Vert \Psi_{\calO}^{(J)} \Vert_{\infty} \left( \E\big[|(\bPhi(\bv_i,a^{(i)}))^J|\big]+\E\big[|(\bPhi(\bu_i,a^{(i)}))^{J}|\big]\right).
\]
Observe that there are $|\{ J \in \N^q \colon |J| = \ell+1\}| =
\Theta(q^{\ell+1})$ many terms in this sum.
Recalling that each value of $(a^{(j)} \cdot \calX_{t \ast})^2$ is at most ${\frac {10 \log n} n}$ (because $\bar{a}$ is not bad),
that both $\bu_i$ and $\bv_i$ are supported on at most $\ell+1$
real values that depend only on $\ell$ (by 
\Cref{prop:matchmoments,prop:negsupport}),
and \Cref{product-mollifier},
we have that for any $\tau>0$ (we will choose a value for $\tau$ soon),
\begin{equation}\label{taylor-error2}
\big|\hspace{-0.03cm}\E[\Psi_{\calO} (\bQ^{(i-1)} )] -
\E[\Psi_{\calO} (\bQ^{(i)} ) ]\hspace{0.01cm}\big| = O_\ell(1) \cdot  \left(\frac{q}{\tau}\right)^{\ell+1} \cdot  \left({\frac{10 \log n}{n}}\right)^{(\ell+1)/2}.
\end{equation}
Summing over all $i\in [n]$ costs us a factor of $n$ and so we get
\begin{equation} \label{eq:exp-bound}
\big|\hspace{-0.03cm}\E[\Psi_{\calO} (\bS)] -
\E[\Psi_{\calO} (\bT)]\hspace{0.01cm}\big| = O_\ell(1) \cdot  \left(\frac{q}{\tau}\right)^{\ell+1} \cdot  {\frac{(10 \log n)^{(\ell+1)/2}}{n^{(\ell-1)/2}}}.
\end{equation}

\Cref{eq:exp-bound} gives us the desired bound on \Cref{eq:Ebd}; it remains only to apply \Cref{simplepro} to finish the argument. To do this, let
\[\calB_\tau=\big\{X\in {\cal O}: \text{$|X_i|\le \tau$ for some $i\in [q]$}\hspace{0.01cm}\big\}\]
(${\cal B}_\tau$ corresponds to the region $\calA\setminus \calA_{in}$ of \Cref{simplepro}).
Since both $\bv$ and $\bu$ are supported on values of magnitude $O_\ell(1)$, using the one-dimensional Berry-Esseen inequality (Theorem~\ref{thm:be}) and a union bound across the $q$ coordinates we get that 
\begin{equation} \label{eq:prob-bound}
\Pr[ \bS \in \calB_\tau], \Pr[ \bT \in \calB_\tau] \leq O_{\ell}(q \tau) + O_{\ell}(q/\sqrt{n}).
\end{equation}
So by applying \Cref{simplepro}, we get that
\[
\duo(\bS,\bT) \leq O_{\ell}(q\hspace{0.02cm} \tau) + O_{\ell}(q/\sqrt{n}) + O_{\ell}(1)\cdot \left(\frac{q}{\tau}\right)^{\ell+1} \cdot  \frac{(10 \log n)^{(\ell+1)/2}}{n^{(\ell-1)/2}}. \]
Choosing $\tau = 1/n^{1/4}$ and recalling that $\ell$ is the smallest odd integer that is at least $1/c$, we get that for $q=O(n^{1/4-c})$ the RHS above is $O_{\ell}((10 \log n)^{(\ell+1)/2}n^{-c})$. This is $o(1)$ for any constants $c>0,\ell \in \N$, and the proof of Theorem~\ref{thm:two-sided} is complete.